\documentclass{article}
\usepackage{graphicx} 
\usepackage{amsfonts} 
\usepackage{helvet}         
\usepackage{courier}        
\usepackage{amsmath,bm}
\usepackage{amsthm}
\usepackage{amssymb}
\usepackage{subfigure}
\usepackage{blindtext}
\usepackage{apacite}
\usepackage{comment}
\usepackage{natbib}
\usepackage{type1cm}        
\usepackage[toc,page]{appendix}
\usepackage{makeidx}         
\usepackage{multicol}        
\usepackage{graphicx}      
\usepackage[bottom]{footmisc}
\usepackage{colortbl,booktabs}
\usepackage{tablefootnote}
\usepackage[dvipsnames]{xcolor}

\usepackage{enumitem} 
\usepackage{algorithm}
\usepackage{algpseudocode}
\usepackage{chngcntr}

\usepackage{booktabs}
\usepackage{tabularx}
\usepackage{adjustbox}
\usepackage{array}
\usepackage{threeparttable}
\usepackage{pifont}

\usepackage{hyperref}         
\usepackage{glossaries}
\newacronym{dp}{DP}{Dirichlet process}
\newacronym{degmfr}{DEG-MFR}{Dynamic Extended Generalized Many-Facet Rasch}
\newacronym{gmfr}{G-MFR}{Generalized Many-Facet Rasch}
\newacronym{egmfr}{EG-MFR}{Extended Generalized Many-Facet Rasch}
\newacronym{irt}{IRT}{Item Response Theory}
\newacronym{mcmc}{MCMC}{Monte Carlo Markov Chain}
\newacronym{mfr}{MFR}{Many-Facet Rasch}
\newacronym{pyp}{PYP}{Pitman-Yor process}
\newacronym{emfr}{E-MFR}{Extended Many-Facet}
\usepackage{bbm}
\usepackage{authblk}
\usepackage{tikz}
\usetikzlibrary{arrows.meta}

\newcommand{\indep}{\perp \!\!\! \perp}
\newcolumntype{M}[1]{>{\centering\arraybackslash}m{#1}}

\newtheorem{prop}{Proposition}

\usepackage{setspace}
\usepackage[top=2.5cm, bottom=2.5cm, left=0.7in, right=0.7in, includefoot]{geometry}

\title{Dynamic Extended Generalized Many-Facet Rasch Model: A Unified Approach to Rater-Mediated Assessment}

\author[1,2]{Giuseppe Mignemi} 
\author[2]{Patr{\'i}cia Martinkov{\'a}} 
\author[3]{Sally Paganin}
\affil[1]{Department of Human and Social Sciences, Kore University}
\affil[2]{Institute of Computer Science, Czech Academy of Science}
\affil[3]{Department of Statistics, The Ohio State University}

\date{}

\usepackage{multibib}
\makeatletter
\def\@mb@citenamelist{cite,citep,citet,citealp,citealt,citeauthor,citeyear}
\makeatother
\newcites{App}{Appendix References}

\begin{document}

\maketitle

\begin{abstract}
Rating procedures are crucial in psychological and medical diagnosis, educational assessment, peer review, and grant evaluation. A major methodological challenge is disentangling multiple sources of variability in ratings, such as subject ability, rater behavior, and item characteristics. Item Response Theory-based models, most notably many-facet Rasch models and hierarchical rater models, have been widely adopted for jointly estimating subject, rater, and item effects, while retaining high interpretability. We introduce the Dynamic Extended Generalized Many-Facet Rasch model, a general class of Bayesian semiparametric models that integrate and extend existing approaches. Its measurement model allows location and discrimination parameters to depend jointly on item and rater characteristics and incorporates subject--item and rater--subject variation. The dynamic extension introduces a rater-specific cross-lagged effect capturing dependence on the previously rated subject, and a Bayesian P-spline capturing systematic changes in discrimination over the rating sequence, accounting for learning, fatigue, or adaptation effects. Subject heterogeneity is addressed through Bayesian nonparametric mixtures, flexibly learning the latent ability distribution without strict parametric assumptions and enabling model-based subject clustering. We also derive a closed-form, model-based inter-rater consensus coefficient and study its characteristics. We discuss theoretical properties, develop posterior computation strategies, and evaluate the model in simulation studies and two real-data applications.
\end{abstract}

\section{Introduction}
Rating procedures are crucial in several applied scientific fields, such as psychological and medical diagnosis \citep{Nelson2015}, educational assessment \citep{mignemi2025BNP, Casabianca_2015}, grant review processes \citep{Erosheva2021, Martinkova2023}, and personnel selection \citep{Martinkova2018}.
They involve an observer, commonly called a rater (e.g., a psychologist, doctor, teacher, or selection committee member), assessing a subject's latent attribute (e.g., patient disease severity, student proficiency, or applicant quality) on a continuous or ordered rating scale. In some contexts, raters are asked to independently score each subject on a single item \citep{mignemi2025BNP, Mignemi_Chen_Moustaki_2025, Martinkova2023, Erosheva2021, Martinkova2018}, which reduces rater burden, but limits the amount of information available for modeling rater effects and measurement error \citep{martinkovabook}. For this reason, evaluations comprise multiple items or rubric dimensions designed to capture distinct aspects of the same latent attribute \citep{Uto_24,patz1999applications}. The association between this attribute and the observed scores depends on the characteristics of both the item and the rater \citep{Verhelst2001}. The psychometric properties of the items and the raters' variability can pose reliability concerns, meaning that ratings may not be consistent or stable across different raters or occasions, leading to uncertainty about the quality of ratings \citep{Martinkova2023}. Accounting for these rater-related sources of bias is essential to ensure fair score interpretation and use \citep{AERA2014}.

Consequently, several statistical methods have been proposed to estimate the latent attributes of the subjects under examination; we refer the reader to \cite{wang2019} for a comprehensive review of these methods. They aim to assess and improve rating accuracy by jointly estimating subject and rater features (e.g., subject latent trait and rater severity) and item characteristics (e.g., item difficulty) when multiple items are involved in the rating procedure. \gls{irt} provides an interpretable framework for simultaneously accounting for these sources of variability in ratings. 

The \gls{mfr} model \citep{linacre1989many} is one of the most widely used \gls{irt} models for rater-mediated assessments. It extends the Rasch model by simultaneously accounting for item difficulty and rater severity \citep{lunz1998measurement}. This method has gained widespread popularity and has been widely applied to performance assessment \citep[e.g.][]{polat_2022_multifacet, peng_2019_multifacet}. However, it relies on strong assumptions regarding both the items and the raters. Under the MFR model, all raters are assumed to be equally consistent, and the items share the same discrimination parameter; both assumptions may be unrealistic in real rating contexts \citep{Uto_2016}. 
While \acrfull{emfr} models \citep[e.g.,][]{Eckes2017, An_Extended_MFRM} capture 
raters' severity shifts and response style, they are still assumed to be equally consistent across items and subjects.
\acrfull{gmfr} \cite[e.g.,][]{Uto_2016, Uto_2020, Uto_2022} relax most of these assumptions by allowing raters to differ in consistency and items to differ in discrimination. These models have often been formulated in a Bayesian framework, with corresponding methods developed for posterior computation. 

While these methods offer several advantages, three limitations motivate the present work. First, commonly used \gls{gmfr} specifications assume a standard normal distribution for the latent attribute, although empirical latent distributions may exhibit skewness, multimodality, or other forms of heterogeneity \citep{Uto_24, mignemi2025BNP, Mignemi2024}. Second, rater-specific severity is typically modeled as constant across subjects, which may be unrealistic in the presence of differential rater bias \citep{Molenaar2021}. Such local biases may compromise comparability among subjects and affect the accuracy of latent-trait estimates.  
Third, rater-specific severity and consistency are assumed to be independent of the sequence in which subjects are rated. Rater behavior may change over the course of an assessment, and ratings of a current subject may depend on previously evaluated subjects \citep{Huang2025rater, Halo_IRT_2025raters}. For example, the score given to a subject might be influenced by the attribute level of the immediately preceding subject, since raters may implicitly compare them \citep{tversky1974judgment}. In rater-mediated assessments, each rater may exhibit an anchoring tendency, similar to that observed with self-report scales, which can affect accuracy \citep{Bolt2022anchoring}.

\subsection{Our Contribution}

We introduce a flexible class of Bayesian semiparametric models to overcome the limitations of existing approaches and provide an interpretable statistical framework for rater-mediated assessments. Our formulation explicitly distinguishes between the measurement model (MM), which describes how observed ratings arise from subject, item, and rater characteristics and their interactions, and the structural model (SM), which characterizes the distribution of the subjects' latent trait. In the SM, we account for subject heterogeneity using Bayesian nonparametric mixture priors based on the \gls{dp} and the more general \gls{pyp}, allowing the shape of the latent-trait distribution to be learned from the data. 
In the MM, we model the ordinal ratings through a Generalized Partial Credit response function \citep{Muraki_1992} with parameters that capture subject-item, rater-item, and rater-subject variation, as well as item- and rater-specific category effects and discrimination. The resulting \gls{egmfr} model integrates and extends components of existing \gls{gmfr}, Multiple Rater IRT \citep{Verhelst2001} and hierarchical rater formulations \citep{Uto_24, Uto_2020, Patz2002, Molenaar2021}. 

We then extend the \gls{egmfr} model to account for sequential dependencies and changes in rater behavior over the rating process. The 
\gls{degmfr} model introduces a rater-specific cross-lagged effect through which the local bias for the current subject may depend on the latent trait of the immediately preceding subject. With multiple raters, these dependencies induce a network in which the vertices represent the subjects' latent traits and edges represent rater-specific sequential dependencies. In addition, discrimination is allowed to change smoothly with rating position through a Bayesian P-spline \citep{heinzl2012,BDA3}, accommodating systematic patterns such as learning, adaptation, recalibration, or fatigue. Unlike models of rater severity drift \citep{Uto_2022}, the proposed formulation distinguishes sequential dependence on the previously evaluated subject from smooth changes in discrimination across the rating sequence.

Together, these extensions have three methodological consequences. First, separating measurement and structural components provides a modular framework in which additional covariates or distributional assumptions can be incorporated at the appropriate level; an explanatory extension is described in Appendix~\ref{sec:explanatory}. Second, the subject-item component allows uncertainty in the subject's latent trait to reflect information contributed jointly by items and raters, thus addressing the information-accumulation issue identified by \cite{mariano2002critique}. Third, the model yields rater- and subject-specific posterior summaries that can be used to investigate heterogeneity in severity, discrimination, consistency, category use, and sequential effects. The framework further provides a model-based approach to inter-rater consensus. We derive a closed-form expression for the pairwise correlation between raters' bias-adjusted latent evaluations of the same subject (Section~\ref{Agreement}). Unlike conventional agreement measures applied directly to observed ratings, this quantity summarizes agreement at the latent level after accounting for item characteristics and rater-specific sources of variation.

Posterior inference is implemented in \textit{nimble} (\textsc{R} package \texttt{nimble}, \citealp{nimble-software2026}). Because the dynamic model induces strong posterior dependencies through its cross-lagged structure, we develop customized MCMC updates to improve sampling efficiency. 
We illustrate the framework in an olfactory assessment with many raters, few items, and explicitly recorded stimulus order. A simulation study evaluating parameter recovery and the contribution of the different model components, and a second application to a clinical performance assessment with few raters and many items, are reported in Appendices~\ref{sim} and~\ref{application}, respectively.

\subsection{Paper Outline}
The remainder of the paper is organized as follows: Sections \ref{E-GMFRM} and \ref{Dynamic} introduce the \gls{egmfr} and \gls{degmfr} models, respectively. Section~\ref{Agreement} introduces the latent inter-rater consensus coefficient and studies its properties. Section \ref{posterior} describes the posterior inference. Section \ref{Odor} presents an application to olfactory perception. Section ~\ref{Discussion} discusses the main findings, limitations, and directions for future research. The simulation study (Appendix~\ref{sim}), an application to health-care professional assessment (Appendix~\ref{application}), and additional methodological details, extensions, and empirical results are provided in the Appendix and Supplementary Materials, available online through the link: \url{https://osf.io/rt8pc/overview?view_only=66b70abd01b642778051f2fdace8ea0c}

\section{Extended Generalized Many-Facet Rasch Model}\label{E-GMFRM}
In this section, we present the \acrfull{egmfr}. The first level of the model comprises the measurement model, in which the probability of observing a score is determined by the characteristics of the item, the features of the rater, and the latent trait of the subject through a response function. The distribution of the subjects' latent traits is specified in the structural model. Additional remarks are given at the end of the section, while identifiability is discussed in Appendix~\ref{identifiability}. 

\subsection{Rating setting}
Consider $N$ subjects who are independently assessed on $I$ items by $R$ raters. Each subject is rated by a subset of raters $\mathcal{R}_j \subseteq \mathcal{R}$ on the same set of items $i=1,\dots, I$. We let $\mathcal{J}_r = \{ j: r \in \mathcal{R}_j\}$ be the set of subjects scored by rater $r \in \mathcal{R}$. We consider the set of subjects and the set of raters to be disjoint, that is $\mathcal{J} \; \cap \; \mathcal{R} = \emptyset$. This implies that subjects do not act as raters, and raters are not themselves rated. 
Let $Y_{jir}$, for $j=1,\dots, N$, $r \in \mathcal{R}_j$ and $i=1,\dots, I$, be a categorical random variable recording the score assigned to subject $j$ by rater $r$ on item $i$. We assume that $Y_{jir} \in \{0, 1, \dots, K\}$, so that all items comprise the same range of $K+1$ ordered categorical values. We further assume that the items measure a unidimensional construct modeled by a subject-specific latent trait $\theta_j \in \mathbb{R}$.

 \subsection{Measurement Model}\label{Measurement:1}
We let the probability that subject $j=1,\dots, N$ receives a score $Y_{jir}=k \in \{0,1,\dots, K\}$ on item $i=1, \dots, I$ from rater $r \in \mathcal{R}_j$ jointly depend on the subject's latent trait (e.g., subject's proficiency, disease severity) and the interaction between the characteristics of the item and those of the rater. We assume that, for each rater, the probability of choosing the categorical score $k$ over $k-1$ is governed by a logistic adjacent-category response formulation \citep{masters1982rasch}:
\begin{equation}\label{eq:adiacent_logit}
\mathbb{P}[Y_{jir}=k|Y_{jir} \in \{k,k-1\},\theta_{j}, \phi_{ir}, \delta_{jirk}]= 
\frac{exp\{\phi_{ir}(\theta_{j} - \delta_{jirk})\}}{1+exp\{\phi_{ir}(\theta_{j} - \delta_{jirk})\}}  
\end{equation}
for $k=1,\dots, K;$ $j=1,\dots, N;$ $i=1,\dots, I$ and $r \in \mathcal{R}_j$. As for the Generalized Partial Credit Model (GPCM; \citealt{Muraki_1992}), this formulation implies that the marginal probability $\pi_{jirk}$ that subject $j$ receives a score $Y_{jir}=k$ to the item $i$ by the rater $r$ is given by:
\begin{equation}\label{eq:marg}
\mathbb{P}[Y_{jir}=k|\theta_{j}, \phi_{ir}, \boldsymbol{\delta}_{jir}]=\pi_{jirk}= 
\frac{exp \{\sum_{v=0}^{k}\phi_{ir}(\theta_{j} - \delta_{jirv})\}}{\sum_{h=0}^{K} exp\{\sum_{v=0}^{h}\phi_{ir}(\theta_{j} - \delta_{jirv})\}},
\end{equation}
with $\boldsymbol{\delta}_{jir}=[\delta_{jir1},\dots, \delta_{jirK}]$. The parameter $\delta_{jirk} \in \mathbb{R}$ is the location on the support of the subject latent trait $\theta$ where $\pi_{jirk}=\pi_{jir,k-1}$. This implies that, under the assumption that $\phi_{ir}>0$, for subject $j$ satisfying $\theta_{j}>\delta_{jirk}$, rater $r$ chooses category $k$ with a larger probability than category $k-1$ to item $i$; that is $\pi_{jirk}>\pi_{jir,k-1}$. On the contrary, for $\theta_{j}<\delta_{jirk}$, rater $r$ chooses category $k$ with a smaller probability than category $k-1$ to item $i$; that is $\pi_{jirk}<\pi_{jir,k-1}$.  
Since $\delta_{jirk}$ jointly depends on item- and rater-specific characteristics, we refer to it as the \textit{joint difficulty step parameter}.\\
Here $\phi_{ir}>0$ is a slope parameter and indicates the degree to which rater categorical scoring varies along the subject's latent trait scale. Larger values of $\phi_{ir}$ imply a better discrimination between different levels of latent traits by the rater $r$ on item $i$. In contrast, smaller values indicate poorer discrimination between different levels of the latent trait, as scored by rater $r$ on item $i$. This parameter depends jointly on the item discrimination and rater consistency, and we refer to $\phi_{ir}$ as the \textit{joint discrimination parameter}. 
The joint difficulty parameter $\delta_{jirk}$ is assumed to be independent of the joint discrimination parameter $\phi_{ir}$, and both are assumed to be independent of the subject's latent trait $\theta_{j}$. As a consequence, any dependence between the subject $j$ and the parameters $\delta_{jirk}$ captures the systematic bias of the rater $r$ in scoring $j$ and the item-specific residual for subject $j$ when assessed on item $i$. 

\subsubsection{Joint difficulty step parameter}\label{JDP:1} For each subject $j=1,\dots, N$, we assume the following decomposition of the difficulty of the rater $r \in \mathcal{R}_j$ to score $j$ with $Y_{jir}=k$ on item $i=1\dots, I$: 
\begin{equation}\label{Measurement:difficulty1}
    \delta_{jirk} = \epsilon_{ij} + \beta_{ir}  + \eta_{rj} + \tau_{rk} + \gamma_{ik},  \quad k=0,\dots,K. 
\end{equation}

\paragraph{Item-related parameters $\epsilon_{ij}$, $\beta_{ir}$ and $\gamma_{ik}$.}
Here $\epsilon_{ij}\in \mathbb{R}$ is the item-specific residual term for subject $j$ when scored on item $i$. It can also be interpreted as the marginal subject-specific difficulty of item $i$ for subject $j$ across raters. Given two subjects $j \neq j'$, with equal latent trait $\theta_j=\theta_{j'}$ and satisfying $\epsilon_{ij}<\epsilon_{ij'}$, conditional on the other parameters, $j$ has a larger probability of being scored into higher categories on item $i$ than $j'$. For each item $i=1,\dots,I$, we assume that residuals $\epsilon_{ij}$, are independent and identically distributed (i.i.d.) across subjects $j\in \mathcal{J}$, and follow a normal distribution with zero mean and variance $\varphi_i^2$:
\begin{equation}\label{item-shift}
   \epsilon_{ij}|\varphi_i^2 \overset{\mathrm{iid}}{\sim} N(
   0, \varphi_i^2), \quad j \in \mathcal{J}, \;\; i=1,\dots, I.
\end{equation}
Here $\varphi_i^2$ captures the variability of subject-specific difficulties or residuals of item $i$ across subjects. Larger values of this parameter imply larger differences in item difficulty across subjects, whereas smaller values correspond to smaller differences in their individual difficulty on item $i$.
Given two items $i$ and $i'$, satisfying $\varphi^2_i>\varphi^2_{i'}$, item $i'$ exhibits less subject-specific variation in difficulty. 
\\
The parameter $\beta_{ir}\in \mathbb{R}$ is the systematic bias of the rater $r$ on item $i$, that is, the average severity of the rater when scoring subjects on this item.  
For any pair of different raters $r \neq r'$, such that $\beta_{ir}<\beta_{ir'}$,  conditional on the other parameters, $r$ assigns higher scores to subjects on item $i$ with a larger probability than rater $r'$. We refer to $\beta_{ir}$ as the \textit{rater local difficulty shift} parameter for item $i$. Here, $\gamma_{ik}\in \mathbb{R}$ is the item-specific step-difficulty parameter that represents the difficulty of transition from score $k-1$ to score $k$ for item $i$.

\paragraph{
Rater parameters $\eta_{rj}$ and $\tau_{rk}$.}
The parameter $\eta_{rj}$ is the systematic bias of the rater $r$ in scoring subject $j$. Given two subjects $j, j'\in \mathcal{J}_r$, with $j \neq j'$, satisfying $\eta_{rj}<\eta_{rj'}$, the rater $r$ will systematically score subject $j$ in higher categories than subject $j'$ across the items. The parameter $\eta_{rj}$ can be interpreted as the subject-specific severity or \textit{subject local bias} of rater $r$ for subject $j$. We assume that the rater $r$'s subject local biases $\eta_{rj}$, $j \in \mathcal{J}_r$, are i.i.d. from a normal distribution with zero mean and variance $\omega^2_r$: 
\begin{equation}\label{Measurement:difficulty2}
    \eta_{rj}|\omega^2_r \overset{\mathrm{iid}}{\sim} N(0,\omega^2_r), \quad j \in \mathcal{J}_r, \;\; r \in \mathcal{R}.
\end{equation}
The parameter $\omega^2_r$ is the variance of the local or subject-specific biases of rater $r$ and captures the \textit{consistency} of the rater across different subjects. Smaller values of this parameter imply higher consistency levels of the rater across subjects, whereas larger values indicate an inconsistent severity level across subjects. \\
The parameter $\tau_{rk} \in \mathbb{R}$ is the threshold step parameter of the category $k$ for rater $r$, and might be interpreted as the average relative difficulty of category $k$ over $k-1$ for rater $r$. 

\subsubsection{Joint discrimination parameter}\label{jdp}
For each rater $r\in \mathcal{R}$, and each item $i=1,\dots, I$, we assume the following decomposition of the logarithm of the joint discrimination $\phi_{ir}$: 
\begin{equation}\label{Measurement:discrimination1}
    log(\phi_{ir}) = \lambda_i + \xi_{r}.
\end{equation}
Here $\lambda_i \in \mathbb{R}$ is the item-specific component of discrimination and captures the degree to which categorical scoring varies along the subject's latent trait scale, on average, across different raters for item $i$. Larger values of this parameter imply a higher discriminating power of the item. In comparison, smaller values imply that the item, on average, discriminates poorly between different levels of the latent trait.\\
The parameter $\xi_{r} \in \mathbb{R}$ is the rater-specific component of discrimination across different subjects and items. Larger values of this parameter imply higher rater discrimination levels; that is, the rater $r$ assigns similar categorical scores to subjects with similar latent traits. On the contrary, smaller values of $\xi_{r}$ indicate lower rater discrimination levels; that is, the rater $r$ assigns different categorical scores on the items to subjects with similar latent traits. We refer to $\xi_{r}$ as the \textit{rater discrimination shift} parameter.

\subsection{Structural Model}\label{Structural_EG}

The structural model specifies the population-level distributions of subjects' latent traits and raters' characteristics. Both groups can, in principle, be modeled as draws from independent populations, allowing information to be shared within each group separately and supporting population-level inference for both subjects and raters. However, this requires sufficient sample sizes in each group. In practice, in many rater-mediated assessment settings, such as educational or clinical evaluations, a small panel of raters (e.g., $R = 5$) evaluates a moderate number of subjects (e.g., $N=30$ in the clinical application reported in Appendix~\ref{application}) on many items (\textit{small-$R$, large-$I$} regime). In these cases, population-level inference for raters is difficult, and we recommend placing independent priors on rater-specific parameters. In contrast, settings such as odor perception studies involve many raters and relatively few items (\textit{large-$R$, small-$I$} regime), making population-level inference about raters feasible, whenever raters themselves are the target of inference. 

In the remainder of this section, we focus on the first regime, which is most common in rater-mediated assessments \citep{wang2019}. We place independent priors on rater-specific parameters and model the subjects' latent trait distribution nonparametrically. Rather than assuming a standard normal distribution for $\theta_j$, we place a nonparametric mixture prior that accommodates heterogeneity and potential multimodality, letting the data determine the shape of the latent trait distribution. The extension to a joint population model for raters is presented in Appendix \ref{A-A}.

\subsubsection{Modelling subjects' latent trait}\label{Structural:subjects}
We assume that subjects' latent traits $\theta_j$, $j=1,\dots, N$, are i.i.d. from a random probability measure $G$:
\begin{equation}\label{Structural_1}
    \theta_j|G \overset{\mathrm{iid}}{\sim} G, \; j=1,\dots,N,  
\end{equation}
where $G$ is defined as: 
\begin{equation}\label{BNP_1}
    G = \int \mathcal{K}\left(\theta_j| \theta^* \right) F\left(d \theta^* \right), \quad F \sim Q.
\end{equation}
Here $\mathcal{K}(\cdot|\psi)$ is a suitable kernel indexed by the set of parameters $\theta^*$, and $Q$ is a discrete nonparametric prior; that is, its realizations are almost surely (a.s.) discrete random measures. As a result, $G$ is a mixture distribution and Eq.~(\ref{Structural_1}) might be represented through the following hierarchical formulation:
\begin{equation}\label{G:hierarchical}
    \theta_j|\theta^*_j \overset{\mathrm{ind}}{\sim} \mathcal{K}\left(\theta^*_j \right), \quad \theta^*_j|F \overset{\mathrm{iid}}{\sim}F, \quad F \sim Q, \quad \quad j=1,\dots, N.
\end{equation}
Since the mixing distribution $F$ is a.s. discrete, it implies a positive probability that any two random subjects might share the same values of the kernel's parameters, that is $\mathbb{P}(\theta^*_j=\theta^*_{j'})>0$, for $j \neq j'$. A random sample $(\theta^*_1,\dots, \theta_N^*)$ from $F$ features $1 \leq M_N \leq N$ different unique values or \textit{atoms} $(\theta^*_1,\dots, \theta_{M_N}^*)$, and leads to a random partition of subjects $\{1,\dots, N\}$ into $M_N$ blocks, such that $\theta_j^* \in (\theta^*_1,\dots, \theta_{M_N}^*)$, for $j=1,\dots, N$. This hierarchical mixture specification naturally allows us to estimate the density of the subject's latent trait and relax the common normality assumption over this parameter. 

\paragraph{Dirichlet Process Mixtures.}
We let $Q$ be a \acrfull{dp} prior $Q=DP(\alpha, F_0)$, with concentration parameter $\alpha>0$ and base measure $F_0$.
The latter is the prior expectation of the \gls{dp}, $\mathbf{E}[F]=F_0$, and represents our \textit{prior guess} at $F$. The parameter $\alpha$ controls the spread of the prior for $F$ about $F_0$. In the limit of $\alpha \rightarrow \infty$, the probability mass is spread out, and $F$ gets closer to $F_0$; on the contrary, as $\alpha \rightarrow 0$, $F$ is less close to $F_0$ and concentrates at a point mass. \\
When $Q$ is a \gls{dp}, it induces an infinite mixture distribution over the subject latent traits, and Eq.~(\ref{G:hierarchical}) might be represented as follows:
\begin{equation}\label{Infinite_Mixture}
    G=\sum_{c \geq 1} \pi_c \mathcal{K}\left(\theta^{**}_c \right), \quad \pi_c = V_c \prod_{l < c} (1-V_l), \quad V_c|\alpha \overset{\mathrm{iid}}{\sim} Beta(1, \alpha), \quad \theta_c^{**} \overset{\mathrm{iid}}{\sim} F_0.
\end{equation}
Here $\boldsymbol{\pi}=\{\pi_c\}_{c\geq 1}$ is the vector of the components' weights that takes values in the infinite simplex; $\boldsymbol{\theta}^{**}=\{\theta^{**}_c\}_{c\geq 1}$ is the vector of infinitely many atoms i.i.d. from the non-atomic distribution $F_0$. The construction of the \gls{dp} given in Eq.~(\ref{Infinite_Mixture}) is referred to as \textit{stick-breaking representation} since the weights are generated from the so-called \textit{stick-breaking process} \citep{Sethuraman}. We hereafter consider $\mathcal{K}(\cdot)$ to be a normal kernel indexed by the pair $\theta_c^{**}=(\mu_c,\sigma_c^2)$, such that $\mathcal{K}(\theta_c^{**})=N(\mu_c,\sigma_c^2)$, where $\mu_c$ and $\sigma_c^2$ are the location and scale parameters, respectively, implying an infinite Gaussian mixture.

\paragraph{Pitman-Yor Process Mixtures.} 
For $Q=DP(\alpha, F_0)$, the number $M_N$ of different unique values $\theta_j^*$ generated by $F$ increases asymptotically at a logarithmic rate, with $M_N \sim \alpha \; log(N)$ a.s. for $N \rightarrow \infty$. Alternative discrete nonparametric priors might be specified over $F$ to overcome this issue and allow for greater flexibility in the prior distribution over the number of clusters \citep{mignemi2025BNP}. We refer to \cite{DeBlasi2015} and \cite{ Ghosal_van_der_Vaart_2017} for a broader treatment of this class of priors and the respective modeling advantages. \\
A widely used extension of the \gls{dp} is the \glsfirst{pyp} \citep{Pitman-Yor}, $PYP(d, \alpha, P_0)$, indexed by a discount parameter $0<d<1$, a concentration parameter $\alpha>-d$, and a base measure $P_0$. We may let $Q=PYP(d, \alpha, P_0)$, which might be represented through the stick-breaking representation displayed in Eq.~(\ref{Infinite_Mixture}). The only exception is that, for the \gls{pyp}, the sequence $\{V_c\}_{c\geq 1}$ is generated as follows: 
\begin{equation}\label{PY_v}
    V_c|\alpha,d \overset{\mathrm{iid}}{\sim} Beta(1-d, \alpha+cd). 
\end{equation}
Under this specification, the prior distribution of the mixture weights $\boldsymbol{\pi}$ depends on the concentration parameter $\alpha$, the discount parameter $d$, and the order $c$ in the sequence. As might be noted by Eq.~(\ref{PY_v}), the Pitman-Yor process implies a power-law behavior for the mixture weights $\pi$, with  $\mathbb{E}[\pi_{c}]=O(c^{-1/d})$, for $d \in (0,1)$, which suggests that the decay of the cluster sizes is governed by a power law \citep{Hjort_Holmes_Muller_Walker_2010}. For this reason, the discount parameter $d$ controls the prevalence and relative frequency of small clusters in the induced partition. See \cite{mignemi2025BNP} for other details on the interpretation of this class of priors in psychometric modeling.  

\paragraph{Conditional Moments.}
Leveraging the stick-breaking construction of the \gls{dp} and the \gls{pyp}, we can express the conditional moments of the latent trait $\theta$ as: 
\begin{equation}\label{moments}
    E(\theta|G)=\mu_\theta = \sum_{c \geq 1} \pi_c \; \mu_c, \quad \quad \quad Var(\theta|G)=\sigma^2_\theta = \sum_{c \geq 1}\pi_c (\sigma^2_c + \mu_c^2) -\mu_\theta^2, 
\end{equation}
where $\mu_\theta$ and $\sigma^2_\theta$ are the mean and the variance of the infinite Gaussian mixture $G$ induced by $F$, respectively. Dirichlet and \gls{pyp} mixtures with Gaussian kernels lead to posterior consistency in weak topology and, under standard assumptions, to consistent density estimation \citep{Ghosal_van_der_Vaart_2017}. As a by-product, under our model specification, the posterior moments of $G$ also converge to the true moments, by uniform integrability of the posterior distribution \citep{vanderVaart1998}. Model identifiability constraints for the \gls{egmfr} model and the dynamic version detailed in the next Section are discussed in Appendix \ref{identifiability}. 

\subsection{Relationship to Existing Rater Models}
The \gls{gmfr}  model, introduced in \cite{Uto_24}, is a special case of the \gls{egmfr} model detailed above. Specifically, the \gls{gmfr} is recovered if $\epsilon_{ij}=0$, $\forall\; i=1,\dots,I$, and $\forall\; j \in \mathcal{J}$; and if $\eta_{rj}=0$, $\forall\; r \in \mathcal{R}$, and $\forall\; j \in \mathcal{J}$. Moreover, under the \gls{gmfr}, the subject latent trait $\theta_j$ is assumed to be i.i.d. across subjects from a standard normal distribution. In contrast, the \gls{egmfr} places a nonparametric mixture prior on the latent trait distribution, allowing its shape to be informed by the data. 
We note that, by defining $\theta_{ji}=\theta_j-\epsilon_{ji}$, 
and $ \widetilde{\delta}_{jirk}
= \beta_{ir}  + \eta_{rj} + \tau_{rk} + \gamma_{ik}$, Eq.~(\ref{eq:marg}) can be rewritten as:
\begin{equation*}
\mathbb{P}[Y_{jir}=k|\theta_{ji}, \phi_{ir}, \widetilde{\boldsymbol{\delta}}_{jir}]=\pi_{jirk}= 
\frac{exp \{\sum_{v=0}^{k}\phi_{ir}(\theta_{ji} - \widetilde{\delta}_{jirv})\}}{\sum_{h=0}^{K} exp\{\sum_{v=0}^{h}\phi_{ir}(\theta_{ji} - \widetilde{\delta}_{jirv})\}}.
\end{equation*}
The parameter $\theta_{ji}$ can be interpreted as an item-specific latent score or the quality of response \citep{Verhelst2001}, connecting the proposed formulation to Hierarchical Rater Models with continuous ideal scores \citep[e.g.,][]{Molenaar2021, decarlo2011hierarchical, Patz2002} and Multiple Rater IRT models \citep{Verhelst2001, Mariano2007covariates}. As a result, the measurement accuracy of $\theta_j$ jointly depends on the number of items and raters, addressing the information-accumulation issue of standard \gls{mfr} models \citep{mariano2002critique}. A more detailed comparison of our proposal with most of the existing rater-mediated assessment models is given in Table~\ref{TAB_comp} in Appendix \ref{Additional_fig_tab}.

\section{Dynamic Extended Generalized MFR Model}\label{Dynamic}
According to the joint difficulty specification in Eq.~(\ref{Measurement:difficulty2}), for each rater, the local biases $\eta_{rj}$ are assumed to be i.i.d. and the rater discrimination shift $\xi_r$ is constant across all subjects in the subset $\mathcal{J}_r$. This implies that the rater consistency and local biases are invariant under any permutation of the subjects in $\mathcal{J}_r$ scored by rater $r$. That is, the scoring order does not affect the rater's severity or consistency during the rating procedure. This may be a strong assumption when subjects are scored sequentially by design, and the comparison between the current and previously scored subject may affect ratings. Moreover, discrimination between subjects with different latent traits may change as the rater progresses through the scoring sequence. \\
In this section, we detail the dynamic extension of the \gls{egmfr} model, the \gls{degmfr} model. When the scoring order of the subjects is known, we can embed this information into the model and account for the sequential dependencies in the scores given by each rater. More specifically, we relax the assumption of order invariance in the rater's scoring behavior and consider the subjects scored by each rater as an ordered sequence. 
The measurement and structural models presented for the \gls{egmfr} model are extended to account for the sequential effects during the rating procedure.

Building on this extension, we also introduce the Dynamic \textit{Latent Inter-Rater Consensus} (LIC) coefficient, $LIC_{rr'}\in (-1,1)$, defined as the correlation between the bias-adjusted latent evaluations that two raters $r$ and $r'$ provide for the same subject. Because these latent evaluations might depend on the position of the subject within the respective rating sequences, $LIC_{rr'}$ can vary along the rating process, capturing how sequential rater effects shape the latent consensus (or dissensus) between raters. A closed-form expression for $LIC_{rr'}$ and a discussion of its properties are provided in Section~\ref{Agreement}.

\subsection{Sequential Rating Setting}
\label{Long_setting}
We let $\mathcal{J}_r = \{ j: r \in \mathcal{R}_j\}$ be the set of subjects scored by rater $r \in \mathcal{R}$ and $J_r=|\mathcal{J}_r|$ be its cardinality.
For each rater $r$, let $b_r(t): \{1,\dots, J_r\} \rightarrow \mathcal{J}_r$ be a known bijection mapping the rater-specific index $t$ to the global identity $j \in \mathcal{J}_r$ of the $t$-th subject scored by $r$; the observation $Y_{ir}^{b_r(t)}$ indicates the score given by rater $r$ to subject $b_r(t)$ on item $i$. We note that this bijection is needed for each rater, since the same subject might be scored by different raters and occupies a different position in the respective rater-specific scoring orders. 
Figure \ref{Example_net} illustrates an example in which subject $C$ is the second subject scored by both raters 1 and 2, indicated respectively by solid and dashed lines; this implies the following bijections between the rater-specific orders and the subject identity: $b_1(2)=b_2(2)=C$, that is, for both raters the subject $C$ is the second subject in their respective rating orders. We note that $C$ is the first subject scored by rater 3, indicated by the dotted line. In this case, the bijection between the rater-specific order and the subject identity is $b_3(1)=C$. 
Throughout the remainder of the paper, we use the following notational conventions. We alternatively replace $b_r(t)$ with $t$ for the measurement parameters and with $j$ for the latent trait within each rater-specific scoring order. That is, we define $Y_{ir}^t:=Y_{ir}^{b_r(t)}$ and $\theta_j:=\theta^{b_r(t)}=\theta^t$, for each $r \in \mathcal{R}$. Indeed, we note that, for a given rater $r \in \mathcal{R}$, $b_r(t)$ uniquely identifies the $t$-th subject $j$ scored by $r$. This definition allows us to use $j$ to indicate that a parameter only depends on the identity of the subject and is invariant under any permutation of $\mathcal{J}_r$ (e.g., the latent trait $\theta_j$); on the contrary, we use $t$ to indicate that a parameter is not invariant to the order of the subjects. \\

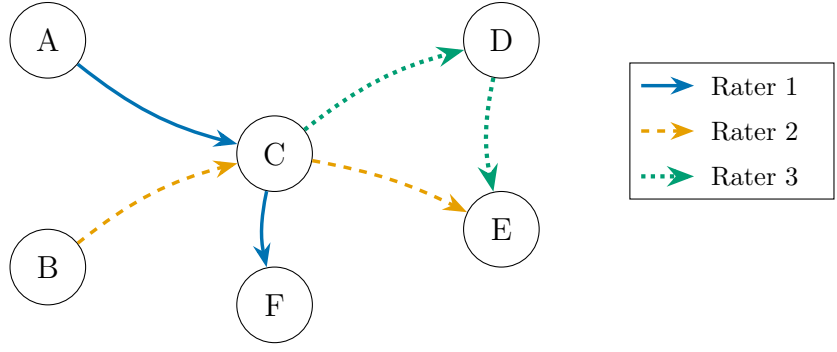
\begin{figure}[t]
\centering
\begin{tikzpicture}[
    node/.style={draw, circle, minimum size=1cm, font=\large},
    >=Stealth,
    rater1/.style={solid, line width=1.3pt, ->, color=okblue},
    rater2/.style={dashed, line width=1.3pt, ->, color=okorange},
    rater3/.style={dotted, line width=1.6pt, ->, color=okgreen}
    ]

\definecolor{okblue}{RGB}{0,114,178}
\definecolor{okorange}{RGB}{230,159,0}
\definecolor{okgreen}{RGB}{0,158,115}

\node[node] (A) at (-3, 1.5)  {A};
\node[node] (D) at (3, 1.5)   {D};
\node[node] (C) at (0, 0)     {C};
\node[node] (B) at (-3, -1.5) {B};
\node[node] (F) at (0, -2)    {F};
\node[node] (E) at (3, -1)    {E};

\draw[rater1] (A) to[bend right=12] (C);
\draw[rater1] (C) to[bend right=12] (F);

\draw[rater2] (B) to[bend left=12] (C);
\draw[rater2] (C) to[bend left=8] (E);

\draw[rater3] (C) to[bend left=12] (D);
\draw[rater3] (D) to[bend right=12] (E);

\begin{scope}[shift={(4.7,1.2)}]
  \draw[draw=black] (0,0) rectangle (2.7,-1.8);
  \draw[rater1] (0.15,-0.3) -- (0.85,-0.3);
  \node[anchor=west] at (0.95,-0.3) {Rater 1};
  \draw[rater2] (0.15,-0.9) -- (0.85,-0.9);
  \node[anchor=west] at (0.95,-0.9) {Rater 2};
  \draw[rater3] (0.15,-1.5) -- (0.85,-1.5);
  \node[anchor=west] at (0.95,-1.5) {Rater 3};
\end{scope}

\end{tikzpicture}
\caption{Graphical representation of the network implied by the order of examination of subjects A, B, C, D, E, and F by three different raters, indicated respectively by the solid, dashed, and dotted arrows. The arrows represent the rating order, connecting each subject to the subsequent one in the sequence.}
\label{Example_net}
\end{figure}

\subsection{Measurement Model}\label{measurement:dynamic}
For each rater $r\in \mathcal{R}$, we let the probability that subject $j \in \mathcal{J}_r$, who is the $t$-th rated subject by $r$, receives a score $Y_{ir}^{t}=k \in \{0,1,\dots, K\}$ on item $i=1, \dots, I$, jointly depend on the subject's latent trait and the interaction between the characteristics of the item and those of the rater. Consistent with the \gls{egmfr} model, we assume, for each rater, a Generalized Partial Credit response function to the items \citep{Muraki_1992} and define the marginal probability $\pi_{irk}^t$ that the subject $t$-th receives a score $Y_{ir}^t=k$ to the item $i$ by the rater $r$ as:
\begin{equation}\label{eq:marg_D}
\mathbb{P}[Y_{ir}^{t}=k|\theta_j, \phi_{ir}^{t}, \delta_{irk}^{t}]=\pi_{irk}^{t}= 
\frac{exp \{\sum_{v=0}^{k}\phi_{ir}^{t}(\theta_j - \delta_{irv}^{t})\}}{\sum_{h=0}^{K} exp\{\sum_{v=0}^{h}\phi_{ir}^{t}(\theta_j - \delta_{irv}^{t})\}}.
\end{equation}

The parameter $\delta_{irk}^{t} \in \mathbb{R}$ is the location on the support of the subject latent trait $\theta$ where $\pi_{irk}^{t}=\pi_{ir,k-1}^{t}$. This implies that, under the assumption that $\phi_{ir}^{t}>0$, for subject $j$ satisfying $\theta_{j}>\delta_{irk}^{t}$, rater $r$ chooses category $k$ with a larger probability than category $k-1$ on item $i$, given $j$ is the $t$-th subject examined by $r$; that is $\pi_{irk}^{t}>\pi_{ir,k-1}^{t}$. On the contrary, for $\theta_{j}<\delta_{irk}^{t}$, rater $r$ chooses category $k$ with a smaller probability than category $k-1$ on item $i$; that is $\pi_{irk}^{t}<\pi_{ir,k-1}^{t}$. 
We refer to $\delta_{irk}^{t}$ as the \textit{dynamic joint difficulty step parameter}.\\
Here $\phi_{ir}^{t}>0$ is a slope parameter and indicates the degree to which rater categorical scoring varies along the subject's latent trait scale. Larger values of $\phi_{ir}^{t}$ imply better discrimination between different levels of latent traits by the rater $r$ for item $i$ when rating the $t$-th subject. On the contrary, smaller values of this parameter indicate poorer discrimination between different levels of the latent trait for the $t$-th subject, as scored by rater $r$ on item $i$. We refer to $\phi_{ir}^{t}$ as the \textit{dynamic joint discrimination parameter}. 
The dynamic joint difficulty parameter $\delta_{irk}^{t}$ is assumed to be independent of the dynamic joint discrimination parameter $\phi_{ir}^{t}$, and both are assumed to be independent of the subject's latent trait $\theta_j$. As a consequence, any dependence between the identity and the order of subject $j$ and the measurement parameters $\delta_{irk}^{t}$ and $\phi_{ir}^{t}$ captures characteristics related to the subject $j$ which are independent of the latent trait $\theta_j$. Design features such as the order of the subjects in the rating procedure or the item-specific residual term might represent some of these dependencies. 

\subsubsection{Dynamic joint difficulty parameter}
 For each rater $r \in \mathcal{R}$, we assume the following decomposition of the difficulty of scoring the $t$-th subject with $Y_{ir}^{t}=k$ on item $i=1\dots, I$:
\begin{equation}\label{Measurement:Dynamic_difficulty1}
    \delta_{irk}^{t} = \epsilon_{ij} + \beta_{ir} + \eta_{r}^t + \tau_{rk} + \gamma_{ik},  \quad k=0,\dots,K. 
\end{equation}
Here $\eta_{r}^{t} \in \mathbb{R}$ indicates the systematic bias of the rater $r$ in scoring the $t$-th subject in the rater's scoring order. It jointly captures the subject-rater interaction---as the local biases in the \gls{egmfr} detailed in Section \ref{E-GMFRM}---and the sequential dependency across subjects. We refer to $\eta_{r}^{t}$ as the \textit{rater dynamic local bias} parameter. The distribution and the interpretation of the other parameters (i.e., $\epsilon_{ij}$, $\beta_{ir}$, $\tau_{rk}$, and $\gamma_{ik}$) are equivalent to those detailed for the \gls{egmfr} in Section \ref{JDP:1}.

\paragraph{Rater's Anchoring Effect.} 
For the first subject $t=1$ scored by rater $r\in \mathcal{R}$, we assume that the dynamic local bias $\eta_{r}^{t}$ follows a normal distribution with zero mean and variance $\omega_r^2$ as in Eq.~(\ref{Measurement:difficulty2}).
From the second to the last subject scored by rater $r$, that is $t =2,\dots, J_r$, we assume a dependence between the severity level of the rater for the current subject $t$ under examination and the latent trait of the immediately preceding scored subject $t-1$:
\begin{equation}\label{CL_path}
   \eta_{r}^{t}|\rho_r,\theta^{t-1},r \overset{\mathrm{ind}}{\sim} N(\rho_r\theta^{t-1},\omega^2_r),
\end{equation}
for $t = 2,\dots, J_r$, and $r \in \mathcal{R}$; where $\omega_r^2$ is the variance of the local or subject-specific biases of rater $r$
and $\theta^{t-1}$ is the latent trait of the subject $t-1$ rated immediately before the current subject by $r$. 
Here $\rho_r \in \mathbb{R}$ is the rater-specific cross-lagged effect of the previous subject's latent trait on the current local bias. 
It represents a rater-specific \textit{anchoring effect} in the scoring procedure, that is, the influence of the latent trait of the previous subject during the scoring of the current one. 
Positive values of $\rho_r$ imply that, on average, higher levels of the latent trait of the previously scored subject are associated with higher levels of rater severity for the next one. On the contrary, negative values of this parameter indicate that, on average, larger values of the latent trait of the previously rated subject are associated with lower rater severity levels for the current subject. In both instances, the raters' severity level is systematically associated with the previous subject's latent trait. We note that the dynamic latent regression in Eq.~(\ref{CL_path}) is not an autoregressive model, that is $\theta_j \indep \theta^{t-1}$ since $t$ indexes the rater-specific order of different subjects, rather than repeated time points for the same subject. 

\subsubsection{Dynamic joint discrimination parameter}
For each rater $r \in \mathcal{R}$ and each item $i=1,\dots, I$, we assume the following decomposition of the logarithm of the joint discrimination $\phi_{ir}^{t}$: 
\begin{equation}\label{Measurement:Dynamic_discrimination1}
    log(\phi_{ir}^{t}) = \lambda_i + \xi_{r} + \nu^{t}.
\end{equation}
The parameter $\nu^t \in \mathbb{R}$ represents the average change in log-discrimination associated with rating position $t$, across raters and items. 
Larger values of $\nu^t$ imply higher average discrimination at rating position $t$, while smaller values imply lower average discrimination (see Figure \ref{Example_discrimination}). The interpretation of the other parameters (i.e., $\lambda_i$ and $\xi_{r}$) is equivalent to those detailed for the \gls{egmfr} in Section \ref{jdp}.

\definecolor{okorange}{RGB}{230,159,0}
\definecolor{okgreen}{RGB}{0,158,115}
\definecolor{okskyblue}{RGB}{86,180,233}
\definecolor{okpurple}{RGB}{204,121,167}

\makeatletter
\pgfmathdeclarefunction{catp}{3}{%
  \begingroup
    \pgfmathsetmacro\sA{#2*(#1+1.5)}%
    \pgfmathsetmacro\sB{\sA+#2*(#1)}%
    \pgfmathsetmacro\sC{\sB+#2*(#1-1.5)}%
    \pgfmathsetmacro\mM{max(0,\sA,\sB,\sC)}%
    \pgfmathsetmacro\eZ{exp(max(-12,-\mM))}%
    \pgfmathsetmacro\eA{exp(max(-12,\sA-\mM))}%
    \pgfmathsetmacro\eB{exp(max(-12,\sB-\mM))}%
    \pgfmathsetmacro\eC{exp(max(-12,\sC-\mM))}%
    \pgfmathsetmacro\num{(#3==0)*\eZ+(#3==1)*\eA+(#3==2)*\eB+(#3==3)*\eC}%
    \pgfmathparse{\num/(\eZ+\eA+\eB+\eC)}%
    \pgfmath@smuggleone\pgfmathresult
  \endgroup}
\makeatother

\begin{figure}[t]
\centering
\begin{tikzpicture}[>=Stealth, font=\small]
\def\xs{0.38}   
\def\ph{3.2}    
\def\gap{4.3}   
\foreach \pos/\slope/\ord in {1/0.6/1st, 2/1.4/2nd, 3/2.4/3rd, 4/1.0/4th}{
  \begin{scope}[xshift={(\pos-1)*\gap cm}]
    \foreach \d in {-1.5,0,1.5}
      \draw[gray!50, densely dashed] ({(\d+4)*\xs},0) -- ({(\d+4)*\xs},\ph);
    \draw[okorange,  line width=1.2pt] plot[domain=-4:4, samples=101, smooth] ({(\x+4)*\xs}, {\ph*catp(\x,\slope,0)});
    \draw[okgreen,   line width=1.2pt] plot[domain=-4:4, samples=101, smooth] ({(\x+4)*\xs}, {\ph*catp(\x,\slope,1)});
    \draw[okskyblue, line width=1.2pt] plot[domain=-4:4, samples=101, smooth] ({(\x+4)*\xs}, {\ph*catp(\x,\slope,2)});
    \draw[okpurple,  line width=1.2pt] plot[domain=-4:4, samples=101, smooth] ({(\x+4)*\xs}, {\ph*catp(\x,\slope,3)});
    \draw[gray!70] (0,0) rectangle (8*\xs,\ph);
    \foreach \v in {-4,-2,0,2,4}
      \draw[gray!70] ({(\v+4)*\xs},0) -- ++(0,-0.07)
        node[below, font=\scriptsize, text=black] {$\v$};
    \node at (4*\xs,-0.65) {$\theta$};
    \node[above, align=center] at (4*\xs,\ph) {Subject $b_r(\pos)$\\[-3pt]{\footnotesize(\ord\ rated by $r$)}};
    \ifnum\pos=1
      \foreach \yy in {0,0.5,1}
        \draw[gray!70] (0,\yy*\ph) -- ++(-0.07,0)
          node[left, font=\scriptsize, text=black] {$\yy$};
    \fi
    \ifnum\pos<4
      \draw[->, very thick, gray!70] ({8*\xs+0.25},0.5*\ph) -- ({\gap-0.25},0.5*\ph);
    \fi
  \end{scope}
}
\end{tikzpicture}
\caption{\small Example of Dynamic Item-Rater Characteristic Curves for the same item $i$ and the same rater $r$ across the first four subjects $b_r(1),\dots,b_r(4)$ in the rating order of $r$ (arrows), as functions of the latent trait $\theta$. The orange, green, light blue, and pink lines display the category probabilities $\pi^{t}_{irk}$ in Eq.~(\ref{eq:marg_D}) for the four ordered categories $k=0,1,2,3$, respectively. The joint difficulty step parameters are held fixed (dashed vertical lines, where adjacent curves intersect), while the dynamic joint discrimination $\phi^{t}_{ir}$ increases from the first to the third subject (e.g., warm-up) and then decreases from the third to the fourth (e.g., fatigue) along the rating sequence, making the curves more peaked and then flatter.}
\label{Example_discrimination}
\end{figure}
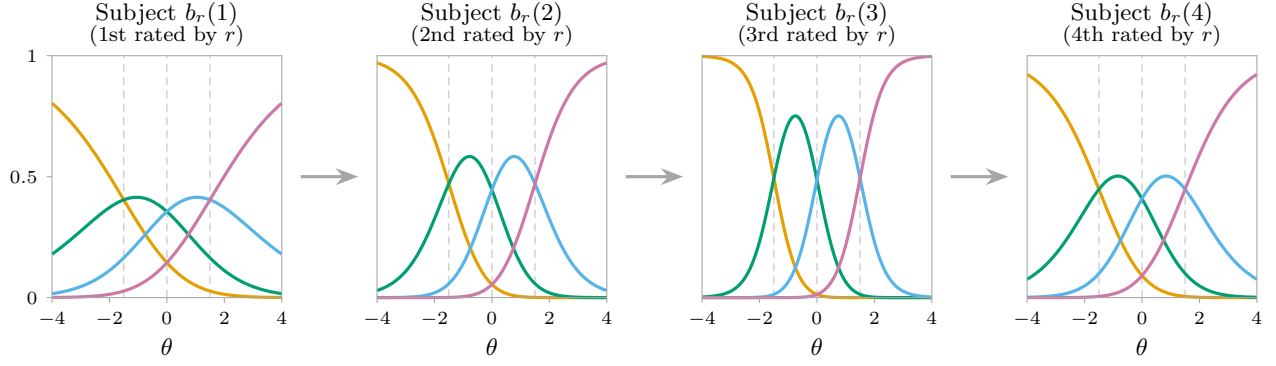

\paragraph{Rater's Practice Effect.}
For each rating position, we let $\nu^{t}$ capture the latent trend in average log-discrimination along the scoring procedure. We refer to this latent trajectory as the average \textit{rater practice effect}, specified as $f(t)$: 
\begin{equation}\label{practice_effect}
    \nu^t = f(t), \quad t=1,\dots, T,
\end{equation}
where $T=\max \limits_{r \in \mathcal{R}}J_r$ is the largest cardinality among the sets $\{\mathcal{J}_r\}_{r \in \mathcal{R}}$. We note that, in our proposal, the unbalanced design is treated as the general case, allowing each rater to score a different number of subjects. We model the latent trend $f(t)$ through a P-spline: 
\begin{equation}\label{P-Spline}
    f(t) = \sum_{d=1}^{D} a_{d} B_{d}^l(t),
\end{equation}
where $B_{d}^l(t)$ is the $d$-th B-spline basis function of degree $l$ defined for a grid of $m$ inner knots and $a_d \in \mathbb{R}$, $d=1,\dots,D$ with $D = m + l + 1$, are the corresponding coefficients. By regularizing the curvature of $f(t)$, this Penalized B-spline specification allows a flexible basis representation while avoiding overfitting. From an interpretative perspective, this specification allows us to identify different phases of the rating process. For instance, an initial warm-up phase might be associated with improving discrimination between subjects with different latent traits, followed by a period of stable engagement and, potentially, a final phase characterized by fatigue. Other trajectories are also possible, motivating the use of P-splines as a semiparametric smoothing technique. We use a cubic B-spline with $m = 6$ equidistant inner knots, yielding $D = m + 4 = 10$ basis functions. This choice reflects the minimum number of time points in the datasets considered ($T = 10$), so that the number of spline coefficients does not exceed the number of time points, while the smoothness is regulated by the random-walk prior. The specification can be tailored to the dataset under consideration. 
An example of a path diagram is displayed in Figure \ref{Example_path} in Appendix \ref{Additional_fig_tab}.

\section{Latent Inter-Rater Consensus}\label{Agreement}
Rater-specific features such as consistency, discrimination, response style, or anchoring tendency are jointly estimated under the \gls{degmfr} model, and raters can be characterized along these dimensions. This allows us to improve the accuracy of subjects' latent-trait estimates and to compare raters' individual latent patterns. Nevertheless, a summary measure quantifying the overall consensus or agreement between pairs of raters may also be useful, for example, to 
support the design of studies \citep{glas2024reducing} and inform guidelines for training programs \citep{Reliable_Teacher2020}.\\
In this section, we propose a pairwise latent inter-rater consensus measure that quantifies the degree of association between the global latent scores provided by two 
raters for the same subjects. More specifically, we measure the agreement between their model-implied, bias-adjusted evaluations of the same subject across multiple items. The resulting coefficient is therefore conceptually distinct from conventional inter-rater reliability measures computed directly from observed ratings. 

Inter-rater reliability is commonly formulated in terms of the relative contributions of between-subject and measurement-error variation to observed ratings. More flexible formulations allow these variance components, and consequently reliability itself, to vary across raters or contextual characteristics \citep{Martinkova2023}. Such reliability differences may have practical consequences when ratings are subsequently used for classification or selection decisions \citep{Bartos24}. The measure proposed here has a related signal-to-noise interpretation but addresses a different target: agreement between model-implied, bias-adjusted latent evaluations after accounting for item and rater characteristics.

\subsection{Global Latent Score Correlation}
Each rater $r \in \mathcal{R}$ provides an average or global evaluation of subject $j \in \mathcal{J}_r$, captured by $\theta_j - \eta_{rj}$, under the \gls{egmfr} model, or equivalently, $\theta_j - \eta_r^t$ under the \gls{degmfr} model. Given a pair of raters $r$ and $r'$ scoring the same random subject $j$, we define latent inter-rater consensus as the correlation between their respective global latent scores. We present the corresponding expressions under the \gls{egmfr} and \gls{degmfr} models. \\  
Under the general \gls{degmfr} model, we define  $\zeta_{jr}=\theta_j - \eta_{r}^t$ as the \textit{global latent score} for subject $j \in \mathcal{J}_r$ associated with rater $r \in \mathcal{R}$. As detailed in Section \ref{Long_setting}, 
$b_r(t)$ maps the rater-specific position $t$ to the global subject identity $j$. 
%
For a subject $j \in \mathcal{J}_r$, let $t_r(j) = b_r^{-1}(j)$ denote the position at which subject $j$ is evaluated by rater $r$. We can then write $\zeta_{jr}=\theta_j - \eta_{r}^{t_r(j)}$. Therefore, for two distinct raters $r \neq r' \in \mathcal{R}$ who evaluate the same subject $j \in \mathcal{J}_r \cap\mathcal{J}_{r'}$, 
$\zeta_{jr}=\theta_j - \eta_{r}^{t_r(j)}$, and $ \zeta_{jr'}=\theta_j - \eta_{r'}^{t_{r'}(j)}.$

\begin{prop}\label{prop_1}
Given two distinct raters $r \neq r' \in \mathcal{R}$ scoring the same subject $j \in \mathcal{J}$, the conditional correlation between the respective global latent scores $\zeta_{jr}$ and $\zeta_{jr'}$, under the \gls{degmfr}, is
\begin{equation}\label{ICC}
LIC_{rr'} = \operatorname{Corr}\left(\zeta_{jr},\zeta_{jr'} \mid G,\boldsymbol{\rho},\boldsymbol{\omega}^2\right) = \frac{\sigma^2_\theta\left(1  + \tilde\rho_r \, \tilde\rho_{r'}\, s_{rr'}(j)\right)}{\sqrt{\sigma^2_\theta(1+\tilde\rho^{\,2}_r)+\omega_r^2}\,\sqrt{\sigma^2_\theta(1+\tilde\rho^{\,2}_{r'})+\omega_{r'}^2}},
\end{equation}
where $\boldsymbol{\rho}=(\rho_r,\rho_{r'})$, $\boldsymbol{\omega}^2=(\omega^2_{r},\omega^2_{r'})$, $\tilde\rho_r = \rho_r\,\mathbbm{1}\{t_r(j)>1\}$, and $s_{rr'}(j)=\mathbbm{1}\{b_r(t_r(j)-1)=b_{r'}(t_{r'}(j)-1)\}$ indicates whether the two raters rated the same subject immediately before $j$.
\end{prop}
The proof is reported in Appendix \ref{A_P}. The \textit{Latent Inter-Rater Consensus} coefficient $LIC_{rr'}$ provides a model-implied measure of latent consensus between raters $r$ and $r'$ along the respective rating sequence. Under the \gls{egmfr} model, $\rho_r=0$ for all raters and Eq.~\eqref{ICC} reduces to the corresponding expression without sequential effects. 

When the two raters have different predecessors ($s_{rr'}(j)=0$), $LIC_{rr'}$ is non-increasing in $\tilde\rho^{\,2}_r$ and $\tilde\rho^{\,2}_{r'}$: it is largest when $j$ opens both rating sequences ($t_r(j)=t_{r'}(j)=1$) and smallest when $t_r(j),t_{r'}(j)>1$, the latter providing a conservative lower bound across rating positions. When they share the predecessor ($s_{rr'}(j)=1$), consensus is also affected by the covariance $\sigma^2_\theta\rho_r\rho_{r'}$ induced by the common anchor. Unlike reliability coefficients based on the intra-class correlation, which are bounded in $[0,1]$ by construction and therefore cannot represent systematic disagreement \citep[e.g.][]{mignemi2025BNP, Martinkova2023, tenHove2021, Joreskog_74}, $LIC_{rr'}$ takes values in $(-1,1)$. A negative value can arise only when the two raters share the preceding subject and have opposite anchoring tendencies strong enough to outweigh the common signal $\theta_j$. The index thus distinguishes latent consensus ($LIC_{rr'}>0$), absence of linear association ($LIC_{rr'}\approx 0$), and sequence-induced latent dissensus ($LIC_{rr'}<0$). 
Because $LIC_{rr'}$ is defined conditionally on $G$ and the rater-specific parameters, it can also be computed for raters with no subjects in common ($\mathcal{J}_r\cap\mathcal{J}_{r'}=\emptyset$), in which case it is a model-implied prediction rather than an estimate from overlapping ratings.

The signal-to-noise form of 
Eq.~\eqref{ICC} for $s_{rr'}=0$ relates the $LIC_{rr'}$ to 
multilevel reliability measures for latent scores within the \gls{irt} 
\citep{ChoShenNaveiras2019} and to variance-component formulations of inter-rater reliability 
\citep[e.g.,][]{mignemi2025BNP,lin2025reliability,Martinkova2023,tenHove2021,Joreskog_74}. 
Unlike conventional agreement and reliability coefficients, which provide only group-level information \citep{Eckes2017}, $LIC_{rr'}$ is informative at the level of individual rater pairs. However, $LIC_{rr'}$ should be interpreted specifically as latent inter-rater consensus under the fitted model rather than as the reliability of the observed ratings.

\section{Posterior Inference}\label{posterior}

Posterior inference is performed via Markov chain Monte Carlo (MCMC) with a custom algorithm implemented in NIMBLE \citep{nimble-software2026}. Priors and samplers for all parameter blocks are summarized in Table~\ref{tab:samplers}; the full prior specification and posterior computation strategy are detailed in Appendix~\ref{Posterior_App}, and the complete code is shared in the Supplementary Materials. Normal distributions are parameterized by mean and variance, Gamma by shape and rate, and inverse-Gamma by shape and scale. The P-spline in Eq.~(\ref{P-Spline}) uses cubic B-splines with $m=6$ equidistant inner knots, and $\boldsymbol{\Omega}=\boldsymbol{D}^T\boldsymbol{D}$ is the second-order difference penalty matrix \citep{heinzl2012}. The discount parameter of the \gls{pyp} is fixed at $d\in\{0.3,0.6\}$, with $d=0$ corresponding to the \gls{dp}.

The mixture prior is approximated by a finite stick-breaking construction truncated at $C=25$ components, with $V_C=1$ \citep{Ishwaran2001}. Conjugacy is exploited for the component parameters $(\mu_c,\sigma^2_c)$; labels are drawn from their exact categorical full conditionals, stick-breaking variables are updated by Beta-categorical Gibbs steps, and the remaining blocks use adaptive random-walk Metropolis--Hastings (RWMH), with variance parameters updated on the log scale (Table~\ref{tab:samplers}). Because the cross-lagged links (Figure~\ref{Example_net}) induce strong posterior dependence among subjects rated consecutively, the latent traits $\theta_j$ are updated with multivariate Barker proposals \citep{LivingstoneZanella2022} on \textit{network-informed} blocks, complemented by a second layer of randomly generated blocks to improve global exploration (Appendix~\ref{sec:posterior_computation}). To resolve the location indeterminacy among $\lambda_i$, $\xi_r$ and $\nu^t$ in Eq.~(\ref{Measurement:Dynamic_discrimination1}), we adopt a parameter-expanded weighted-average centering of the P-spline (Appendix~\ref{sec:posterior_computation}).

For all simulated and real-data analyses, each chain is run for 55,000 iterations, the first 5,000 are discarded as burn-in, and the remaining draws are thinned by a factor of 10, yielding 5,000 retained draws. Convergence and mixing are assessed by trace plots, effective sample sizes from the \textsc{R} package \texttt{coda} \citep{CODA}, and multiple runs from different starting values (Section~4 of the Supplementary Materials). Parameters are summarized by their posterior means (EAP) and 95\% quantile-based credible intervals; $LIC_{rr'}$ (Eq.~(\ref{ICC})) is evaluated at each retained iteration, yielding its full posterior. The latent trait density is monitored over a grid of equally spaced points. Subject clustering is summarized by the partition minimizing the posterior expected variation of information loss \citep{Meila_2007,Wade&Ghahramani_2018}, computed from the posterior similarity matrix with the package \texttt{salso} \citep{Dahl_salso}. Average run times for the \gls{gmfr}, \gls{egmfr}, and \gls{degmfr} models were 25, 35, and 43 minutes, respectively, on a 12th Gen Intel Core i9-12900H CPU.

\begin{table}[t]
\centering
\small
\begin{adjustbox}{max width=\textwidth}
\begin{tabular}{l p{4.3cm} l l}
\hline
Parameter & Description & Prior & Sampler \\
\hline
$\theta_j$ & Subject latent trait & DP/PYP mixture & Barker block (network-informed + random) \\
$z_j$ & Cluster membership indicator & Categorical$(\boldsymbol{\pi})$ & Exact categorical \\
$v_c$ & Stick-breaking variable & Beta$(1-d,\,\alpha+cd)$ & Beta-Categorical \\
$\alpha$ & Concentration parameter & Ga$(2,4)$ & Adaptive RWMH \\
$\mu_c$ & Mixture component mean & $N(0,2)$ & Conjugate normal \\
$\sigma^2_c$ & Mixture component variance & IGa$(2.01,1.01)$ & Conjugate inverse-Gamma \\
$\boldsymbol{a}$, $b$ & P-spline basis coefficients; smoothing precision & $N(\mathbf{0},(b\boldsymbol{\Omega})^{-1})$, Ga$(2,0.05)$ & Adaptive RWMH \\
$\lambda_i$, $\xi_r$ & Item log-discrimination; rater discrimination shift & $N(0,2)$ & Adaptive RWMH \\
$\omega_r^2$ & Rater consistency (variance of subject local biases) & $\log\omega_r^2 \sim N(-0.5,2)$ & Adaptive RWMH (log scale) \\
$\varphi_i^2$ & Item residual variance & $\log\varphi_i^2 \sim N(-0.5,2)$ & Adaptive RWMH (log scale) \\
$\beta_{ir}$, $\rho_r$ & Rater local difficulty shift; rater anchoring effect & $N(0,1)$ & Adaptive RWMH \\
$\eta_r^t$ & Rater dynamic local bias & $N(\tilde\rho_r\theta^{t-1},\omega_r^2)$ & Adaptive RWMH \\
$\epsilon_{ij}$ & Item-specific residual & $N(0,\varphi_i^2)$ & Adaptive RWMH \\
$\gamma_{ik}$ & Item step difficulty & $N(0,1.5)$ & Adaptive RWMH \\
$\tau_{rk}$ & Rater threshold step & $N(0,0.5)$ & Adaptive RWMH \\
\hline
\end{tabular}
\end{adjustbox}
\caption{\small Prior and sampler for each block of the joint posterior
(Eq.~(\ref{joint_posterior}) in Appendix~\ref{Posterior_App}); RWMH: random-walk Metropolis--Hastings. Normal distributions are parameterized by mean and variance; for $\eta_r^t$, $\tilde\rho_r=\rho_r\mathbbm{1}\{t>1\}$, so the prior mean is zero for the first subject rated by $r$.}
\label{tab:samplers}
\end{table}

\section{Olfactory Mono-Molecular Odors Assessment}\label{Odor}
Human olfactory perception presents a distinctive setting for rater-mediated assessment: unlike educational or clinical evaluations, odor ratings are highly subjective, vary substantially across individuals, and are susceptible to contextual and sequential effects that depend on the order in which stimuli are presented \citep{tversky1974judgment}. These features make it a natural setting for the \gls{degmfr} framework, which explicitly accounts for rater heterogeneity, item-specific difficulty, and dynamic rater behavior.

The analysis uses the dataset from \cite{bierling2025dataset}, who collected perceptual ratings of 74 mono-molecular odor compounds from a large sample of lay participants. Each participant was instructed to smell 10 odors and evaluate them on eight perceptual dimensions: pleasantness, intensity, familiarity, warmth, coldness, irritation, edibility, and disgust. The study followed an incomplete block design with $2$ anchors; within each odor set, four pseudo-randomized presentation orders (groups A, B, C, D) were used to balance position effects, and two anchor odors were shared across sets to connect the design 
across groups.\\
The original study collected data under two protocols: participants either evaluated the odors remotely at home (\textit{@home}) or in a laboratory setting (\textit{@lab}). We consider the \textit{@home} subset of the study, which provides the largest sample (8,495 observed ratings from 863 participants on 58 odors, after removing missing ratings)  and for which the randomization procedure is fully documented. 
For computational tractability, we further restrict the analysis to a random sample of $R = 100$ raters. 
We model the attractiveness of each odor as a shared underlying olfactory construct, integrating four key perceptual items: \textit{pleasant}, \textit{warm}, \textit{familiar}, and \textit{edible}. These items represent salient perceptual aspects of odor experience \citep{berglund1976dimensionality}. Following \citet{secundo2015individual}, attractiveness is defined as a composite construct combining hedonic (\textit{pleasantness}), cognitive (\textit{familiarity}), biological (\textit{edibility}), and trigeminal (\textit{warm}) facets of odor perception. The resulting dataset comprises $R = 100$ raters, $I = 4$ items, $J = 58$ odors for a total of 1000 rater--odor pairs. Responses were originally recorded on continuous visual analog scales ranging from 0 to 100. For illustrative purposes and to align with our ordinal response model, we coarsened all dimensions into $K+1=5$ ordered categories using common thresholds (0--20, 20--40, 40--60, 60--80, 80--100) applied uniformly across ratings. 
 

We compare six nested models using the Watanabe--Akaike information criterion (WAIC): the \gls{gmfr}, the \gls{gmfr} with a \gls{dp} mixture prior over the latent trait distribution (G-MFR-DP), the \gls{egmfr}, the \gls{degmfr}, and the DEG-MFR models with a \gls{pyp} mixture over the latent trait distribution with discount parameter equal, respectively, to $d=0.3$ and $d=0.6$. WAIC provides a fully Bayesian approach to estimating the out-of-sample expectation, relying on the computed log-pointwise posterior predictive density and a penalty term on the effective number of parameters to prevent overfitting \citep{Gelman2014}. The same comparison for the clinical assessment data is reported in Appendix~\ref{application}. Based on graphical inspections of the MCMCs and diagnostics from the \texttt{coda} package \citep{CODA}, we found no convergence or mixing issues; further details and trace-plot examples are in Section 4 of the Supplementary Materials. Table~\ref {Tab_WAIC_bis} reports the WAIC indices for all fitted models. The DEG-MFR DP model achieves the lowest WAIC, indicating superior out-of-sample predictive performance relative to all competing specifications. Based on this model comparison, we focus on results from the DEG-MFR DP. 

We report Expected A Posteriori (EAP) estimates and 95\% credible intervals for a subset of subject latent traits $\theta_j$, rater-specific parameters ($\xi_r$, $\omega_r^2$, $\rho_r$, $\{\tau_{rk}\}_{k=1}^{K}$), and item discrimination parameters $\lambda_i$ in Table~\ref{Tab_posteriors_Odor} of Appendix~\ref{Additional_fig_tab}. Some raters show EAP estimates of $\rho_r$ larger than 1 in absolute value (see Figure~\ref{Wright_map_odor}), suggesting a strong anchoring effect; however, their 95\% credible intervals still include zero, and values close to zero cannot be excluded.

Figure~\ref{Density_P-Spline_Odor} displays the posterior estimates of the latent trait density, the posterior point estimate of subjects' clustering structure, and of the rater's practice effect curve $\text{exp}(\nu^{t*})$. The estimated latent trait distribution is unimodal and much more concentrated than the standard normal, with a sharper, taller peak slightly below zero, suggesting that odors are fairly similar to one another and tend to receive, on average, lower category scores. The posterior point estimate of the number of clusters in the sample suggested by the VI loss function is $\hat{M_N}=2$ (see Figure~\ref{PSM_real_data_set} in Appendix \ref{Additional_fig_tab}). Subject $j=43$ is assigned to a different cluster than the other subjects, suggesting this odor is systematically more appealing than the others.

The practice effect curve increases from about $0.86$ for the first odor rated to a peak of about $1.10$ between positions $t=8$ and $t=9$, and then slightly decreases to about $1.05$ between $t=9$ and $t=10$. The 95\% credible interval at $t=1$ does not overlap with those at $t=7$ and $t=8$, and from $t=3$ and $t=4$ onward the 95\% credible intervals no longer contain the EAP estimate of $\text{exp}(\nu^{t*})$ at $t=1$. This suggests that raters progressively improve their discrimination ability as they rate more odors.

Figures~\ref{ICC_Odor_Dynamic_1} and~\ref{ICC_Odor_Dynamic_2} in Appendix \ref{Additional_fig_tab} show, respectively, through the Dynamic ICCs of raters $r=38$ and $r=68$, how their specific ICCs change across the respective rating sequences for the same items (\textit{familiar} and \textit{warm}, respectively). The larger values of the anchoring effect $\rho_{38}$ of rater $r=38$ are reflected in the more prominent shifts in the relative ICCs across the rated odors, suggesting a poor, inconsistent use of the response categories for \textit{familiar} across different odors.Figure~\ref{ICC_Odor_2} compares the ICCs of the subset of raters $\mathcal{R}_j$ who score odor $j=50$ on \textit{warm}. For the same odor and the same item, raters imply markedly different ICCs, and thus different probabilities of assigning each category.

Table~\ref{post_summary_LIC_odors} in Appendix~\ref{Additional_fig_tab} reports the posterior estimates of the pairwise latent inter-rater consensus (LIC; Eq.~(\ref{ICC}) in Section~\ref{Agreement}) between rater 1 and other raters under different rating configurations. EAP estimates are small to moderate ($0.280$--$0.410$). For raters 1 and 3, with opposite anchoring tendencies ($\rho_1\rho_3<0$), consensus is highest when both rate the odor first ($0.370$) and lowest when both have just rated the same odor ($0.280$), consistent with Proposition~\ref{prop_1}; however, credible intervals largely overlap.

The Wright map in Figure~\ref{Wright_map_odor} displays the rater local difficulty shifts for \textit{warm} on a subset of raters. The estimated shifts span a range comparable to the latent trait distribution. For instance, the 95\% credible intervals of $\beta_{2,83}$ and $\beta_{2,70}$ do not overlap, indicating that raters 83 and 70 differ credibly in how they locate the second threshold of the \textit{warm} scale.


\begin{table}
\centering
\begin{tabular}{l c c}
  \toprule
  Fitted Model & $WAIC$ & $\Delta WAIC$ \\
  \midrule
  DEG-MFR DP Model            & 9646.0  & $-$ \\
  DEG-MFR PYP Model $(d=0.3)$ & 9656.1  & -10.1 \\
  EG-MFR DP Model             & 9659.8  & -13.8 \\
  DEG-MFR PYP Model $(d=0.6)$ & 9667.6  & -21.6 \\
  G-MFR DP Model              & 11158.7 & -1512.6 \\
  G-MFR Model                 & 11166.3 & -1520.3 \\
        \bottomrule
        \end{tabular}
\caption[]{\small The Watanabe--Akaike information criterion (WAIC) is reported for each of the fitted models on Odors Assessment data: Generalized Many-Facet Rasch model (G-MFR), the Generalized Many-Facet Rasch model with a \gls{dp} mixture prior over the latent trait distribution (G-MFRM DP), the Extended Generalized Many-Facet Rasch model (EG-MFR DP), the Dynamic Extended Generalized Many-Facet Rasch model (DEG-MFR DP), and the DEG-MFR models with a \gls{pyp} mixture over the latent trait distribution (DEG-MFR PYP) Model with discount parameter equal, respectively, to $d=0.3$ and $d=0.6$. The pairwise WAIC difference ($\Delta WAIC$) between the model with the best fit and the others is reported.}
\label{Tab_WAIC_bis}
\end{table}

\begin{figure}
    \centering
    \includegraphics[scale=0.75]{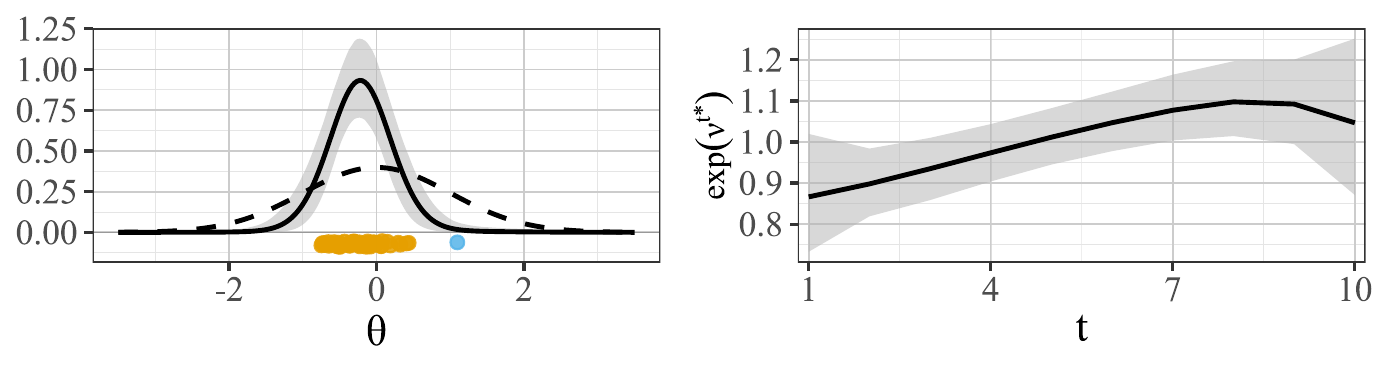}
    \caption{ Left panel: graphical representation of the estimated latent trait density; the black solid line indicates the Expected A Posteriori, the shaded grey area indicates the 95\% quantile-based credible interval, and the black dashed line indicates the standard normal reference density. Right panel: graphical representation of the estimated rater practice effect $\text{exp}(\nu^{t*})$; the black solid line and the shaded grey area indicate, respectively, the  Expected A Posteriori and the 95\% quantile-based credible interval of the P-Spline.}
    \label{Density_P-Spline_Odor}
\end{figure}

\begin{figure}
    \centering
     \includegraphics[scale= 0.75]{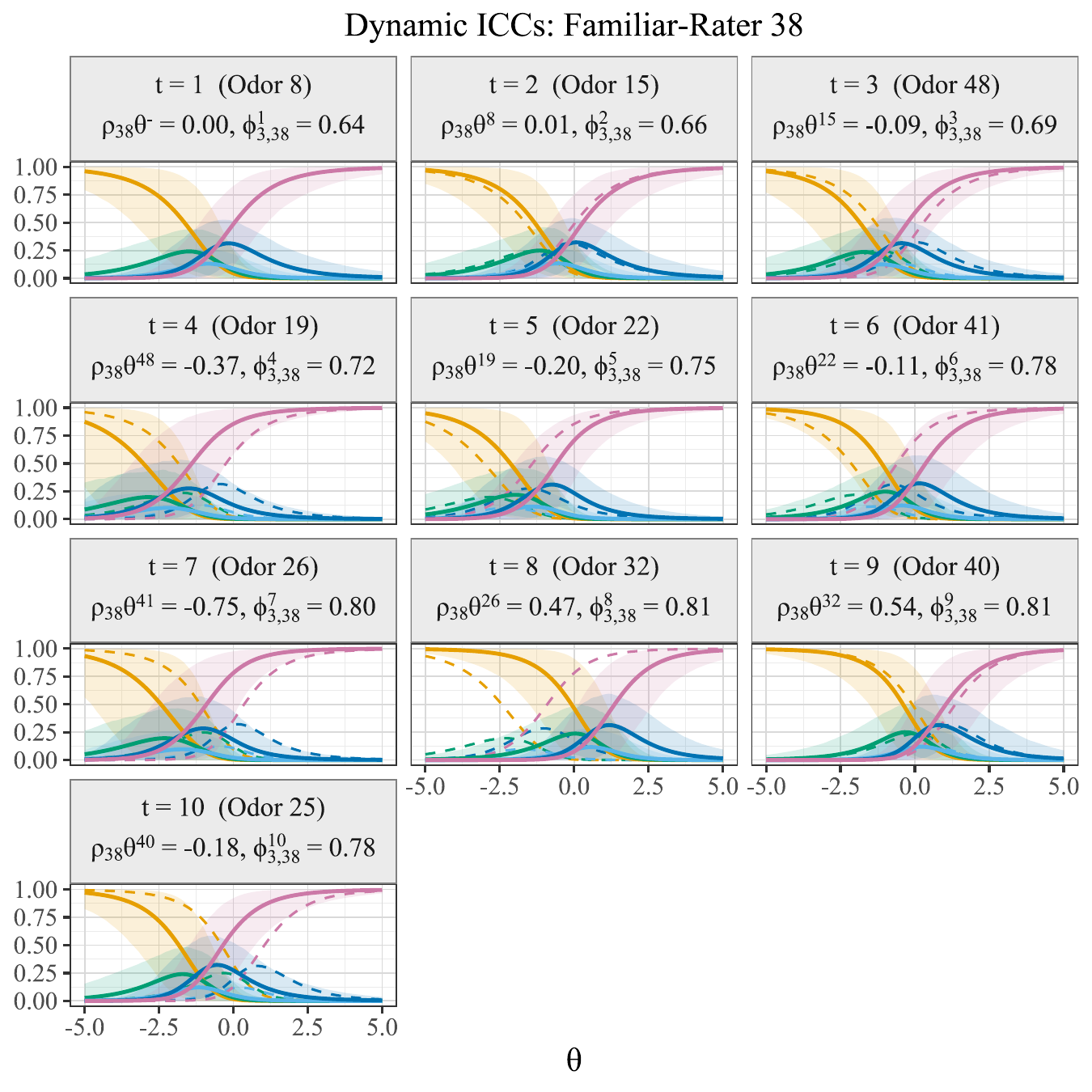}
 \caption{\small Dynamic Item-Rater Characteristic Curve for \textit{familiar} and rater 38 across the sequence of odors assessed, from Olfactory Assessment data. The orange, green, aquamarine, blue, and violet lines represent the five ordered response categories; the corresponding solid lines and shaded areas indicate the Expected A Posteriori and the 95\% quantile-based credible interval for the category response curves, respectively. Dashed lines reproduce the Expected A Posteriori curves of the odor rated at the previous position $t-1$, so that each panel shows how location and shape of the curves change from one odor to the next. Each panel title reports the posterior mean of the anchoring effect $\rho_r\theta^{t-1}$ and of the dynamic joint discrimination $\phi^t_{ir}$. Curves are computed controlling for the odor-specific item residual $\epsilon_{ij}$, so that differences across panels reflect only rater-related effects. }
    \label{ICC_Odor_Dynamic_1}
\end{figure}

\begin{figure}
    \centering
    \includegraphics[scale=0.6]{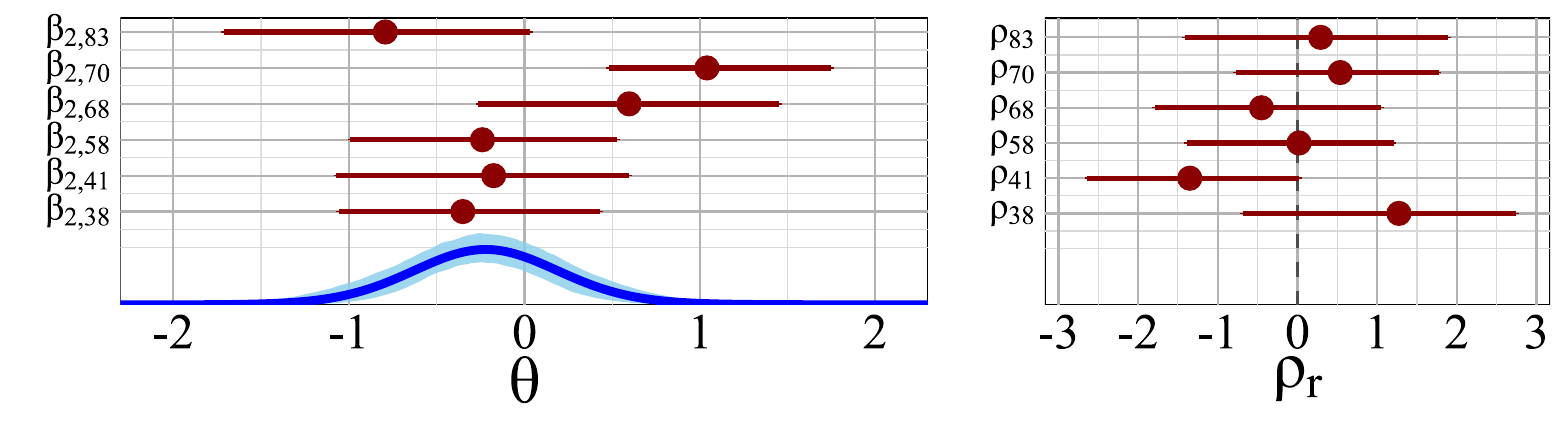}
    \caption{\small Left panel: Wright map of raters' local difficulty shifts $\beta_{ir}$ for \textit{warm} on a subset of raters from Odors Assessment data. Right panel: posterior anchoring effects $\rho_r$ for the same raters; the dashed vertical line marks $\rho_r=0$. For each parameter, the Expected A Posteriori value and the 95\% quantile-based credible interval are indicated by the red point and the bar, respectively. The Expected A Posteriori value and the 95\% quantile-based credible interval of the latent trait density are indicated, respectively, by the blue solid line and the blue shaded area.}
    \label{Wright_map_odor}
\end{figure}

\begin{figure}
    \centering
    \includegraphics[scale= 0.9]{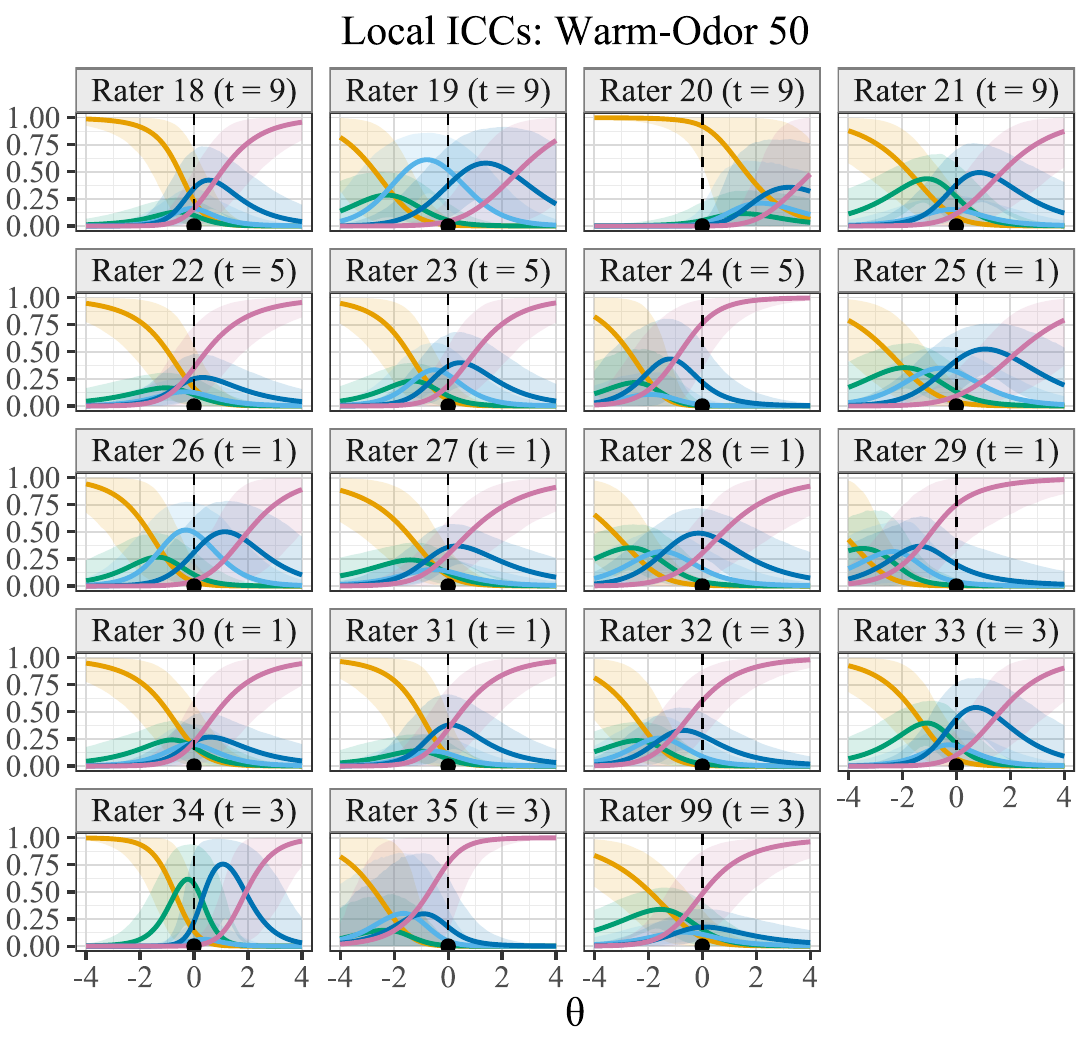}

    \caption{\small Example of Item-Odor Characteristic Curves for \textit{warm} and odor $j=50$ from Olfactory Assessment data. The orange, green, aquamarine, blue, and violet lines represent, respectively, the five ordered response categories; the corresponding solid lines and shaded areas indicate, respectively, the Expected A Posteriori and the 95\% quantile-based credible interval for the category response curves; the black point and the vertical dashed line indicate the Expected A Posteriori latent trait of odor 50.}
    \label{ICC_Odor_2}
\end{figure}

\newpage
\section{Conclusion}\label{Discussion}
We introduced the \acrfull{degmfr} model, a flexible Bayesian semiparametric framework for rater-mediated assessment that aims to combine the modeling strategies of the two dominant research traditions in \gls{irt}-based multi-rater modeling developed over the past three decades. The first tradition, originating with \cite{linacre1989many} and extended through the family of \gls{gmfr} models \citep{Uto_2016, Uto_2020, Uto_24}, focuses on modeling the joint effect of item characteristics and rater-specific features on the probability of observed scores. The second, rooted in the Hierarchical Rater Model \citep{Patz2002} and its extensions \citep{Molenaar2021, Casabianca_2015}, emphasizes the distinction between the subject's \textit{ideal rating} and the raters' rating. The \gls{degmfr} model bridges these two lines by simultaneously accounting for multiple dyad-specific effects (i.e., item-subject difficulty, dynamic rater-subject severity, and item-category and rater-category difficulty), within a single, coherent statistical framework. This formulation addresses the \textit{information accumulation} problem identified by \cite{mariano2002critique} and enables richer inferences about both subjects and raters than either tradition alone affords. 

A distinctive practical contribution of the framework is its ability to produce \textit{personalized} posterior summaries for each rater and subject. The rater-specific parameters---local difficulty shifts $\beta_{ir}$, consistency $\omega_r^2$, discrimination $\xi_r$, and anchoring bias $\rho_r$---together with the average practice effect $\nu^{t*}$, characterize different aspects of rating behavior and can be used to investigate systematic biases, identify potentially inconsistent evaluators, and provide information for formative feedback. Heterogeneity in these parameters translates directly into heterogeneity in rating category probabilities: as illustrated in Figures~\ref{ICC_Odor_2} and~\ref{ICC_H} in Appendix \ref{Additional_fig_tab}, different raters produce markedly different Item Characteristic Curves across and within items, reflecting the combined influence of item-specific difficulty $\epsilon_{ij}$, rater-specific severity $\beta_{ir}$, and the dynamic effects captured by the practice effect $\nu^{t*}$, the anchoring bias $\rho_r$, and the dynamic local bias $\eta_r^t$. Figure~\ref{ICC_Odor_Dynamic_1} further illustrates how rater's category response curves can change across the rating sequence, reflecting the combined effects of sequential changes in location and discrimination. 
The Wright map of rater local difficulty shifts (Figure~\ref{Wright_map_odor}) and the posterior pairwise LIC estimates (Table~\ref{post_summary_LIC_odors}) provide complementary tools for auditing rater quality at both the individual and population level. From the subject's perspective, the individual latent trait estimate $\theta_j$ and its credible interval yield a model-adjusted estimate of ability that accounts for the particular raters and items encountered. These features align naturally with the goals of \textit{precision education} \citep{Cook_2018_precision}, where individualized diagnosis and targeted intervention require measurement models that disentangle learner ability from the idiosyncratic characteristics of the assessment instrument and its evaluators.\\
The local independence assumption that typically characterizes subjects within Many-Facet Rasch Models \citep{Uto_2020,linacre1989many} is relaxed at the measurement level through the rater anchoring effect, which introduces sequential dependence between ratings of different subjects. At the structural level, the nonparametric mixture independently relaxes the assumption of a homogeneous parametric latent-trait distribution, allowing latent heterogeneity and clustering among subjects.
\\
Beyond individual profiling, the framework supports a principled analysis of \textit{latent inter-rater consensus}. 
We proposed a pairwise Latent Inter-Rater Consensus index, $LIC_{rr'}$ (Section~\ref{Agreement}), defined as the correlation between the bias-adjusted latent evaluations of two raters for the same subject. 

Methodologically, $LIC_{rr'}$ admits a more general characterization than the intraclass and reliability coefficients typically derived in multilevel \gls{irt} \citep{ChoShenNaveiras2019}. Unlike these measures, it takes values in $(-1,1)$, so negative values can signal sequence-induced latent dissensus, and, being pairwise, it can inform rating designs, for instance in how raters are paired \citep{glas2024reducing}.
\\
The simulation results in Appendix~\ref{sim} show that the \gls{egmfr} and \gls{degmfr} models generally improve parameter recovery relative to the \gls{gmfr} model, with the \gls{degmfr} model providing more accurate latent-trait estimates under the dynamic generative scenarios. In the clinical OSCE application (Appendix~\ref{application}), the rater-specific posterior summaries identified a rater who was both less discriminating and less consistent across subjects than the rest of the panel, and who correspondingly showed the lowest latent consensus with the other raters, as well as item-specific differences in rater severity. 
These properties may support future applications in automated or computer-assisted scoring systems, in which posterior estimates could inform item calibration, rater monitoring, and score reporting. 
\\ 
An increasingly relevant application domain is the comparison of human and AI-generated ratings \citep{Uto_2021_AI}. The \gls{degmfr} framework could be useful in this setting because it allows multiple sources of rater heterogeneity and sequential dependence to be estimated and compared across different types of evaluators. This may provide a basis for studying hybrid scoring systems combining human and AI-generated ratings.

Several limitations of the present work point to directions for future research. The model assumes a unidimensional latent trait, which may be overly restrictive in settings where subjects are evaluated on conceptually distinct dimensions; a multidimensional extension incorporating a multivariate nonparametric prior over a vector of latent traits would broaden the scope of the framework considerably. 
Under the current formulation, the effect of rating position on log-discrimination, $\nu^t$, is a trend shared across raters. This may be restrictive in the presence of high rater heterogeneity, and more flexible model specifications could allow rater-specific deviations from this average latent trend. However, based on our preliminary simulation results, a large number of items and subjects are needed to estimate $\nu^t$ and rater-specific deviations jointly.
While the P-spline captures smooth rating-position trends in rater discrimination, more abrupt change-points arising from fatigue or recalibration may require nonparametric change-point models. Extending the framework to longitudinal settings, in which the same subjects are rated at multiple time points, would enable joint modeling of ability growth and rater dynamics \citep{Uto_2022}. Under the proposed framework, we assume that subjects and raters belong to disjoint sets, which makes it unsuitable for peer grading settings in which the same individual is both assessed and evaluates others \citep{Mignemi_Chen_Moustaki_2025}. These directions are left for future investigation.

\section*{Acknowledgement}
The study was partly funded by the project ``Research of Excellence on Digital Technologies and Wellbeing CZ.02.01.01/\allowbreak 00/\allowbreak 22\_008/\allowbreak 0004583'' which is co-financed by the European Union.\\

\bibliographystyle{apacite}

\bibliography{reference} 

\begin{thebibliography}{}

\bibitem [\protect \citeauthoryear {%
Dahl%
, Johnson%
\BCBL {}\ \BBA {} M{\"u}ller%
}{%
Dahl%
\ \protect \BOthers {.}}{%
{\protect \APACyear {2022}}%
}]{%
Dahl_salso}
\APACinsertmetastar {%
Dahl_salso}%
\begin{APACrefauthors}%
Dahl, D\BPBI B.%
, Johnson, D\BPBI J.%
\BCBL {}\ \BBA {} M{\"u}ller, P.%
\end{APACrefauthors}%
\unskip\
\newblock
\APACrefYearMonthDay{2022}{}{}.
\newblock
{\BBOQ}\APACrefatitle {Search Algorithms and Loss Functions for {Bayesian}
  Clustering} {Search algorithms and loss functions for {Bayesian}
  clustering}.{\BBCQ}
\newblock
\APACjournalVolNumPages{Journal of Computational and Graphical
  Statistics}{31}{4}{1189--1201}.
\PrintBackRefs{\CurrentBib}

\bibitem [\protect \citeauthoryear {%
{de Valpine}%
\ \protect \BOthers {.}}{%
{de Valpine}%
\ \protect \BOthers {.}}{%
{\protect \APACyear {2026}}%
}]{%
nimble-software2026}
\APACinsertmetastar {%
nimble-software2026}%
\begin{APACrefauthors}%
{de Valpine}, P.%
, Paciorek, C.%
, Turek, D.%
, Michaud, N.%
, Anderson-Bergman, C.%
, Obermeyer, F.%
\BDBL {}Paganin, S.%
\end{APACrefauthors}%
\unskip\
\newblock
\APACrefYearMonthDay{2026}{}{}.
\newblock
{\BBOQ}\APACrefatitle {{NIMBLE}: {MCMC}, Particle Filtering, and Programmable
  Hierarchical Modeling} {{NIMBLE}: {MCMC}, particle filtering, and
  programmable hierarchical modeling}{\BBCQ}\ [\bibcomputersoftwaremanual].
\newblock
\begin{APACrefURL} \url{https://cran.r-project.org/package=nimble}
  \end{APACrefURL}
\newblock
\APACrefnote{{R} package version 1.4.2}
\PrintBackRefs{\CurrentBib}

\bibitem [\protect \citeauthoryear {%
De~Blasi%
\ \protect \BOthers {.}}{%
De~Blasi%
\ \protect \BOthers {.}}{%
{\protect \APACyear {2015}}%
}]{%
DeBlasi2015}
\APACinsertmetastar {%
DeBlasi2015}%
\begin{APACrefauthors}%
De~Blasi, P.%
, Favaro, S.%
, Lijoi, A.%
, Mena, R\BPBI H.%
, Pr{\"u}nster, I.%
\BCBL {}\ \BBA {} Ruggiero, M.%
\end{APACrefauthors}%
\unskip\
\newblock
\APACrefYearMonthDay{2015}{}{}.
\newblock
{\BBOQ}\APACrefatitle {Are Gibbs-Type Priors the Most Natural Generalization of
  the {D}irichlet Process?} {Are gibbs-type priors the most natural
  generalization of the {D}irichlet process?}{\BBCQ}
\newblock
\APACjournalVolNumPages{IEEE Transactions on pattern analysis and machine
  intelligence}{37}{}{}.
\PrintBackRefs{\CurrentBib}

\bibitem [\protect \citeauthoryear {%
De~Boeck%
\ \BBA {} Wilson%
}{%
De~Boeck%
\ \BBA {} Wilson%
}{%
{\protect \APACyear {2004}}%
}]{%
deboeck2004explanatory}
\APACinsertmetastar {%
deboeck2004explanatory}%
\begin{APACrefauthors}%
De~Boeck, P.%
\BCBT {}\ \BBA {} Wilson, M.%
\end{APACrefauthors}%
\unskip\
\newblock
\APACrefYear{2004}.
\newblock
\APACrefbtitle {Explanatory Item Response Models: A Generalized Linear and
  Nonlinear Approach} {Explanatory item response models: A generalized linear
  and nonlinear approach}.
\newblock
\APACaddressPublisher{New York}{Springer}.
\PrintBackRefs{\CurrentBib}

\bibitem [\protect \citeauthoryear {%
Eckes%
\ \BBA {} Jin%
}{%
Eckes%
\ \BBA {} Jin%
}{%
{\protect \APACyear {2024}}%
}]{%
Eckes_2024}
\APACinsertmetastar {%
Eckes_2024}%
\begin{APACrefauthors}%
Eckes, T.%
\BCBT {}\ \BBA {} Jin, K\BPBI Y.%
\end{APACrefauthors}%
\unskip\
\newblock
\APACrefYearMonthDay{2024}{}{}.
\newblock
{\BBOQ}\APACrefatitle {Human ratings take time: {A} hierarchical facets model
  for the joint analysis of ratings and rating times} {Human ratings take time:
  {A} hierarchical facets model for the joint analysis of ratings and rating
  times}.{\BBCQ}
\newblock
\APACjournalVolNumPages{Behavior Research Methods}{56}{}{3535--3547}.
\PrintBackRefs{\CurrentBib}

\bibitem [\protect \citeauthoryear {%
Eilers%
\ \BBA {} Marx%
}{%
Eilers%
\ \BBA {} Marx%
}{%
{\protect \APACyear {1996}}%
}]{%
EilersMarx1996}
\APACinsertmetastar {%
EilersMarx1996}%
\begin{APACrefauthors}%
Eilers, P\BPBI H\BPBI C.%
\BCBT {}\ \BBA {} Marx, B\BPBI D.%
\end{APACrefauthors}%
\unskip\
\newblock
\APACrefYearMonthDay{1996}{}{}.
\newblock
{\BBOQ}\APACrefatitle {Flexible smoothing with {B}-splines and penalties}
  {Flexible smoothing with {B}-splines and penalties}.{\BBCQ}
\newblock
\APACjournalVolNumPages{Statistical Science}{11}{2}{89--102}.
\PrintBackRefs{\CurrentBib}

\bibitem [\protect \citeauthoryear {%
Escobar%
\ \BBA {} West%
}{%
Escobar%
\ \BBA {} West%
}{%
{\protect \APACyear {1994}}%
}]{%
Escobar_95}
\APACinsertmetastar {%
Escobar_95}%
\begin{APACrefauthors}%
Escobar, M.%
\BCBT {}\ \BBA {} West, M.%
\end{APACrefauthors}%
\unskip\
\newblock
\APACrefYearMonthDay{1994}{}{}.
\newblock
{\BBOQ}\APACrefatitle {Bayesian Density Estimation and Inference Using
  Mixtures} {Bayesian density estimation and inference using mixtures}.{\BBCQ}
\newblock
\APACjournalVolNumPages{Journal of the American Statistical
  Association}{90}{}{577--588}.
\PrintBackRefs{\CurrentBib}

\bibitem [\protect \citeauthoryear {%
Gelman%
\ \protect \BOthers {.}}{%
Gelman%
\ \protect \BOthers {.}}{%
{\protect \APACyear {2013}}%
}]{%
BDA3}
\APACinsertmetastar {%
BDA3}%
\begin{APACrefauthors}%
Gelman, A.%
, Carlin, J.%
, Stern, H.%
, Dunson, D.%
, Vehtari, A.%
\BCBL {}\ \BBA {} Rubin, D.%
\end{APACrefauthors}%
\unskip\
\newblock
\APACrefYear{2013}.
\newblock
\APACrefbtitle {Bayesian Data Analysis, Third Edition} {Bayesian data analysis,
  third edition}.
\newblock
\APACaddressPublisher{}{Taylor \& Francis}.
\PrintBackRefs{\CurrentBib}

\bibitem [\protect \citeauthoryear {%
Gelman%
, Hwang%
\BCBL {}\ \BBA {} Vehtari%
}{%
Gelman%
\ \protect \BOthers {.}}{%
{\protect \APACyear {2014}}%
}]{%
Gelman2014}
\APACinsertmetastar {%
Gelman2014}%
\begin{APACrefauthors}%
Gelman, A.%
, Hwang, J.%
\BCBL {}\ \BBA {} Vehtari, A.%
\end{APACrefauthors}%
\unskip\
\newblock
\APACrefYearMonthDay{2014}{}{}.
\newblock
{\BBOQ}\APACrefatitle {Understanding predictive information criteria for
  {B}ayesian models} {Understanding predictive information criteria for
  {B}ayesian models}.{\BBCQ}
\newblock
\APACjournalVolNumPages{Stat Comput}{24}{}{997-1016}.
\PrintBackRefs{\CurrentBib}

\bibitem [\protect \citeauthoryear {%
Ghosal%
, Ghosh%
\BCBL {}\ \BBA {} Ramamoorthi%
}{%
Ghosal%
\ \protect \BOthers {.}}{%
{\protect \APACyear {1999}}%
}]{%
Ghosal1999}
\APACinsertmetastar {%
Ghosal1999}%
\begin{APACrefauthors}%
Ghosal, S.%
, Ghosh, J\BPBI K.%
\BCBL {}\ \BBA {} Ramamoorthi, R\BPBI V.%
\end{APACrefauthors}%
\unskip\
\newblock
\APACrefYearMonthDay{1999}{}{}.
\newblock
{\BBOQ}\APACrefatitle {Posterior consistency of {D}irichlet mixtures in density
  estimation} {Posterior consistency of {D}irichlet mixtures in density
  estimation}.{\BBCQ}
\newblock
\APACjournalVolNumPages{Annals of Statistics}{27}{}{}.
\PrintBackRefs{\CurrentBib}

\bibitem [\protect \citeauthoryear {%
Ghosal%
\ \BBA {} van~der Vaart%
}{%
Ghosal%
\ \BBA {} van~der Vaart%
}{%
{\protect \APACyear {2017}}%
}]{%
Ghosal_van_der_Vaart_2017}
\APACinsertmetastar {%
Ghosal_van_der_Vaart_2017}%
\begin{APACrefauthors}%
Ghosal, S.%
\BCBT {}\ \BBA {} van~der Vaart, A.%
\end{APACrefauthors}%
\unskip\
\newblock
\APACrefYear{2017}.
\newblock
\APACrefbtitle {Fundamentals of Nonparametric {B}ayesian Inference}
  {Fundamentals of nonparametric {B}ayesian inference}.
\newblock
\APACaddressPublisher{}{Cambridge University Press}.
\PrintBackRefs{\CurrentBib}

\bibitem [\protect \citeauthoryear {%
Heinzl%
, Fahrmeir%
\BCBL {}\ \BBA {} Kneib%
}{%
Heinzl%
\ \protect \BOthers {.}}{%
{\protect \APACyear {2012}}%
}]{%
heinzl2012}
\APACinsertmetastar {%
heinzl2012}%
\begin{APACrefauthors}%
Heinzl, F.%
, Fahrmeir, L.%
\BCBL {}\ \BBA {} Kneib, T.%
\end{APACrefauthors}%
\unskip\
\newblock
\APACrefYearMonthDay{2012}{}{}.
\newblock
{\BBOQ}\APACrefatitle {Additive mixed models with {D}irichlet process mixture
  and P-spline priors} {Additive mixed models with {D}irichlet process mixture
  and p-spline priors}.{\BBCQ}
\newblock
\APACjournalVolNumPages{Asta advances in statistical analysis}{96}{}{47--68}.
\PrintBackRefs{\CurrentBib}

\bibitem [\protect \citeauthoryear {%
Hjort%
, Holmes%
, M{\"u}ller%
\BCBL {}\ \BBA {} Walker%
}{%
Hjort%
\ \protect \BOthers {.}}{%
{\protect \APACyear {2010}}%
}]{%
Hjort_Holmes_Muller_Walker_2010}
\APACinsertmetastar {%
Hjort_Holmes_Muller_Walker_2010}%
\begin{APACrefauthors}%
Hjort, N.%
, Holmes, C.%
, M{\"u}ller, P.%
\BCBL {}\ \BBA {} Walker, S.%
\end{APACrefauthors}%
\unskip\
\newblock
\APACrefYear{2010}.
\newblock
\APACrefbtitle {Bayesian Nonparametrics} {Bayesian nonparametrics}.
\newblock
\APACaddressPublisher{}{Cambridge University Press}.
\PrintBackRefs{\CurrentBib}

\bibitem [\protect \citeauthoryear {%
Ishwaran%
\ \BBA {} James%
}{%
Ishwaran%
\ \BBA {} James%
}{%
{\protect \APACyear {2001}}%
}]{%
Ishwaran2001}
\APACinsertmetastar {%
Ishwaran2001}%
\begin{APACrefauthors}%
Ishwaran, H.%
\BCBT {}\ \BBA {} James, L.%
\end{APACrefauthors}%
\unskip\
\newblock
\APACrefYearMonthDay{2001}{}{}.
\newblock
{\BBOQ}\APACrefatitle {Gibbs Sampling Methods for Stick-Breaking Priors} {Gibbs
  sampling methods for stick-breaking priors}.{\BBCQ}
\newblock
\APACjournalVolNumPages{Journal of the American Statistical
  Association}{96}{}{161-173}.
\PrintBackRefs{\CurrentBib}

\bibitem [\protect \citeauthoryear {%
Livingstone%
\ \BBA {} Zanella%
}{%
Livingstone%
\ \BBA {} Zanella%
}{%
{\protect \APACyear {2022}}%
}]{%
LivingstoneZanella2022}
\APACinsertmetastar {%
LivingstoneZanella2022}%
\begin{APACrefauthors}%
Livingstone, S.%
\BCBT {}\ \BBA {} Zanella, G.%
\end{APACrefauthors}%
\unskip\
\newblock
\APACrefYearMonthDay{2022}{}{}.
\newblock
{\BBOQ}\APACrefatitle {The Barker proposal: combining robustness and efficiency
  in gradient-based {MCMC}} {The barker proposal: combining robustness and
  efficiency in gradient-based {MCMC}}.{\BBCQ}
\newblock
\APACjournalVolNumPages{Journal of the Royal Statistical Society Series
  B}{84}{2}{496--523}.
\PrintBackRefs{\CurrentBib}

\bibitem [\protect \citeauthoryear {%
Mariano%
}{%
Mariano%
}{%
{\protect \APACyear {2002}}%
}]{%
mariano2002critique}
\APACinsertmetastar {%
mariano2002critique}%
\begin{APACrefauthors}%
Mariano, L\BPBI T.%
\end{APACrefauthors}%
\unskip\
\newblock
\APACrefYear{2002}.
\unskip\
\newblock
\APACrefbtitle {{Information Accumulation and the Many-Facet {Rasch} Model: A
  Critical Examination}} {{Information Accumulation and the Many-Facet {Rasch}
  Model: A Critical Examination}}\ \APACtypeAddressSchool {\BUPhD}{}{}.
\unskip\
\newblock
\APACaddressSchool {}{RAND Graduate School}.
\unskip\
\newblock
\APACrefnote{Doctoral dissertation discussing limitations of conventional MFRM
  (e.g., information accumulation)}
\PrintBackRefs{\CurrentBib}

\bibitem [\protect \citeauthoryear {%
Mariano%
\ \BBA {} Junker%
}{%
Mariano%
\ \BBA {} Junker%
}{%
{\protect \APACyear {2007}}%
}]{%
Mariano2007covariates}
\APACinsertmetastar {%
Mariano2007covariates}%
\begin{APACrefauthors}%
Mariano, L\BPBI T.%
\BCBT {}\ \BBA {} Junker, B\BPBI W.%
\end{APACrefauthors}%
\unskip\
\newblock
\APACrefYearMonthDay{2007}{}{}.
\newblock
{\BBOQ}\APACrefatitle {Covariates of the Rating Process in Hierarchical Models
  for Multiple Ratings of Test Items} {Covariates of the rating process in
  hierarchical models for multiple ratings of test items}.{\BBCQ}
\newblock
\APACjournalVolNumPages{Journal of Educational and Behavioral
  Statistics}{32}{3}{287--314}.
\PrintBackRefs{\CurrentBib}

\bibitem [\protect \citeauthoryear {%
Meil{\u{a}}%
}{%
Meil{\u{a}}%
}{%
{\protect \APACyear {2007}}%
}]{%
Meila_2007}
\APACinsertmetastar {%
Meila_2007}%
\begin{APACrefauthors}%
Meil{\u{a}}, M.%
\end{APACrefauthors}%
\unskip\
\newblock
\APACrefYearMonthDay{2007}{}{}.
\newblock
{\BBOQ}\APACrefatitle {Comparing clusterings---an information based distance}
  {Comparing clusterings---an information based distance}.{\BBCQ}
\newblock
\APACjournalVolNumPages{Journal of Multivariate Analysis}{98}{5}{873-895}.
\PrintBackRefs{\CurrentBib}

\bibitem [\protect \citeauthoryear {%
Mignemi%
, Chen%
\BCBL {}\ \BBA {} Moustaki%
}{%
Mignemi%
\ \protect \BOthers {.}}{%
{\protect \APACyear {2025}}%
}]{%
Mignemi_Chen_Moustaki_2025}
\APACinsertmetastar {%
Mignemi_Chen_Moustaki_2025}%
\begin{APACrefauthors}%
Mignemi, G.%
, Chen, Y.%
\BCBL {}\ \BBA {} Moustaki, I.%
\end{APACrefauthors}%
\unskip\
\newblock
\APACrefYearMonthDay{2025}{}{}.
\newblock
{\BBOQ}\APACrefatitle {Unfolding the Network of Peer Grades: A Latent Variable
  Approach} {Unfolding the network of peer grades: A latent variable
  approach}.{\BBCQ}
\newblock
\APACjournalVolNumPages{Psychometrika}{}{}{1--22}.
\PrintBackRefs{\CurrentBib}

\bibitem [\protect \citeauthoryear {%
Mignemi%
\ \BBA {} Manolopoulou%
}{%
Mignemi%
\ \BBA {} Manolopoulou%
}{%
{\protect \APACyear {2025}}%
}]{%
mignemi2025BNP}
\APACinsertmetastar {%
mignemi2025BNP}%
\begin{APACrefauthors}%
Mignemi, G.%
\BCBT {}\ \BBA {} Manolopoulou, I.%
\end{APACrefauthors}%
\unskip\
\newblock
\APACrefYearMonthDay{2025}{}{}.
\newblock
{\BBOQ}\APACrefatitle {Bayesian Nonparametric Models for Multiple Raters: A
  General Statistical Framework} {Bayesian nonparametric models for multiple
  raters: A general statistical framework}.{\BBCQ}
\newblock
\APACjournalVolNumPages{Psychometrika}{90}{4}{}.
\PrintBackRefs{\CurrentBib}

\bibitem [\protect \citeauthoryear {%
Paganin%
\ \protect \BOthers {.}}{%
Paganin%
\ \protect \BOthers {.}}{%
{\protect \APACyear {2023}}%
}]{%
Paganin2023}
\APACinsertmetastar {%
Paganin2023}%
\begin{APACrefauthors}%
Paganin, S.%
, Paciorek, C\BPBI J.%
, Wehrhahn, C.%
, Rodr{\'i}guez, A.%
, Rabe-Hesketh, S.%
\BCBL {}\ \BBA {} de Valpine, P.%
\end{APACrefauthors}%
\unskip\
\newblock
\APACrefYearMonthDay{2023}{}{}.
\newblock
{\BBOQ}\APACrefatitle {Computational Strategies and Estimation Performance With
  {Bayesian} Semiparametric Item Response Theory Models} {Computational
  strategies and estimation performance with {Bayesian} semiparametric item
  response theory models}.{\BBCQ}
\newblock
\APACjournalVolNumPages{Journal of Educational and Behavioral
  Statistics}{48}{2}{147-188}.
\PrintBackRefs{\CurrentBib}

\bibitem [\protect \citeauthoryear {%
Patz%
, Junker%
, Johnson%
\BCBL {}\ \BBA {} Mariano%
}{%
Patz%
\ \protect \BOthers {.}}{%
{\protect \APACyear {2002}}%
}]{%
Patz2002}
\APACinsertmetastar {%
Patz2002}%
\begin{APACrefauthors}%
Patz, R\BPBI J.%
, Junker, B\BPBI W.%
, Johnson, M\BPBI S.%
\BCBL {}\ \BBA {} Mariano, L\BPBI T.%
\end{APACrefauthors}%
\unskip\
\newblock
\APACrefYearMonthDay{2002}{}{}.
\newblock
{\BBOQ}\APACrefatitle {The Hierarchical Rater Model for Rated Test Items and
  its Application to Large-Scale Educational Assessment Data} {The hierarchical
  rater model for rated test items and its application to large-scale
  educational assessment data}.{\BBCQ}
\newblock
\APACjournalVolNumPages{{Journal of Educational and Behavioral
  Statistics}}{27}{4}{341-384}.
\PrintBackRefs{\CurrentBib}

\bibitem [\protect \citeauthoryear {%
Plummer%
, Best%
, Cowles%
\BCBL {}\ \BBA {} Vines%
}{%
Plummer%
\ \protect \BOthers {.}}{%
{\protect \APACyear {2006}}%
}]{%
CODA}
\APACinsertmetastar {%
CODA}%
\begin{APACrefauthors}%
Plummer, M.%
, Best, N.%
, Cowles, K.%
\BCBL {}\ \BBA {} Vines, K.%
\end{APACrefauthors}%
\unskip\
\newblock
\APACrefYearMonthDay{2006}{}{}.
\newblock
{\BBOQ}\APACrefatitle {{CODA}: Convergence Diagnosis and Output Analysis for
  {MCMC}} {{CODA}: Convergence diagnosis and output analysis for
  {MCMC}}.{\BBCQ}
\newblock
\APACjournalVolNumPages{R News}{6}{1}{7--11}.
\PrintBackRefs{\CurrentBib}

\bibitem [\protect \citeauthoryear {%
Roberts%
\ \BBA {} Sahu%
}{%
Roberts%
\ \BBA {} Sahu%
}{%
{\protect \APACyear {1997}}%
}]{%
RobertsSahu1997}
\APACinsertmetastar {%
RobertsSahu1997}%
\begin{APACrefauthors}%
Roberts, G.%
\BCBT {}\ \BBA {} Sahu, S.%
\end{APACrefauthors}%
\unskip\
\newblock
\APACrefYearMonthDay{1997}{}{}.
\newblock
{\BBOQ}\APACrefatitle {Updating schemes, correlation structure, blocking and
  parameterization for the Gibbs sampler} {Updating schemes, correlation
  structure, blocking and parameterization for the gibbs sampler}.{\BBCQ}
\newblock
\APACjournalVolNumPages{Journal of the Royal Statistical Society Series
  B}{59}{2}{291--317}.
\PrintBackRefs{\CurrentBib}

\bibitem [\protect \citeauthoryear {%
San~Mart{\'i}n%
, Jara%
, Rolin%
\BCBL {}\ \BBA {} Mouchart%
}{%
San~Mart{\'i}n%
\ \protect \BOthers {.}}{%
{\protect \APACyear {2011}}%
}]{%
SanMartin2011}
\APACinsertmetastar {%
SanMartin2011}%
\begin{APACrefauthors}%
San~Mart{\'i}n, E.%
, Jara, A.%
, Rolin, J\BPBI M.%
\BCBL {}\ \BBA {} Mouchart, M.%
\end{APACrefauthors}%
\unskip\
\newblock
\APACrefYearMonthDay{2011}{}{}.
\newblock
{\BBOQ}\APACrefatitle {On the Bayesian Nonparametric Generalization of
  {IRT}-Type Models} {On the bayesian nonparametric generalization of
  {IRT}-type models}.{\BBCQ}
\newblock
\APACjournalVolNumPages{Psychometrika}{76}{}{385-409}.
\PrintBackRefs{\CurrentBib}

\bibitem [\protect \citeauthoryear {%
Uto%
, Tsuruta%
, Araki%
\BCBL {}\ \BBA {} Ueno%
}{%
Uto%
\ \protect \BOthers {.}}{%
{\protect \APACyear {2024}}%
}]{%
Uto_24}
\APACinsertmetastar {%
Uto_24}%
\begin{APACrefauthors}%
Uto, M.%
, Tsuruta, J.%
, Araki, K.%
\BCBL {}\ \BBA {} Ueno, M.%
\end{APACrefauthors}%
\unskip\
\newblock
\APACrefYearMonthDay{2024}{}{}.
\newblock
{\BBOQ}\APACrefatitle {Item response theory model highlighting rating scale of
  a rubric and rater--rubric interaction in objective structured clinical
  examination} {Item response theory model highlighting rating scale of a
  rubric and rater--rubric interaction in objective structured clinical
  examination}.{\BBCQ}
\newblock
\APACjournalVolNumPages{PLOS ONE}{19}{9}{1-23}.
\PrintBackRefs{\CurrentBib}

\bibitem [\protect \citeauthoryear {%
Uto%
\ \BBA {} Ueno%
}{%
Uto%
\ \BBA {} Ueno%
}{%
{\protect \APACyear {2016}}%
}]{%
Uto_2016}
\APACinsertmetastar {%
Uto_2016}%
\begin{APACrefauthors}%
Uto, M.%
\BCBT {}\ \BBA {} Ueno, M.%
\end{APACrefauthors}%
\unskip\
\newblock
\APACrefYearMonthDay{2016}{}{}.
\newblock
{\BBOQ}\APACrefatitle {Item Response Theory for Peer Assessment} {Item response
  theory for peer assessment}.{\BBCQ}
\newblock
\APACjournalVolNumPages{IEEE Transactions on Learning
  Technologies}{9}{}{157-170}.
\PrintBackRefs{\CurrentBib}

\bibitem [\protect \citeauthoryear {%
Uto%
\ \BBA {} Ueno%
}{%
Uto%
\ \BBA {} Ueno%
}{%
{\protect \APACyear {2020}}%
}]{%
Uto_2020}
\APACinsertmetastar {%
Uto_2020}%
\begin{APACrefauthors}%
Uto, M.%
\BCBT {}\ \BBA {} Ueno, M.%
\end{APACrefauthors}%
\unskip\
\newblock
\APACrefYearMonthDay{2020}{}{}.
\newblock
{\BBOQ}\APACrefatitle {A generalized many-facet {Rasch} model and its
  {Bayesian} estimation using {Hamiltonian Monte Carlo}} {A generalized
  many-facet {Rasch} model and its {Bayesian} estimation using {Hamiltonian
  Monte Carlo}}.{\BBCQ}
\newblock
\APACjournalVolNumPages{Behaviormetrika}{47}{}{1-28}.
\PrintBackRefs{\CurrentBib}

\bibitem [\protect \citeauthoryear {%
Wade%
\ \BBA {} Ghahramani%
}{%
Wade%
\ \BBA {} Ghahramani%
}{%
{\protect \APACyear {2018}}%
}]{%
Wade&Ghahramani_2018}
\APACinsertmetastar {%
Wade&Ghahramani_2018}%
\begin{APACrefauthors}%
Wade, S.%
\BCBT {}\ \BBA {} Ghahramani, Z.%
\end{APACrefauthors}%
\unskip\
\newblock
\APACrefYearMonthDay{2018}{}{}.
\newblock
{\BBOQ}\APACrefatitle {Bayesian Cluster Analysis: Point Estimation and Credible
  Balls (with Discussion)} {Bayesian cluster analysis: Point estimation and
  credible balls (with discussion)}.{\BBCQ}
\newblock
\APACjournalVolNumPages{Bayesian Analysis}{13}{2}{559 -- 626}.
\PrintBackRefs{\CurrentBib}

\bibitem [\protect \citeauthoryear {%
Wood%
}{%
Wood%
}{%
{\protect \APACyear {2017}}%
}]{%
wood2017}
\APACinsertmetastar {%
wood2017}%
\begin{APACrefauthors}%
Wood, S\BPBI N.%
\end{APACrefauthors}%
\unskip\
\newblock
\APACrefYear{2017}.
\newblock
\APACrefbtitle {Generalized Additive Models: An Introduction with {R}}
  {Generalized additive models: An introduction with {R}}\ (\PrintOrdinal{2}\
  \BEd).
\newblock
\APACaddressPublisher{Boca Raton, FL}{Chapman and Hall/CRC}.
\PrintBackRefs{\CurrentBib}

\end{thebibliography}


\begin{thebibliography}{}

\bibitem [\protect \citeauthoryear {%
{AERA}%
, {APA}%
\BCBL {}\ \BBA {} {NCME}%
}{%
{AERA}%
\ \protect \BOthers {.}}{%
{\protect \APACyear {2014}}%
}]{%
AERA2014}
\APACinsertmetastar {%
AERA2014}%
\begin{APACrefauthors}%
{AERA}%
, {APA}%
\BCBL {}\ \BBA {} {NCME}.%
\end{APACrefauthors}%
\unskip\
\newblock
\APACrefYear{2014}.
\newblock
\APACrefbtitle {Standards for Educational and Psychological Testing} {Standards
  for educational and psychological testing}.
\newblock
\APACaddressPublisher{Washington, DC}{American Educational Research
  Association}.
\PrintBackRefs{\CurrentBib}

\bibitem [\protect \citeauthoryear {%
Barto{\v{s}}%
\ \BBA {} Martinkov{\'a}%
}{%
Barto{\v{s}}%
\ \BBA {} Martinkov{\'a}%
}{%
{\protect \APACyear {2024}}%
}]{%
Bartos24}
\APACinsertmetastar {%
Bartos24}%
\begin{APACrefauthors}%
Barto{\v{s}}, F.%
\BCBT {}\ \BBA {} Martinkov{\'a}, P.%
\end{APACrefauthors}%
\unskip\
\newblock
\APACrefYearMonthDay{2024}{}{}.
\newblock
{\BBOQ}\APACrefatitle {Assessing quality of selection procedures: Lower bound
  of false positive rate as a function of inter-rater reliability} {Assessing
  quality of selection procedures: Lower bound of false positive rate as a
  function of inter-rater reliability}.{\BBCQ}
\newblock
\APACjournalVolNumPages{British Journal of Mathematical and Statistical
  Psychology}{00}{}{1-21}.
\PrintBackRefs{\CurrentBib}

\bibitem [\protect \citeauthoryear {%
Berglund%
, Berglund%
\BCBL {}\ \BBA {} Lindvall%
}{%
Berglund%
\ \protect \BOthers {.}}{%
{\protect \APACyear {1976}}%
}]{%
berglund1976dimensionality}
\APACinsertmetastar {%
berglund1976dimensionality}%
\begin{APACrefauthors}%
Berglund, B.%
, Berglund, U.%
\BCBL {}\ \BBA {} Lindvall, T.%
\end{APACrefauthors}%
\unskip\
\newblock
\APACrefYearMonthDay{1976}{}{}.
\newblock
{\BBOQ}\APACrefatitle {Psychological processing of odor mixtures}
  {Psychological processing of odor mixtures}.{\BBCQ}
\newblock
\APACjournalVolNumPages{Psychological Review}{83}{6}{432--441}.
\PrintBackRefs{\CurrentBib}

\bibitem [\protect \citeauthoryear {%
Bierling%
}{%
Bierling%
}{%
{\protect \APACyear {2025}}%
}]{%
bierling2025dataset}
\APACinsertmetastar {%
bierling2025dataset}%
\begin{APACrefauthors}%
Bierling, A\BPBI L\BPBI e\BPBI a.%
\end{APACrefauthors}%
\unskip\
\newblock
\APACrefYearMonthDay{2025}{}{}.
\newblock
{\BBOQ}\APACrefatitle {A dataset of laymen olfactory perception for 74
  mono-molecular odors} {A dataset of laymen olfactory perception for 74
  mono-molecular odors}.{\BBCQ}
\newblock
\APACjournalVolNumPages{Scientific Data}{}{}{}.
\PrintBackRefs{\CurrentBib}

\bibitem [\protect \citeauthoryear {%
Casabianca%
, Lockwood%
\BCBL {}\ \BBA {} Mccaffrey%
}{%
Casabianca%
\ \protect \BOthers {.}}{%
{\protect \APACyear {2015}}%
}]{%
Casabianca_2015}
\APACinsertmetastar {%
Casabianca_2015}%
\begin{APACrefauthors}%
Casabianca, J\BPBI M.%
, Lockwood, J\BPBI R.%
\BCBL {}\ \BBA {} Mccaffrey, D\BPBI F.%
\end{APACrefauthors}%
\unskip\
\newblock
\APACrefYearMonthDay{2015}{}{}.
\newblock
{\BBOQ}\APACrefatitle {Trends in classroom observation scores} {Trends in
  classroom observation scores}.{\BBCQ}
\newblock
\APACjournalVolNumPages{Educational and Psychological
  Measurement}{75}{2}{311-337}.
\PrintBackRefs{\CurrentBib}

\bibitem [\protect \citeauthoryear {%
Cho%
, Shen%
\BCBL {}\ \BBA {} Naveiras%
}{%
Cho%
\ \protect \BOthers {.}}{%
{\protect \APACyear {2019}}%
}]{%
ChoShenNaveiras2019}
\APACinsertmetastar {%
ChoShenNaveiras2019}%
\begin{APACrefauthors}%
Cho, S\BHBI J.%
, Shen, J.%
\BCBL {}\ \BBA {} Naveiras, M.%
\end{APACrefauthors}%
\unskip\
\newblock
\APACrefYearMonthDay{2019}{}{}.
\newblock
{\BBOQ}\APACrefatitle {Multilevel Reliability Measures of Latent Scores within
  an Item Response Theory Framework} {Multilevel reliability measures of latent
  scores within an item response theory framework}.{\BBCQ}
\newblock
\APACjournalVolNumPages{Multivariate Behavioral Research}{54}{6}{856--881}.
\PrintBackRefs{\CurrentBib}

\bibitem [\protect \citeauthoryear {%
Cook%
, Kilgus%
\BCBL {}\ \BBA {} Burns%
}{%
Cook%
\ \protect \BOthers {.}}{%
{\protect \APACyear {2018}}%
}]{%
Cook_2018_precision}
\APACinsertmetastar {%
Cook_2018_precision}%
\begin{APACrefauthors}%
Cook, C\BPBI R.%
, Kilgus, S\BPBI P.%
\BCBL {}\ \BBA {} Burns, M\BPBI K.%
\end{APACrefauthors}%
\unskip\
\newblock
\APACrefYearMonthDay{2018}{}{}.
\newblock
{\BBOQ}\APACrefatitle {Advancing the science and practice of precision
  education to enhance student outcomes} {Advancing the science and practice of
  precision education to enhance student outcomes}.{\BBCQ}
\newblock
\APACjournalVolNumPages{Journal of School Psychology}{66}{}{4-10}.
\PrintBackRefs{\CurrentBib}

\bibitem [\protect \citeauthoryear {%
Dahl%
, Johnson%
\BCBL {}\ \BBA {} M{\"u}ller%
}{%
Dahl%
\ \protect \BOthers {.}}{%
{\protect \APACyear {2022}}%
}]{%
Dahl_salso}
\APACinsertmetastar {%
Dahl_salso}%
\begin{APACrefauthors}%
Dahl, D\BPBI B.%
, Johnson, D\BPBI J.%
\BCBL {}\ \BBA {} M{\"u}ller, P.%
\end{APACrefauthors}%
\unskip\
\newblock
\APACrefYearMonthDay{2022}{}{}.
\newblock
{\BBOQ}\APACrefatitle {Search Algorithms and Loss Functions for {Bayesian}
  Clustering} {Search algorithms and loss functions for {Bayesian}
  clustering}.{\BBCQ}
\newblock
\APACjournalVolNumPages{Journal of Computational and Graphical
  Statistics}{31}{4}{1189--1201}.
\PrintBackRefs{\CurrentBib}

\bibitem [\protect \citeauthoryear {%
{de Valpine}%
\ \protect \BOthers {.}}{%
{de Valpine}%
\ \protect \BOthers {.}}{%
{\protect \APACyear {2026}}%
}]{%
nimble-software2026}
\APACinsertmetastar {%
nimble-software2026}%
\begin{APACrefauthors}%
{de Valpine}, P.%
, Paciorek, C.%
, Turek, D.%
, Michaud, N.%
, Anderson-Bergman, C.%
, Obermeyer, F.%
\BDBL {}Paganin, S.%
\end{APACrefauthors}%
\unskip\
\newblock
\APACrefYearMonthDay{2026}{}{}.
\newblock
{\BBOQ}\APACrefatitle {{NIMBLE}: {MCMC}, Particle Filtering, and Programmable
  Hierarchical Modeling} {{NIMBLE}: {MCMC}, particle filtering, and
  programmable hierarchical modeling}{\BBCQ}\ [\bibcomputersoftwaremanual].
\newblock
\begin{APACrefURL} \url{https://cran.r-project.org/package=nimble}
  \end{APACrefURL}
\newblock
\APACrefnote{{R} package version 1.4.2}
\PrintBackRefs{\CurrentBib}

\bibitem [\protect \citeauthoryear {%
De~Blasi%
\ \protect \BOthers {.}}{%
De~Blasi%
\ \protect \BOthers {.}}{%
{\protect \APACyear {2015}}%
}]{%
DeBlasi2015}
\APACinsertmetastar {%
DeBlasi2015}%
\begin{APACrefauthors}%
De~Blasi, P.%
, Favaro, S.%
, Lijoi, A.%
, Mena, R\BPBI H.%
, Pr{\"u}nster, I.%
\BCBL {}\ \BBA {} Ruggiero, M.%
\end{APACrefauthors}%
\unskip\
\newblock
\APACrefYearMonthDay{2015}{}{}.
\newblock
{\BBOQ}\APACrefatitle {Are Gibbs-Type Priors the Most Natural Generalization of
  the {D}irichlet Process?} {Are gibbs-type priors the most natural
  generalization of the {D}irichlet process?}{\BBCQ}
\newblock
\APACjournalVolNumPages{IEEE Transactions on pattern analysis and machine
  intelligence}{37}{}{}.
\PrintBackRefs{\CurrentBib}

\bibitem [\protect \citeauthoryear {%
DeCarlo%
, Kim%
\BCBL {}\ \BBA {} Johnson%
}{%
DeCarlo%
\ \protect \BOthers {.}}{%
{\protect \APACyear {2011}}%
}]{%
decarlo2011hierarchical}
\APACinsertmetastar {%
decarlo2011hierarchical}%
\begin{APACrefauthors}%
DeCarlo, L\BPBI T.%
, Kim, Y.%
\BCBL {}\ \BBA {} Johnson, M\BPBI S.%
\end{APACrefauthors}%
\unskip\
\newblock
\APACrefYearMonthDay{2011}{}{}.
\newblock
{\BBOQ}\APACrefatitle {A hierarchical rater model for constructed responses,
  with a signal detection rater model} {A hierarchical rater model for
  constructed responses, with a signal detection rater model}.{\BBCQ}
\newblock
\APACjournalVolNumPages{Journal of Educational Measurement}{48}{3}{333--356}.
\PrintBackRefs{\CurrentBib}

\bibitem [\protect \citeauthoryear {%
Eckes%
}{%
Eckes%
}{%
{\protect \APACyear {2017}}%
}]{%
Eckes2017}
\APACinsertmetastar {%
Eckes2017}%
\begin{APACrefauthors}%
Eckes, T.%
\end{APACrefauthors}%
\unskip\
\newblock
\APACrefYearMonthDay{2017}{}{}.
\newblock
{\BBOQ}\APACrefatitle {Rater effects: Advances in item response modeling of
  human ratings -- {Part I}} {Rater effects: Advances in item response modeling
  of human ratings -- {Part I}}.{\BBCQ}
\newblock
\APACjournalVolNumPages{Psychological Test and Assessment
  Modeling}{59}{4}{443--452}.
\PrintBackRefs{\CurrentBib}

\bibitem [\protect \citeauthoryear {%
Erosheva%
, Martinkov{\'a}%
\BCBL {}\ \BBA {} Lee%
}{%
Erosheva%
\ \protect \BOthers {.}}{%
{\protect \APACyear {2021}}%
}]{%
Erosheva2021}
\APACinsertmetastar {%
Erosheva2021}%
\begin{APACrefauthors}%
Erosheva, E\BPBI A.%
, Martinkov{\'a}, P.%
\BCBL {}\ \BBA {} Lee, C\BPBI J.%
\end{APACrefauthors}%
\unskip\
\newblock
\APACrefYearMonthDay{2021}{}{}.
\newblock
{\BBOQ}\APACrefatitle {When Zero May Not Be Zero: A Cautionary Note on the Use
  of Inter-Rater Reliability in Evaluating Grant Peer Review} {When zero may
  not be zero: A cautionary note on the use of inter-rater reliability in
  evaluating grant peer review}.{\BBCQ}
\newblock
\APACjournalVolNumPages{Journal of the Royal Statistical Society Series A:
  Statistics in Society}{184}{3}{904-919}.
\PrintBackRefs{\CurrentBib}

\bibitem [\protect \citeauthoryear {%
Gelman%
\ \protect \BOthers {.}}{%
Gelman%
\ \protect \BOthers {.}}{%
{\protect \APACyear {2013}}%
}]{%
BDA3}
\APACinsertmetastar {%
BDA3}%
\begin{APACrefauthors}%
Gelman, A.%
, Carlin, J.%
, Stern, H.%
, Dunson, D.%
, Vehtari, A.%
\BCBL {}\ \BBA {} Rubin, D.%
\end{APACrefauthors}%
\unskip\
\newblock
\APACrefYear{2013}.
\newblock
\APACrefbtitle {Bayesian Data Analysis, Third Edition} {Bayesian data analysis,
  third edition}.
\newblock
\APACaddressPublisher{}{Taylor \& Francis}.
\PrintBackRefs{\CurrentBib}

\bibitem [\protect \citeauthoryear {%
Gelman%
, Hwang%
\BCBL {}\ \BBA {} Vehtari%
}{%
Gelman%
\ \protect \BOthers {.}}{%
{\protect \APACyear {2014}}%
}]{%
Gelman2014}
\APACinsertmetastar {%
Gelman2014}%
\begin{APACrefauthors}%
Gelman, A.%
, Hwang, J.%
\BCBL {}\ \BBA {} Vehtari, A.%
\end{APACrefauthors}%
\unskip\
\newblock
\APACrefYearMonthDay{2014}{}{}.
\newblock
{\BBOQ}\APACrefatitle {Understanding predictive information criteria for
  {B}ayesian models} {Understanding predictive information criteria for
  {B}ayesian models}.{\BBCQ}
\newblock
\APACjournalVolNumPages{Stat Comput}{24}{}{997-1016}.
\PrintBackRefs{\CurrentBib}

\bibitem [\protect \citeauthoryear {%
Ghosal%
\ \BBA {} van~der Vaart%
}{%
Ghosal%
\ \BBA {} van~der Vaart%
}{%
{\protect \APACyear {2017}}%
}]{%
Ghosal_van_der_Vaart_2017}
\APACinsertmetastar {%
Ghosal_van_der_Vaart_2017}%
\begin{APACrefauthors}%
Ghosal, S.%
\BCBT {}\ \BBA {} van~der Vaart, A.%
\end{APACrefauthors}%
\unskip\
\newblock
\APACrefYear{2017}.
\newblock
\APACrefbtitle {Fundamentals of Nonparametric {B}ayesian Inference}
  {Fundamentals of nonparametric {B}ayesian inference}.
\newblock
\APACaddressPublisher{}{Cambridge University Press}.
\PrintBackRefs{\CurrentBib}

\bibitem [\protect \citeauthoryear {%
Glas%
, Jorgensen%
\BCBL {}\ \BBA {} Ten~Hove%
}{%
Glas%
\ \protect \BOthers {.}}{%
{\protect \APACyear {2024}}%
}]{%
glas2024reducing}
\APACinsertmetastar {%
glas2024reducing}%
\begin{APACrefauthors}%
Glas, C\BPBI A.%
, Jorgensen, T\BPBI D.%
\BCBL {}\ \BBA {} Ten~Hove, D.%
\end{APACrefauthors}%
\unskip\
\newblock
\APACrefYearMonthDay{2024}{}{}.
\newblock
{\BBOQ}\APACrefatitle {Reducing Attenuation Bias in Regression Analyses
  Involving Rating Scale Data via Psychometric Modeling} {Reducing attenuation
  bias in regression analyses involving rating scale data via psychometric
  modeling}.{\BBCQ}
\newblock
\APACjournalVolNumPages{Psychometrika}{89}{1}{42--63}.
\PrintBackRefs{\CurrentBib}

\bibitem [\protect \citeauthoryear {%
Heinzl%
, Fahrmeir%
\BCBL {}\ \BBA {} Kneib%
}{%
Heinzl%
\ \protect \BOthers {.}}{%
{\protect \APACyear {2012}}%
}]{%
heinzl2012}
\APACinsertmetastar {%
heinzl2012}%
\begin{APACrefauthors}%
Heinzl, F.%
, Fahrmeir, L.%
\BCBL {}\ \BBA {} Kneib, T.%
\end{APACrefauthors}%
\unskip\
\newblock
\APACrefYearMonthDay{2012}{}{}.
\newblock
{\BBOQ}\APACrefatitle {Additive mixed models with {D}irichlet process mixture
  and P-spline priors} {Additive mixed models with {D}irichlet process mixture
  and p-spline priors}.{\BBCQ}
\newblock
\APACjournalVolNumPages{Asta advances in statistical analysis}{96}{}{47--68}.
\PrintBackRefs{\CurrentBib}

\bibitem [\protect \citeauthoryear {%
Hjort%
, Holmes%
, M{\"u}ller%
\BCBL {}\ \BBA {} Walker%
}{%
Hjort%
\ \protect \BOthers {.}}{%
{\protect \APACyear {2010}}%
}]{%
Hjort_Holmes_Muller_Walker_2010}
\APACinsertmetastar {%
Hjort_Holmes_Muller_Walker_2010}%
\begin{APACrefauthors}%
Hjort, N.%
, Holmes, C.%
, M{\"u}ller, P.%
\BCBL {}\ \BBA {} Walker, S.%
\end{APACrefauthors}%
\unskip\
\newblock
\APACrefYear{2010}.
\newblock
\APACrefbtitle {Bayesian Nonparametrics} {Bayesian nonparametrics}.
\newblock
\APACaddressPublisher{}{Cambridge University Press}.
\PrintBackRefs{\CurrentBib}

\bibitem [\protect \citeauthoryear {%
Humphry%
\ \BBA {} Heldsinger%
}{%
Humphry%
\ \BBA {} Heldsinger%
}{%
{\protect \APACyear {2020}}%
}]{%
Reliable_Teacher2020}
\APACinsertmetastar {%
Reliable_Teacher2020}%
\begin{APACrefauthors}%
Humphry, S\BPBI M.%
\BCBT {}\ \BBA {} Heldsinger, S.%
\end{APACrefauthors}%
\unskip\
\newblock
\APACrefYearMonthDay{2020}{}{}.
\newblock
{\BBOQ}\APACrefatitle {A Two-Stage Method for Obtaining Reliable Teacher
  Assessments of Writing} {A two-stage method for obtaining reliable teacher
  assessments of writing}.{\BBCQ}
\newblock
\APACjournalVolNumPages{Frontiers in Education}{5}{}{6}.
\PrintBackRefs{\CurrentBib}

\bibitem [\protect \citeauthoryear {%
Ishwaran%
\ \BBA {} James%
}{%
Ishwaran%
\ \BBA {} James%
}{%
{\protect \APACyear {2001}}%
}]{%
Ishwaran2001}
\APACinsertmetastar {%
Ishwaran2001}%
\begin{APACrefauthors}%
Ishwaran, H.%
\BCBT {}\ \BBA {} James, L.%
\end{APACrefauthors}%
\unskip\
\newblock
\APACrefYearMonthDay{2001}{}{}.
\newblock
{\BBOQ}\APACrefatitle {Gibbs Sampling Methods for Stick-Breaking Priors} {Gibbs
  sampling methods for stick-breaking priors}.{\BBCQ}
\newblock
\APACjournalVolNumPages{Journal of the American Statistical
  Association}{96}{}{161-173}.
\PrintBackRefs{\CurrentBib}

\bibitem [\protect \citeauthoryear {%
Jin%
\ \BBA {} Eckes%
}{%
Jin%
\ \BBA {} Eckes%
}{%
{\protect \APACyear {2025}}%
}]{%
Halo_IRT_2025raters}
\APACinsertmetastar {%
Halo_IRT_2025raters}%
\begin{APACrefauthors}%
Jin, K\BHBI Y.%
\BCBT {}\ \BBA {} Eckes, T.%
\end{APACrefauthors}%
\unskip\
\newblock
\APACrefYearMonthDay{2025}{}{}.
\newblock
{\BBOQ}\APACrefatitle {When raters generalize: Examining sources of halo
  effects with mixture {R}asch facets models} {When raters generalize:
  Examining sources of halo effects with mixture {R}asch facets models}.{\BBCQ}
\newblock
\APACjournalVolNumPages{Behavior Research Methods}{57}{5}{149}.
\PrintBackRefs{\CurrentBib}

\bibitem [\protect \citeauthoryear {%
Jin%
\ \BBA {} Eckes%
}{%
Jin%
\ \BBA {} Eckes%
}{%
{\protect \APACyear {2026}}%
}]{%
An_Extended_MFRM}
\APACinsertmetastar {%
An_Extended_MFRM}%
\begin{APACrefauthors}%
Jin, K\BHBI Y.%
\BCBT {}\ \BBA {} Eckes, T.%
\end{APACrefauthors}%
\unskip\
\newblock
\APACrefYearMonthDay{2026}{}{}.
\newblock
{\BBOQ}\APACrefatitle {When Cut Scores Impact Human Ratings: An Extended
  Many-Facet Rasch Model for Category-Specific Rater Severity Shifts} {When cut
  scores impact human ratings: An extended many-facet rasch model for
  category-specific rater severity shifts}.{\BBCQ}
\newblock
\APACjournalVolNumPages{Journal of Educational Measurement}{63}{3}{e70054}.
\PrintBackRefs{\CurrentBib}

\bibitem [\protect \citeauthoryear {%
Lin%
, Tuerlinckx%
\BCBL {}\ \BBA {} Vanbelle%
}{%
Lin%
\ \protect \BOthers {.}}{%
{\protect \APACyear {2025}}%
}]{%
lin2025reliability}
\APACinsertmetastar {%
lin2025reliability}%
\begin{APACrefauthors}%
Lin, T\BHBI Y.%
, Tuerlinckx, F.%
\BCBL {}\ \BBA {} Vanbelle, S.%
\end{APACrefauthors}%
\unskip\
\newblock
\APACrefYearMonthDay{2025}{}{}.
\newblock
{\BBOQ}\APACrefatitle {Reliability for multilevel data: A correlation
  approach.} {Reliability for multilevel data: A correlation approach.}{\BBCQ}
\newblock
\APACjournalVolNumPages{Psychological Methods}{}{}{}.
\PrintBackRefs{\CurrentBib}

\bibitem [\protect \citeauthoryear {%
Linacre%
}{%
Linacre%
}{%
{\protect \APACyear {1989}}%
}]{%
linacre1989many}
\APACinsertmetastar {%
linacre1989many}%
\begin{APACrefauthors}%
Linacre, J\BPBI M.%
\end{APACrefauthors}%
\unskip\
\newblock
\APACrefYear{1989}.
\unskip\
\newblock
\APACrefbtitle {Many-faceted {Rasch} measurement} {Many-faceted {Rasch}
  measurement}\ \APACtypeAddressSchool {\BUPhD}{}{}.
\unskip\
\newblock
\APACaddressSchool {}{The University of Chicago}.
\PrintBackRefs{\CurrentBib}

\bibitem [\protect \citeauthoryear {%
Livingstone%
\ \BBA {} Zanella%
}{%
Livingstone%
\ \BBA {} Zanella%
}{%
{\protect \APACyear {2022}}%
}]{%
LivingstoneZanella2022}
\APACinsertmetastar {%
LivingstoneZanella2022}%
\begin{APACrefauthors}%
Livingstone, S.%
\BCBT {}\ \BBA {} Zanella, G.%
\end{APACrefauthors}%
\unskip\
\newblock
\APACrefYearMonthDay{2022}{}{}.
\newblock
{\BBOQ}\APACrefatitle {The Barker proposal: combining robustness and efficiency
  in gradient-based {MCMC}} {The barker proposal: combining robustness and
  efficiency in gradient-based {MCMC}}.{\BBCQ}
\newblock
\APACjournalVolNumPages{Journal of the Royal Statistical Society Series
  B}{84}{2}{496--523}.
\PrintBackRefs{\CurrentBib}

\bibitem [\protect \citeauthoryear {%
Lunz%
\ \BBA {} Linacre%
}{%
Lunz%
\ \BBA {} Linacre%
}{%
{\protect \APACyear {1998}}%
}]{%
lunz1998measurement}
\APACinsertmetastar {%
lunz1998measurement}%
\begin{APACrefauthors}%
Lunz, M\BPBI E.%
\BCBT {}\ \BBA {} Linacre, J\BPBI M.%
\end{APACrefauthors}%
\unskip\
\newblock
\APACrefYearMonthDay{1998}{}{}.
\newblock
{\BBOQ}\APACrefatitle {Measurement designs using multifacet {Rasch} modeling}
  {Measurement designs using multifacet {Rasch} modeling}.{\BBCQ}
\newblock
\BIn{} \APACrefbtitle {Modern methods for business research} {Modern methods
  for business research}\ (\BPGS\ 47--77).
\PrintBackRefs{\CurrentBib}

\bibitem [\protect \citeauthoryear {%
Lyu%
\ \BBA {} Bolt%
}{%
Lyu%
\ \BBA {} Bolt%
}{%
{\protect \APACyear {2022}}%
}]{%
Bolt2022anchoring}
\APACinsertmetastar {%
Bolt2022anchoring}%
\begin{APACrefauthors}%
Lyu, W.%
\BCBT {}\ \BBA {} Bolt, D\BPBI M.%
\end{APACrefauthors}%
\unskip\
\newblock
\APACrefYearMonthDay{2022}{}{}.
\newblock
{\BBOQ}\APACrefatitle {A psychometric model for respondent-level anchoring on
  self-report rating scale instruments} {A psychometric model for
  respondent-level anchoring on self-report rating scale instruments}.{\BBCQ}
\newblock
\APACjournalVolNumPages{The British Journal of Mathematical and Statistical
  Psychology}{75}{1}{116--135}.
\PrintBackRefs{\CurrentBib}

\bibitem [\protect \citeauthoryear {%
Mariano%
}{%
Mariano%
}{%
{\protect \APACyear {2002}}%
}]{%
mariano2002critique}
\APACinsertmetastar {%
mariano2002critique}%
\begin{APACrefauthors}%
Mariano, L\BPBI T.%
\end{APACrefauthors}%
\unskip\
\newblock
\APACrefYear{2002}.
\unskip\
\newblock
\APACrefbtitle {{Information Accumulation and the Many-Facet {Rasch} Model: A
  Critical Examination}} {{Information Accumulation and the Many-Facet {Rasch}
  Model: A Critical Examination}}\ \APACtypeAddressSchool {\BUPhD}{}{}.
\unskip\
\newblock
\APACaddressSchool {}{RAND Graduate School}.
\unskip\
\newblock
\APACrefnote{Doctoral dissertation discussing limitations of conventional MFRM
  (e.g., information accumulation)}
\PrintBackRefs{\CurrentBib}

\bibitem [\protect \citeauthoryear {%
Mariano%
\ \BBA {} Junker%
}{%
Mariano%
\ \BBA {} Junker%
}{%
{\protect \APACyear {2007}}%
}]{%
Mariano2007covariates}
\APACinsertmetastar {%
Mariano2007covariates}%
\begin{APACrefauthors}%
Mariano, L\BPBI T.%
\BCBT {}\ \BBA {} Junker, B\BPBI W.%
\end{APACrefauthors}%
\unskip\
\newblock
\APACrefYearMonthDay{2007}{}{}.
\newblock
{\BBOQ}\APACrefatitle {Covariates of the Rating Process in Hierarchical Models
  for Multiple Ratings of Test Items} {Covariates of the rating process in
  hierarchical models for multiple ratings of test items}.{\BBCQ}
\newblock
\APACjournalVolNumPages{Journal of Educational and Behavioral
  Statistics}{32}{3}{287--314}.
\PrintBackRefs{\CurrentBib}

\bibitem [\protect \citeauthoryear {%
Martinkov{\'a}%
, Barto{\v{s}}%
\BCBL {}\ \BBA {} Brabec%
}{%
Martinkov{\'a}%
\ \protect \BOthers {.}}{%
{\protect \APACyear {2023}}%
}]{%
Martinkova2023}
\APACinsertmetastar {%
Martinkova2023}%
\begin{APACrefauthors}%
Martinkov{\'a}, P.%
, Barto{\v{s}}, F.%
\BCBL {}\ \BBA {} Brabec, M.%
\end{APACrefauthors}%
\unskip\
\newblock
\APACrefYearMonthDay{2023}{}{}.
\newblock
{\BBOQ}\APACrefatitle {Assessing inter-rater reliability with heterogeneous
  variance components models: Flexible approach accounting for contextual
  variables} {Assessing inter-rater reliability with heterogeneous variance
  components models: Flexible approach accounting for contextual
  variables}.{\BBCQ}
\newblock
\APACjournalVolNumPages{Journal of Educational and Behavioral
  Statistics}{48}{3}{349--383}.
\PrintBackRefs{\CurrentBib}

\bibitem [\protect \citeauthoryear {%
Martinkova%
, Goldhaber%
\BCBL {}\ \BBA {} Erosheva%
}{%
Martinkova%
\ \protect \BOthers {.}}{%
{\protect \APACyear {2018}}%
}]{%
Martinkova2018}
\APACinsertmetastar {%
Martinkova2018}%
\begin{APACrefauthors}%
Martinkova, P.%
, Goldhaber, D.%
\BCBL {}\ \BBA {} Erosheva, E.%
\end{APACrefauthors}%
\unskip\
\newblock
\APACrefYearMonthDay{2018}{}{}.
\newblock
{\BBOQ}\APACrefatitle {Disparities in ratings of internal and external
  applicants: A case for model-based inter-rater reliability} {Disparities in
  ratings of internal and external applicants: A case for model-based
  inter-rater reliability}.{\BBCQ}
\newblock
\APACjournalVolNumPages{PloS One}{13}{}{e0203002}.
\PrintBackRefs{\CurrentBib}

\bibitem [\protect \citeauthoryear {%
Martinkov{\'a}%
\ \BBA {} Hladk{\'a}%
}{%
Martinkov{\'a}%
\ \BBA {} Hladk{\'a}%
}{%
{\protect \APACyear {2023}}%
}]{%
martinkovabook}
\APACinsertmetastar {%
martinkovabook}%
\begin{APACrefauthors}%
Martinkov{\'a}, P.%
\BCBT {}\ \BBA {} Hladk{\'a}, A.%
\end{APACrefauthors}%
\unskip\
\newblock
\APACrefYear{2023}.
\newblock
\APACrefbtitle {Computational aspects of psychometric methods: With {R}}
  {Computational aspects of psychometric methods: With {R}}.
\newblock
\APACaddressPublisher{}{Chapman and Hall/CRC}.
\PrintBackRefs{\CurrentBib}

\bibitem [\protect \citeauthoryear {%
Masters%
}{%
Masters%
}{%
{\protect \APACyear {1982}}%
}]{%
masters1982rasch}
\APACinsertmetastar {%
masters1982rasch}%
\begin{APACrefauthors}%
Masters, G\BPBI N.%
\end{APACrefauthors}%
\unskip\
\newblock
\APACrefYearMonthDay{1982}{}{}.
\newblock
{\BBOQ}\APACrefatitle {A {Rasch} model for partial credit scoring} {A {Rasch}
  model for partial credit scoring}.{\BBCQ}
\newblock
\APACjournalVolNumPages{Psychometrika}{47}{2}{149--174}.
\PrintBackRefs{\CurrentBib}

\bibitem [\protect \citeauthoryear {%
Meil{\u{a}}%
}{%
Meil{\u{a}}%
}{%
{\protect \APACyear {2007}}%
}]{%
Meila_2007}
\APACinsertmetastar {%
Meila_2007}%
\begin{APACrefauthors}%
Meil{\u{a}}, M.%
\end{APACrefauthors}%
\unskip\
\newblock
\APACrefYearMonthDay{2007}{}{}.
\newblock
{\BBOQ}\APACrefatitle {Comparing clusterings---an information based distance}
  {Comparing clusterings---an information based distance}.{\BBCQ}
\newblock
\APACjournalVolNumPages{Journal of Multivariate Analysis}{98}{5}{873-895}.
\PrintBackRefs{\CurrentBib}

\bibitem [\protect \citeauthoryear {%
Mignemi%
, Calcagn{\`i}%
, Spoto%
\BCBL {}\ \BBA {} Manolopoulou%
}{%
Mignemi%
\ \protect \BOthers {.}}{%
{\protect \APACyear {2024}}%
}]{%
Mignemi2024}
\APACinsertmetastar {%
Mignemi2024}%
\begin{APACrefauthors}%
Mignemi, G.%
, Calcagn{\`i}, A.%
, Spoto, A.%
\BCBL {}\ \BBA {} Manolopoulou, I.%
\end{APACrefauthors}%
\unskip\
\newblock
\APACrefYearMonthDay{2024}{}{}.
\newblock
{\BBOQ}\APACrefatitle {Mixture polarization in inter-rater agreement analysis:
  A {B}ayesian nonparametric index} {Mixture polarization in inter-rater
  agreement analysis: A {B}ayesian nonparametric index}.{\BBCQ}
\newblock
\APACjournalVolNumPages{Statistical Methods \& Applications}{33}{}{325-355}.
\PrintBackRefs{\CurrentBib}

\bibitem [\protect \citeauthoryear {%
Mignemi%
, Chen%
\BCBL {}\ \BBA {} Moustaki%
}{%
Mignemi%
\ \protect \BOthers {.}}{%
{\protect \APACyear {2025}}%
}]{%
Mignemi_Chen_Moustaki_2025}
\APACinsertmetastar {%
Mignemi_Chen_Moustaki_2025}%
\begin{APACrefauthors}%
Mignemi, G.%
, Chen, Y.%
\BCBL {}\ \BBA {} Moustaki, I.%
\end{APACrefauthors}%
\unskip\
\newblock
\APACrefYearMonthDay{2025}{}{}.
\newblock
{\BBOQ}\APACrefatitle {Unfolding the Network of Peer Grades: A Latent Variable
  Approach} {Unfolding the network of peer grades: A latent variable
  approach}.{\BBCQ}
\newblock
\APACjournalVolNumPages{Psychometrika}{}{}{1--22}.
\PrintBackRefs{\CurrentBib}

\bibitem [\protect \citeauthoryear {%
Mignemi%
\ \BBA {} Manolopoulou%
}{%
Mignemi%
\ \BBA {} Manolopoulou%
}{%
{\protect \APACyear {2025}}%
}]{%
mignemi2025BNP}
\APACinsertmetastar {%
mignemi2025BNP}%
\begin{APACrefauthors}%
Mignemi, G.%
\BCBT {}\ \BBA {} Manolopoulou, I.%
\end{APACrefauthors}%
\unskip\
\newblock
\APACrefYearMonthDay{2025}{}{}.
\newblock
{\BBOQ}\APACrefatitle {Bayesian Nonparametric Models for Multiple Raters: A
  General Statistical Framework} {Bayesian nonparametric models for multiple
  raters: A general statistical framework}.{\BBCQ}
\newblock
\APACjournalVolNumPages{Psychometrika}{90}{4}{}.
\PrintBackRefs{\CurrentBib}

\bibitem [\protect \citeauthoryear {%
Molenaar%
, Uluman%
, Tav{\c{s}}anc{\i}l%
\BCBL {}\ \BBA {} De~Boeck%
}{%
Molenaar%
\ \protect \BOthers {.}}{%
{\protect \APACyear {2021}}%
}]{%
Molenaar2021}
\APACinsertmetastar {%
Molenaar2021}%
\begin{APACrefauthors}%
Molenaar, D.%
, Uluman, M.%
, Tav{\c{s}}anc{\i}l, E.%
\BCBL {}\ \BBA {} De~Boeck, P.%
\end{APACrefauthors}%
\unskip\
\newblock
\APACrefYearMonthDay{2021}{01}{}.
\newblock
{\BBOQ}\APACrefatitle {The Hierarchical Rater Thresholds Model for Multiple
  Raters and Multiple Items} {The hierarchical rater thresholds model for
  multiple raters and multiple items}.{\BBCQ}
\newblock
\APACjournalVolNumPages{Open Education Studies}{3}{}{33-48}.
\PrintBackRefs{\CurrentBib}

\bibitem [\protect \citeauthoryear {%
Muraki%
}{%
Muraki%
}{%
{\protect \APACyear {1992}}%
}]{%
Muraki_1992}
\APACinsertmetastar {%
Muraki_1992}%
\begin{APACrefauthors}%
Muraki, E.%
\end{APACrefauthors}%
\unskip\
\newblock
\APACrefYearMonthDay{1992}{}{}.
\newblock
{\BBOQ}\APACrefatitle {A Generalized Partial Credit Model: Application of an EM
  Algorithm} {A generalized partial credit model: Application of an em
  algorithm}.{\BBCQ}
\newblock
\APACjournalVolNumPages{Applied Psychological Measurement}{16}{2}{159-176}.
\PrintBackRefs{\CurrentBib}

\bibitem [\protect \citeauthoryear {%
Nelson%
\ \BBA {} Edwards%
}{%
Nelson%
\ \BBA {} Edwards%
}{%
{\protect \APACyear {2015}}%
}]{%
Nelson2015}
\APACinsertmetastar {%
Nelson2015}%
\begin{APACrefauthors}%
Nelson, K.%
\BCBT {}\ \BBA {} Edwards, D.%
\end{APACrefauthors}%
\unskip\
\newblock
\APACrefYearMonthDay{2015}{}{}.
\newblock
{\BBOQ}\APACrefatitle {Measures of Agreement Between Many Raters for Ordinal
  Classifications} {Measures of agreement between many raters for ordinal
  classifications}.{\BBCQ}
\newblock
\APACjournalVolNumPages{Statistics in medicine}{34}{}{3116--3132}.
\PrintBackRefs{\CurrentBib}

\bibitem [\protect \citeauthoryear {%
Patz%
\ \BBA {} Junker%
}{%
Patz%
\ \BBA {} Junker%
}{%
{\protect \APACyear {1999}}%
}]{%
patz1999applications}
\APACinsertmetastar {%
patz1999applications}%
\begin{APACrefauthors}%
Patz, R\BPBI J.%
\BCBT {}\ \BBA {} Junker, B\BPBI W.%
\end{APACrefauthors}%
\unskip\
\newblock
\APACrefYearMonthDay{1999}{}{}.
\newblock
{\BBOQ}\APACrefatitle {Applications and extensions of {MCMC} in {IRT}: Multiple
  item types, missing data, and rated responses} {Applications and extensions
  of {MCMC} in {IRT}: Multiple item types, missing data, and rated
  responses}.{\BBCQ}
\newblock
\APACjournalVolNumPages{Journal of Educational and Behavioral
  Statistics}{24}{4}{342--366}.
\PrintBackRefs{\CurrentBib}

\bibitem [\protect \citeauthoryear {%
Patz%
, Junker%
, Johnson%
\BCBL {}\ \BBA {} Mariano%
}{%
Patz%
\ \protect \BOthers {.}}{%
{\protect \APACyear {2002}}%
}]{%
Patz2002}
\APACinsertmetastar {%
Patz2002}%
\begin{APACrefauthors}%
Patz, R\BPBI J.%
, Junker, B\BPBI W.%
, Johnson, M\BPBI S.%
\BCBL {}\ \BBA {} Mariano, L\BPBI T.%
\end{APACrefauthors}%
\unskip\
\newblock
\APACrefYearMonthDay{2002}{}{}.
\newblock
{\BBOQ}\APACrefatitle {The Hierarchical Rater Model for Rated Test Items and
  its Application to Large-Scale Educational Assessment Data} {The hierarchical
  rater model for rated test items and its application to large-scale
  educational assessment data}.{\BBCQ}
\newblock
\APACjournalVolNumPages{{Journal of Educational and Behavioral
  Statistics}}{27}{4}{341-384}.
\PrintBackRefs{\CurrentBib}

\bibitem [\protect \citeauthoryear {%
Peng%
, Cui%
, Chung%
\BCBL {}\ \BBA {} Li%
}{%
Peng%
\ \protect \BOthers {.}}{%
{\protect \APACyear {2019}}%
}]{%
peng_2019_multifacet}
\APACinsertmetastar {%
peng_2019_multifacet}%
\begin{APACrefauthors}%
Peng, L.%
, Cui, G.%
, Chung, Y.%
\BCBL {}\ \BBA {} Li, C.%
\end{APACrefauthors}%
\unskip\
\newblock
\APACrefYearMonthDay{2019}{}{}.
\newblock
{\BBOQ}\APACrefatitle {A multi-facet item response theory approach to improve
  customer satisfaction using online product ratings} {A multi-facet item
  response theory approach to improve customer satisfaction using online
  product ratings}.{\BBCQ}
\newblock
\APACjournalVolNumPages{Journal of the Academy of Marketing
  Science}{47}{}{960--976}.
\PrintBackRefs{\CurrentBib}

\bibitem [\protect \citeauthoryear {%
Pitman%
\ \BBA {} Yor%
}{%
Pitman%
\ \BBA {} Yor%
}{%
{\protect \APACyear {1997}}%
}]{%
Pitman-Yor}
\APACinsertmetastar {%
Pitman-Yor}%
\begin{APACrefauthors}%
Pitman, J.%
\BCBT {}\ \BBA {} Yor, M.%
\end{APACrefauthors}%
\unskip\
\newblock
\APACrefYearMonthDay{1997}{}{}.
\newblock
{\BBOQ}\APACrefatitle {{The two-parameter Poisson-{D}irichlet distribution
  derived from a stable subordinator}} {{The two-parameter Poisson-{D}irichlet
  distribution derived from a stable subordinator}}.{\BBCQ}
\newblock
\APACjournalVolNumPages{The Annals of Probability}{25}{2}{855 -- 900}.
\PrintBackRefs{\CurrentBib}

\bibitem [\protect \citeauthoryear {%
Plummer%
, Best%
, Cowles%
\BCBL {}\ \BBA {} Vines%
}{%
Plummer%
\ \protect \BOthers {.}}{%
{\protect \APACyear {2006}}%
}]{%
CODA}
\APACinsertmetastar {%
CODA}%
\begin{APACrefauthors}%
Plummer, M.%
, Best, N.%
, Cowles, K.%
\BCBL {}\ \BBA {} Vines, K.%
\end{APACrefauthors}%
\unskip\
\newblock
\APACrefYearMonthDay{2006}{}{}.
\newblock
{\BBOQ}\APACrefatitle {{CODA}: Convergence Diagnosis and Output Analysis for
  {MCMC}} {{CODA}: Convergence diagnosis and output analysis for
  {MCMC}}.{\BBCQ}
\newblock
\APACjournalVolNumPages{R News}{6}{1}{7--11}.
\PrintBackRefs{\CurrentBib}

\bibitem [\protect \citeauthoryear {%
Polat%
, Turhan%
\BCBL {}\ \BBA {} Toraman%
}{%
Polat%
\ \protect \BOthers {.}}{%
{\protect \APACyear {2022}}%
}]{%
polat_2022_multifacet}
\APACinsertmetastar {%
polat_2022_multifacet}%
\begin{APACrefauthors}%
Polat, M.%
, Turhan, N\BPBI S.%
\BCBL {}\ \BBA {} Toraman, {\c{C}}.%
\end{APACrefauthors}%
\unskip\
\newblock
\APACrefYearMonthDay{2022}{}{}.
\newblock
{\BBOQ}\APACrefatitle {Comparison of classical test theory vs. multi-facet
  {Rasch} theory in writing assessment} {Comparison of classical test theory
  vs. multi-facet {Rasch} theory in writing assessment}.{\BBCQ}
\newblock
\APACjournalVolNumPages{Pegem Journal of Education and
  Instruction}{12}{2}{213--225}.
\PrintBackRefs{\CurrentBib}

\bibitem [\protect \citeauthoryear {%
Secundo%
\ \protect \BOthers {.}}{%
Secundo%
\ \protect \BOthers {.}}{%
{\protect \APACyear {2015}}%
}]{%
secundo2015individual}
\APACinsertmetastar {%
secundo2015individual}%
\begin{APACrefauthors}%
Secundo, L.%
, Snitz, K.%
, Weiss, K.%
, Weiss, L.%
, Ben-Ami, R.%
, Beidhaoui, S.%
\BDBL {}Sobel, N.%
\end{APACrefauthors}%
\unskip\
\newblock
\APACrefYearMonthDay{2015}{}{}.
\newblock
{\BBOQ}\APACrefatitle {Individual olfactory perception reveals meaningful
  variation in histology and genomics} {Individual olfactory perception reveals
  meaningful variation in histology and genomics}.{\BBCQ}
\newblock
\APACjournalVolNumPages{Proceedings of the National Academy of
  Sciences}{112}{28}{8750--8755}.
\PrintBackRefs{\CurrentBib}

\bibitem [\protect \citeauthoryear {%
Sethuraman%
}{%
Sethuraman%
}{%
{\protect \APACyear {1994}}%
}]{%
Sethuraman}
\APACinsertmetastar {%
Sethuraman}%
\begin{APACrefauthors}%
Sethuraman, J.%
\end{APACrefauthors}%
\unskip\
\newblock
\APACrefYearMonthDay{1994}{}{}.
\newblock
{\BBOQ}\APACrefatitle {A constructive definition of {D}irichlet Priors} {A
  constructive definition of {D}irichlet priors}.{\BBCQ}
\newblock
\APACjournalVolNumPages{Statistica Sinica}{4}{2}{639--650}.
\PrintBackRefs{\CurrentBib}

\bibitem [\protect \citeauthoryear {%
ten Hove%
, Jorgensen%
\BCBL {}\ \BBA {} van~der Ark%
}{%
ten Hove%
\ \protect \BOthers {.}}{%
{\protect \APACyear {2021}}%
}]{%
tenHove2021}
\APACinsertmetastar {%
tenHove2021}%
\begin{APACrefauthors}%
ten Hove, D.%
, Jorgensen, T\BPBI D.%
\BCBL {}\ \BBA {} van~der Ark, L\BPBI A.%
\end{APACrefauthors}%
\unskip\
\newblock
\APACrefYearMonthDay{2021}{}{}.
\newblock
{\BBOQ}\APACrefatitle {Interrater reliability for multilevel data: A
  generalizability theory approach} {Interrater reliability for multilevel
  data: A generalizability theory approach}.{\BBCQ}
\newblock
\APACjournalVolNumPages{Psychological Methods}{27}{4}{650-666}.
\PrintBackRefs{\CurrentBib}

\bibitem [\protect \citeauthoryear {%
Tversky%
\ \BBA {} Kahneman%
}{%
Tversky%
\ \BBA {} Kahneman%
}{%
{\protect \APACyear {1974}}%
}]{%
tversky1974judgment}
\APACinsertmetastar {%
tversky1974judgment}%
\begin{APACrefauthors}%
Tversky, A.%
\BCBT {}\ \BBA {} Kahneman, D.%
\end{APACrefauthors}%
\unskip\
\newblock
\APACrefYearMonthDay{1974}{}{}.
\newblock
{\BBOQ}\APACrefatitle {Judgment under uncertainty: Heuristics and biases}
  {Judgment under uncertainty: Heuristics and biases}.{\BBCQ}
\newblock
\APACjournalVolNumPages{Science}{185}{4157}{1124--1131}.
\PrintBackRefs{\CurrentBib}

\bibitem [\protect \citeauthoryear {%
Uang%
}{%
Uang%
}{%
{\protect \APACyear {2025}}%
}]{%
Huang2025rater}
\APACinsertmetastar {%
Huang2025rater}%
\begin{APACrefauthors}%
Uang, H\BPBI Y.%
\end{APACrefauthors}%
\unskip\
\newblock
\APACrefYearMonthDay{2025}{}{}.
\newblock
{\BBOQ}\APACrefatitle {Understanding Rater Cognition in Performance Assessment:
  A Mixed {IRTree} Approach} {Understanding rater cognition in performance
  assessment: A mixed {IRTree} approach}.{\BBCQ}
\newblock
\APACjournalVolNumPages{Applied Psychological Measurement}{49}{7}{}.
\PrintBackRefs{\CurrentBib}

\bibitem [\protect \citeauthoryear {%
Uto%
}{%
Uto%
}{%
{\protect \APACyear {2022}}%
}]{%
Uto_2022}
\APACinsertmetastar {%
Uto_2022}%
\begin{APACrefauthors}%
Uto, M.%
\end{APACrefauthors}%
\unskip\
\newblock
\APACrefYearMonthDay{2022}{}{}.
\newblock
{\BBOQ}\APACrefatitle {{A {B}ayesian many-facet {Rasch} model with Markov
  modeling for rater severity drift}} {{A {B}ayesian many-facet {Rasch} model
  with Markov modeling for rater severity drift}}.{\BBCQ}
\newblock
\APACjournalVolNumPages{Behavior Research Methods}{}{}{3910--3928}.
\PrintBackRefs{\CurrentBib}

\bibitem [\protect \citeauthoryear {%
Uto%
\ \BBA {} Okano%
}{%
Uto%
\ \BBA {} Okano%
}{%
{\protect \APACyear {2021}}%
}]{%
Uto_2021_AI}
\APACinsertmetastar {%
Uto_2021_AI}%
\begin{APACrefauthors}%
Uto, M.%
\BCBT {}\ \BBA {} Okano, M.%
\end{APACrefauthors}%
\unskip\
\newblock
\APACrefYearMonthDay{2021}{}{}.
\newblock
{\BBOQ}\APACrefatitle {Robust neural automated essay scoring using the item
  response theory} {Robust neural automated essay scoring using the item
  response theory}.{\BBCQ}
\newblock
\APACjournalVolNumPages{IEEE Transactions on Learning
  Technologies}{14}{3}{355--369}.
\PrintBackRefs{\CurrentBib}

\bibitem [\protect \citeauthoryear {%
Uto%
, Tsuruta%
, Araki%
\BCBL {}\ \BBA {} Ueno%
}{%
Uto%
\ \protect \BOthers {.}}{%
{\protect \APACyear {2024}}%
}]{%
Uto_24}
\APACinsertmetastar {%
Uto_24}%
\begin{APACrefauthors}%
Uto, M.%
, Tsuruta, J.%
, Araki, K.%
\BCBL {}\ \BBA {} Ueno, M.%
\end{APACrefauthors}%
\unskip\
\newblock
\APACrefYearMonthDay{2024}{}{}.
\newblock
{\BBOQ}\APACrefatitle {Item response theory model highlighting rating scale of
  a rubric and rater--rubric interaction in objective structured clinical
  examination} {Item response theory model highlighting rating scale of a
  rubric and rater--rubric interaction in objective structured clinical
  examination}.{\BBCQ}
\newblock
\APACjournalVolNumPages{PLOS ONE}{19}{9}{1-23}.
\PrintBackRefs{\CurrentBib}

\bibitem [\protect \citeauthoryear {%
Uto%
\ \BBA {} Ueno%
}{%
Uto%
\ \BBA {} Ueno%
}{%
{\protect \APACyear {2016}}%
}]{%
Uto_2016}
\APACinsertmetastar {%
Uto_2016}%
\begin{APACrefauthors}%
Uto, M.%
\BCBT {}\ \BBA {} Ueno, M.%
\end{APACrefauthors}%
\unskip\
\newblock
\APACrefYearMonthDay{2016}{}{}.
\newblock
{\BBOQ}\APACrefatitle {Item Response Theory for Peer Assessment} {Item response
  theory for peer assessment}.{\BBCQ}
\newblock
\APACjournalVolNumPages{IEEE Transactions on Learning
  Technologies}{9}{}{157-170}.
\PrintBackRefs{\CurrentBib}

\bibitem [\protect \citeauthoryear {%
Uto%
\ \BBA {} Ueno%
}{%
Uto%
\ \BBA {} Ueno%
}{%
{\protect \APACyear {2020}}%
}]{%
Uto_2020}
\APACinsertmetastar {%
Uto_2020}%
\begin{APACrefauthors}%
Uto, M.%
\BCBT {}\ \BBA {} Ueno, M.%
\end{APACrefauthors}%
\unskip\
\newblock
\APACrefYearMonthDay{2020}{}{}.
\newblock
{\BBOQ}\APACrefatitle {A generalized many-facet {Rasch} model and its
  {Bayesian} estimation using {Hamiltonian Monte Carlo}} {A generalized
  many-facet {Rasch} model and its {Bayesian} estimation using {Hamiltonian
  Monte Carlo}}.{\BBCQ}
\newblock
\APACjournalVolNumPages{Behaviormetrika}{47}{}{1-28}.
\PrintBackRefs{\CurrentBib}

\bibitem [\protect \citeauthoryear {%
van~der Vaart%
}{%
van~der Vaart%
}{%
{\protect \APACyear {1998}}%
}]{%
vanderVaart1998}
\APACinsertmetastar {%
vanderVaart1998}%
\begin{APACrefauthors}%
van~der Vaart, A\BPBI W.%
\end{APACrefauthors}%
\unskip\
\newblock
\APACrefYear{1998}.
\newblock
\APACrefbtitle {Asymptotic Statistics} {Asymptotic statistics}.
\newblock
\APACaddressPublisher{}{Cambridge University Press}.
\PrintBackRefs{\CurrentBib}

\bibitem [\protect \citeauthoryear {%
Verhelst%
\ \BBA {} Verstralen%
}{%
Verhelst%
\ \BBA {} Verstralen%
}{%
{\protect \APACyear {2001}}%
}]{%
Verhelst2001}
\APACinsertmetastar {%
Verhelst2001}%
\begin{APACrefauthors}%
Verhelst, N\BPBI D.%
\BCBT {}\ \BBA {} Verstralen, H\BPBI H\BPBI F\BPBI M.%
\end{APACrefauthors}%
\unskip\
\newblock
\APACrefYearMonthDay{2001}{}{}.
\newblock
{\BBOQ}\APACrefatitle {An IRT Model for Multiple Raters} {An irt model for
  multiple raters}.{\BBCQ}
\newblock
\BIn{} A.~Boomsma, M\BPBI A\BPBI J.~van Duijn\BCBL {}\ \BBA {} T\BPBI A\BPBI
  B.~Snijders\ (\BEDS), \APACrefbtitle {Essays on Item Response Theory} {Essays
  on item response theory}\ (\BPGS\ 89--108).
\newblock
\APACaddressPublisher{New York, NY}{Springer New York}.
\PrintBackRefs{\CurrentBib}

\bibitem [\protect \citeauthoryear {%
Wade%
\ \BBA {} Ghahramani%
}{%
Wade%
\ \BBA {} Ghahramani%
}{%
{\protect \APACyear {2018}}%
}]{%
Wade&Ghahramani_2018}
\APACinsertmetastar {%
Wade&Ghahramani_2018}%
\begin{APACrefauthors}%
Wade, S.%
\BCBT {}\ \BBA {} Ghahramani, Z.%
\end{APACrefauthors}%
\unskip\
\newblock
\APACrefYearMonthDay{2018}{}{}.
\newblock
{\BBOQ}\APACrefatitle {Bayesian Cluster Analysis: Point Estimation and Credible
  Balls (with Discussion)} {Bayesian cluster analysis: Point estimation and
  credible balls (with discussion)}.{\BBCQ}
\newblock
\APACjournalVolNumPages{Bayesian Analysis}{13}{2}{559 -- 626}.
\PrintBackRefs{\CurrentBib}

\bibitem [\protect \citeauthoryear {%
Wang%
\ \BBA {} Engelhard%
}{%
Wang%
\ \BBA {} Engelhard%
}{%
{\protect \APACyear {2019}}%
}]{%
wang2019}
\APACinsertmetastar {%
wang2019}%
\begin{APACrefauthors}%
Wang, J.%
\BCBT {}\ \BBA {} Engelhard, G.%
\end{APACrefauthors}%
\unskip\
\newblock
\APACrefYearMonthDay{2019}{}{}.
\newblock
{\BBOQ}\APACrefatitle {Conceptualizing rater judgments and rating processes for
  rater-mediated assessments} {Conceptualizing rater judgments and rating
  processes for rater-mediated assessments}.{\BBCQ}
\newblock
\APACjournalVolNumPages{Journal of Educational Measurement}{56}{3}{582--609}.
\PrintBackRefs{\CurrentBib}

\bibitem [\protect \citeauthoryear {%
Werts%
, Linn%
\BCBL {}\ \BBA {} J{\"o}reskog%
}{%
Werts%
\ \protect \BOthers {.}}{%
{\protect \APACyear {1974}}%
}]{%
Joreskog_74}
\APACinsertmetastar {%
Joreskog_74}%
\begin{APACrefauthors}%
Werts, C.%
, Linn, R.%
\BCBL {}\ \BBA {} J{\"o}reskog, K.%
\end{APACrefauthors}%
\unskip\
\newblock
\APACrefYearMonthDay{1974}{}{}.
\newblock
{\BBOQ}\APACrefatitle {Intraclass Reliability Estimates: Testing Structural
  Assumptions} {Intraclass reliability estimates: Testing structural
  assumptions}.{\BBCQ}
\newblock
\APACjournalVolNumPages{Educational and Psychological
  Measurement}{34}{1}{25-33}.
\PrintBackRefs{\CurrentBib}

\end{thebibliography}

\clearpage
\appendix
\setcounter{section}{0}
\setcounter{figure}{0}
\setcounter{table}{0}
\setcounter{equation}{0}
\begin{center}\Large\bfseries Appendix\end{center}
\makeatletter\def\@extra@b@citeb{@App}\makeatother

This appendix complements the main text. Appendix~\ref{Posterior_App} provides the details of posterior inference (prior elicitation, joint posterior, samplers, and post-processing). Appendix~\ref{sim} reports the simulation study on parameter recovery, and Appendix~\ref{application} presents the application to the OSCE clinical performance assessment. Appendix~\ref{A_P} provides the proof of Proposition~\ref{prop_1} and identifiability results; Appendix~\ref{sec:naive} compares naive and model-based latent trait estimates; Appendix~\ref{Additional_fig_tab} collects additional figures and tables from the OSCE and Olfactory Assessment applications; and Appendix~\ref{sec:extensions} describes further extensions of the model. Additional robustness checks, including MCMC diagnostics and sensitivity analyses to the choice of clustering prior and rating design, are reported separately in the Supplementary Materials.

\renewcommand{\thesection}{\Alph{section}}
\numberwithin{equation}{section}
\numberwithin{figure}{section}
\numberwithin{table}{section}
\label{Appendices}

\section{Posterior Inference}\label{Posterior_App}

Posterior inference is performed via Markov Chain Monte Carlo (MCMC). In this section, we first outline our prior specification (Section~\ref{sec:prior_elicitation}), followed by a detailed description of our posterior computation strategy  (Section~\ref{sec:posterior_computation}). We implemented the MCMC algorithm using the NIMBLE software package, making several specific modifications to the default sampler configuration to address strong posterior dependencies and improve mixing. 
Details on the MCMC samplers used for each parameter are reported in Table~\ref{tab:samplers} of the main text.

\subsection{Prior Elicitation}\label{sec:prior_elicitation}
To complete the model specification and enable posterior inference, we must assign either a fixed value or a prior distribution to the parameters of the measurement and structural models. We place a prior over the parameter vectors $\boldsymbol{S}=[\boldsymbol{\varphi}^2,\boldsymbol{\beta,\lambda,\gamma,\omega}^2,\boldsymbol{\xi,\tau,\rho,a}]^T$ and assign the latent trait $\theta$ a nonparametric mixture prior as introduced in Section~\ref{Structural:subjects} governed by a concentration parameter $\alpha$ and a base measure $F_0$. Several specifications may be considered for the base measure \citepApp{Ghosal_van_der_Vaart_2017, DeBlasi2015, Hjort_Holmes_Muller_Walker_2010}, usually assumed to be a suitable diffuse probability function. We let $F_0$ take a parametric form, with hyperparameters fixed in advance. We specify conjugate priors for most model parameters for analytical tractability. We let the respective hyperparameters induce a weakly informative prior distribution over the model parameters, preventing computational instability and overfitting while preserving flexibility over reasonable values. Throughout, Gamma distributions are parameterized by shape and rate, and inverse-Gamma distributions by shape and scale. 

\paragraph{Measurement Model.}
We specify the following priors over the measurement model parameters across items $i=1,\dots, I$ and raters $\forall r\in \mathcal{R}$:
    \begin{equation*}
        log(\varphi_i^2) \overset{\mathrm{iid}}{\sim} N(-0.5,2), \quad \quad \lambda_i \overset{\mathrm{iid}}{\sim} N(0,2),  \quad \quad \beta_{ir} \overset{\mathrm{iid}}{\sim} N(0,1), \quad \quad   \gamma_{ik} \overset{\mathrm{iid}}{\sim} N(0,1.5),
\end{equation*}
    \begin{equation*}
       log(\omega_r^2) \overset{\mathrm{iid}}{\sim} N(-0.5,2), \quad \quad   
      \xi_r \overset{\mathrm{iid}}{\sim} N(0,2), \quad \quad
        \rho_r \overset{\mathrm{iid}}{\sim} N(0,1), \quad \quad \tau_{rk}  \overset{\mathrm{iid}}{\sim} N(0,0.5),
\end{equation*}
    for $k=1,\dots,K$. 
The log-normal priors on $\varphi_i^2$ and $\omega_r^2$ favor values below one while leaving enough mass in the right tail for larger heterogeneity. 
As in \citeApp{Uto_24}, the prior variances of $\beta_{ir}$ and $\rho_r$ are fixed to one to set the scale of these parameters. For the category steps, we assign a larger prior variance to $\gamma_{ik}$ than to $\tau_{rk}$ since differences in category use are expected to be driven mainly by the semantic content of the items, with $\tau_{rk}$ capturing only residual rater-specific deviations. Finally, the priors on $\lambda_i$ and $\xi_r$ are weakly informative, centered on no effect, and allow discrimination parameters to vary over a wide range on the log scale.

We place a Gaussian Markov random field smoothness prior over the vector of basis coefficients $\boldsymbol{a}=\{a_d\}_{d=1}^D$: 
\begin{equation*}
    \boldsymbol{a} \sim N(\boldsymbol{0},\;(b\;\boldsymbol{\Omega})^{-1})
\end{equation*}
where $\boldsymbol{\Omega}=\boldsymbol{D}^T\boldsymbol{D}$ is the $D\times D$ second-order difference matrix for adjacent B-spline coefficients, and $b>0$ is the precision hyperparameter controlling the smoothness of the spline. 
We place a weakly-informative prior over $b$, $b \sim Ga(2, 0.05)$, favoring smoother splines by default while letting the data inform the actual level of smoothness. We note that this formulation implies a log-penalty $\boldsymbol{a}^t \boldsymbol{\Omega}\boldsymbol{a}$ that corresponds to the penalty term introduced by \citeApp{EilersMarx1996} in a frequentist penalized likelihood setting \citepApp{heinzl2012}.

\paragraph{Structural Model.}
Recall from Section~\ref{Structural:subjects} that subjects' latent traits are modeled as $\theta_j \mid G \overset{\mathrm{iid}}{\sim} G$, $j=1,\dots,N$, where $G \sim DP(\alpha, F_0)$ (or $PYP(d,\alpha,F_0)$), with component-specific parameters $\theta_c^{**}=(\mu_c,\sigma_c^2)$, $\mu_c \overset{\mathrm{iid}}{\sim} N(0,s^2_\mu)$ and $\sigma_c^2 \overset{\mathrm{iid}}{\sim} IGa(\nu_1,\nu_2)$, jointly defining the base measure $F_0$. The priors over the concentration parameter $\alpha$ and the hyperparameters of $F_0$ control the posterior distribution over clustering \citepApp{BDA3}: small values of the variance component of $F_0$ favor allocating subjects to different clusters, whereas larger values favor allocating subjects to the same cluster. We set $s^2_\mu=2$ and $(\nu_1,\nu_2)=(2.01,1.01)$, that is $F_0=N(0,2)\,IGa(2.01,1.01)$, and specify a prior $\alpha \sim Ga(2,4)$, following \citeApp{Escobar_95}. Although the discount parameter $d$ of the \gls{pyp} could also be estimated from the data, we fix it in advance, setting $d=0$ for the \gls{dp} and $d=0.3$ or $d=0.6$ for the \gls{pyp} implementations considered in this work.

\subsection{Posterior Computation}\label{sec:posterior_computation}

Posterior computation proceeds via MCMC, using standard full conditional or Metropolis--Hastings updates for each parameter block, with several modifications tailored to our model. Conjugacy is exploited for normal--normal and normal--inverse-Gamma blocks, while adaptive RWMH is used when conjugacy is unavailable (e.g., for P-spline coefficients and nonlinear effects); see Table~\ref{tab:samplers} in Section~\ref{posterior} for details. For the structural model, we use a finite approximation to the \gls{dp} or \gls{pyp} prior, detailed below. We obtain posterior samples via a custom MCMC algorithm implemented in NIMBLE \citepApp{nimble-software2026}, with modifications to the default sampler configuration to address strong dependencies and improve mixing for the \gls{degmfr} model, as highlighted below. The complete code is shared in the Supplementary Materials. 

\paragraph{Finite Stick-breaking Approximation.}
The recursive generation defined in Eq.~(\ref{Infinite_Mixture}) implies that the expected mixture weights generally decrease with the component index $c$. 
Given this property of the weights, in practical applications, the infinite sequences (\ref{Infinite_Mixture}) are truncated at sufficiently large values of $C \in \mathbb{N}$: 
\begin{equation}\label{Trunc1}
     G = \sum_{c=1}^{C} \pi_{c} N(\mu_c, \sigma^2_c). 
\end{equation}
We use this finite stick-breaking approximation proposed by \citetApp{Ishwaran2001}, setting $V_{C}=1$, and discarding terms $C+1,\dots,\infty$, for $G$. This conditional approach samples cluster assignments, weights, and component parameters directly via MCMC, rather than marginalizing over the mixture distribution as in alternative schemes. This makes it straightforward to estimate cluster-level quantities and posterior moments (Eq.~(\ref{moments})), which are used in the latent consensus calculation (Eq.~(\ref{ICC}) in Section~\ref{Agreement}).

\paragraph{Network-Informed Adaptive Metropolis-Hastings.}
In the \gls{degmfr}, the order in which ratings occur induces correlation among individuals connected through cross-lag links (Figure~\ref{Example_net}). For this reason, the latent abilities $\{\theta_j\}_{j=1}^J$ are updated using multivariate Barker proposals defined over \textit{network-informed} blocks, meaning we group and sample these individuals jointly, reducing local random-walk behavior in highly correlated regions of the posterior \citepApp{LivingstoneZanella2022,RobertsSahu1997}. To further enhance global exploration, we add a second layer of randomly generated blocks, which helps break residual dependencies not captured by the structural graph. Cluster indicators are sampled using categorical updates, while remaining parameters rely on adaptive random-walk or conjugate samplers when available.

\paragraph{Weighted Average Centered P-Spline.}

To resolve the location indeterminacy among $\lambda_i$, $\xi_r$, and $\nu^t$ in Eq.~(\ref{Measurement:Dynamic_discrimination1}), we adopt a centered reparameterization of the P-spline linear predictor. Defining the  observation-weighted mean $\bar{\nu} = \sum_{t=1}^T  n_t \nu^t / \sum_{t=1}^T n_t$, where $n_t$ is the number  of subjects at position $t$ across raters, we set  $\nu^{t*} = \nu^t - \bar{\nu}$. This identification strategy resembles the sum-to-zero constraint over the observed data commonly used for smooth terms in additive models \citeApp{wood2017}.  The constraint  $\sum_{t=1}^T n_t \nu^{t*} = 0$, combined with the other constraints enforced in Appendix~\ref{identifiability}, resolves the location  indeterminacy among $\nu^t$, $\lambda_i$, and $\xi_r$ in  Eq.~(\ref{Measurement:Dynamic_discrimination1}). For posterior computation, we sample the unconstrained $\nu^t$ from its P-spline prior and apply the centering transform deterministically at each iteration, while $\nu^{t*}$ enters the likelihood. Accordingly, 
Eq.~(\ref{Measurement:Dynamic_discrimination1}) becomes:
\begin{equation}\label{Measurement:Dynamic_discrimination1_star}
    \log(\phi_{ir}^{t}) = \lambda_i + \xi_{r} + \nu^{t*},
\end{equation}
so that $\exp(\nu^{t*})$ is the multiplicative effect on discrimination of position $t$ relative to the observation-weighted log-average across positions. Since only the centered effects are identified, convergence diagnostics are computed on $\nu^{t*}$ (see Supplementary Materials).

\subsection{Post-processing Procedures}

\paragraph{Posterior Subject Scoring and Inter-Rater Consensus.}
As with all model parameters, inference on subject-specific quantities is based on their posterior distributions. Point estimates of the latent trait $\{\theta_j\}_{j=1}^N$, such as the posterior mean (expected a posteriori, EAP), the posterior median, or the maximum a posteriori (MAP), correspond to Bayesian estimators under $L_2$, $L_1$, and $L_0$ loss, respectively, and may serve as subject-level evaluations (e.g., final grades). Uncertainty around these estimates is quantified through posterior credible sets, such as the 95\% equal-tailed credible interval or the 95\% highest posterior density interval (HPD) \citepApp{BDA3}. The latent inter-rater consensus index $LIC_{rr'}$ (Section~\ref{Agreement}) can likewise be evaluated at each MCMC iteration, yielding a full posterior distribution that supports principled uncertainty quantification and formal inference for pairwise latent inter-rater consensus.

\paragraph{Posterior Densities and Cluster Point Estimates.}
The latent trait density equipped with a nonparametric mixture prior can be monitored throughout the MCMC algorithm by evaluating it over a dense grid of equally spaced points. At each iteration, the density is computed at each grid point according to the mixture induced by the finite stick-breaking approximation, yielding a full pointwise posterior representation after the run.
Beyond density  estimation, the discrete nonparametric prior induces a 
posterior distribution over the space of possible 
partitions for subjects. For interpretation, we summarize this through a point estimate that minimizes the variation of information \citepApp[VI;][]{Meila_2007} loss function, using a decision-theoretic approach to probabilistic clustering  \citepApp{Wade&Ghahramani_2018}.

\paragraph{Computational Details.}
In the present work, for both the simulations (Appendix~\ref{sim}) and the real-data analyses, as in previous work (e.g., \citealtApp{mignemi2025BNP,Paganin2023,heinzl2012}), the number of iterations is fixed at 55,000 (thinning by a factor of 10 due to memory constraints), with the first 5000 discarded as burn-in. We fix the maximum number of clusters to be $C=25$ for the \gls{dp} and \gls{pyp} mixtures \citepApp{BDA3}.
The package \textit{salso} \citepApp{Dahl_salso} is used for the point estimate of the clustering structure based on the VI loss function. We graphically inspect plots for convergence and use the package \textit{coda} for model diagnostics \citepApp{CODA}. Section 4 of the Supplementary Materials details MCMC diagnostics for the generated and real data analyzed in this work. We also confirm convergence through multiple MCMC runs with different starting values. The average elapsed time for the \gls{gmfr}, \gls{egmfr}, and \gls{degmfr} models was, respectively, 1497, 2124, and 2559 seconds (approximately 25, 35, and 43 minutes) on a 12th Gen Intel(R) Core(TM) i9 12900H CPU.

\subsection{Joint Posterior Distribution}
Building on the model specification and prior elicitation detailed in Section~\ref{posterior}, we provide here the full joint posterior distribution and a complete overview of the sampling strategy.  We introduce the auxiliary variable $\textit{z}_j$ which indicates the subject cluster membership, such that $\textit{z}_j=c$, if and only if $\theta^*_j=\theta_c^{**}$. Under this stick-breaking approximation, the posterior distribution of the parameter vector $\boldsymbol{S}$ is: 
\begin{equation}\label{joint_posterior}
\begin{aligned}
p(\boldsymbol{S} \mid \mathbf{y}) \propto & \prod_{r=1}^{R} \prod_{i=1}^I  p\Big(y^1_{ir} \;\Big|\; \theta^{b_r(1)}, \, \epsilon_{i}^{b_r(1)}, \, \beta_{ir}, \, \boldsymbol{\tau_{r}}, \, \boldsymbol{\gamma}_{i}, \, \eta_r^1,  \, \lambda_{i}, \, \xi_{r}, \, \boldsymbol{a}\Big) \\
& \times\prod_{t=2}^{J_r} p\Big(y^t_{ir} \;\Big|\; \theta^{b_r(t)}, \, \theta^{b_r(t-1)}, \, \epsilon_{i}^{b_r(t)}, \, \beta_{ir}, \, \boldsymbol{\tau_{r}}, \, \boldsymbol{\gamma}_{i}, \, \eta_r^t,  \, \lambda_{i}, \, \xi_{r}, \, \boldsymbol{a}\Big) \\
& \times \prod_{j=1}^{J} \left[ p(\theta_j \mid \mu_{z_j}, \sigma^2_{z_j}) \cdot p(z_j \mid \mathbf{\pi}) \right] \times p(\mathbf{\pi} \mid \boldsymbol{v}) \times \left[ \prod_{c=1}^{C-1} p(v_c \mid \alpha, d) \right] \times p(\alpha) \\
& \times \prod_{c=1}^{C} \left[ p(\mu_c) \cdot p(\sigma^2_c) \right] \times p(\boldsymbol{a} \mid b) \cdot p(b) \\
& \times \prod_{i=1}^{I} \left[ p(\lambda_i) \cdot p(\varphi_i^2) \cdot \prod_{j=1}^{J} p(\epsilon_{ij} \mid \varphi_i^2) \right] \times \prod_{r=1}^{R} \left[ p(\xi_r) \cdot p(\omega_r^2) \cdot p(\rho_r) \right] \\
& \times \prod_{r=1}^{R} p\left(\eta_{r}^{b_r(1)} \mid {r} \right) \prod_{t=2}^{J_r} p\left(\eta_{r}^{b_r(t)} \mid {r},\rho_r, \theta^{b_r(t-1)}\right) \\
& \times \prod_{i=1}^{I}\prod_{r=1}^{R} p(\beta_{i,r}) \times \prod_{i=1}^{I} p(\boldsymbol{\gamma}_i) \times \prod_{r=1}^{R} p(\boldsymbol{\tau}_r),
\end{aligned}
\end{equation}

A complete overview of the sampling strategy for each block of the joint 
posterior in Eq.~(\ref{joint_posterior}) is provided in Table~\ref{tab:samplers}.

\section{Simulation Study}\label{sim}
We conduct a simulation study to compare the performance of the proposed models (\gls{egmfr} and \gls{degmfr}) with the \gls{gmfr} model proposed in \citeApp{Uto_24}. The simulation study is designed to reflect the small-$R$, large-$I$ setting introduced in Section~\ref{Structural_EG} as in the application of Appendix~\ref{application}. In Section~3 of the Supplementary Materials, we compare with a balanced (fully crossed) design, where each subject is scored independently by every rater. The complementary large-$R$, small-$I$ setting is addressed empirically in the olfactory application of Section~\ref{Odor} of the main text. We jointly control for raters' dynamic effects (i.e., anchoring bias and practice effect) and the multimodality of the subjects' latent trait distribution to investigate the performance of the proposed methods. In this section, we present the major results on the accuracy of our methods under unbalanced rating designs when the unknown latent trait density is modeled through a \gls{dp} mixture prior. Additional results from this study are presented in Section 1 of the Supplementary Materials, and results under the \gls{pyp} mixture prior on the same datasets in Section~2.

\subsection{Setting}
For each item, we generate subjects' scores on a four-category ordered scale, that is $Y_{jir} \in \{0,1,2,3\}$, for $i=1,\dots, I$. Each subject is scored by a random subset of three 
raters on $I = 30$ independent items. The sample sizes for subjects and raters are $N=30$ and $R=5$, respectively, and are fixed across scenarios. This implies an unbalanced design in which subjects are not scored by all raters, and the number of subjects each rater scores (i.e., the respective workload) might vary across raters. 

Data are generated under nested instances of the \gls{degmfr} model corresponding to increasing levels of complexity. We consider a $2 \times 2$ experimental design and generate $10$ independent data sets under each generative scenario, combining two different heterogeneity features: raters' dynamic effects (i.e., \textit{static/dynamic}) and subject latent trait multimodality (i.e., \textit{unimodal/bimodal}). 
Raters' dynamic effects are manipulated according to the following settings:
\begin{description}
    \item[Static:] Under this scenario, ratings are invariant under permutation, and the raters' dynamic effect is fixed to zero. As a consequence, data are generated from the \gls{egmfr} model detailed in Section \ref{E-GMFRM}, where $\rho_r=0$, for $r=1,\dots, R$, and $\nu^t=0$, for $t=1,\dots, T$. 
    \item[Dynamic:] Data are generated from the \gls{degmfr} model with rater-specific cross-lagged effects, namely $\rho_1=\rho_2=0$, $\rho_3=1$, $\rho_4=-0.8$, $\rho_5=0.5$. The presence of the practice effect is determined by the following values of the basis coefficients of the P-Spline, $\boldsymbol{a}=(0.1,0.3,0.4,0.4,0.4,0.3)$. These coefficients imply a smooth trend in 
    average log-discrimination over the rating sequence, with an initial warm-up phase followed by a stable plateau and a mild decline at the end.
\end{description}
Subject latent trait multimodality is manipulated as follows: in both cases, the marginal variance of  $\theta_j$ is fixed to match across scenarios to enable comparability:
\begin{description}
    \item[Unimodal:] Under this scenario we generate each latent trait i.i.d. from a normal distribution with $\mu_\theta=0$ and variance $\sigma_\theta^2=2.386^2$ across subjects, that is $\theta_j \overset{\mathrm{iid}}{\sim}N(0,2.386^2)$. 
    \item[Bimodal:] Under the bimodal scenarios, latent traits are generated i.i.d. from a mixture of two equiprobable normal distributions, with component-specific means, respectively, $\mu_{\theta_1}=-2$ and $\mu_{\theta_2}=2$ and variance $\sigma^2_{\theta_1}=2.5$ and $\sigma^2_{\theta_2}=0.8$, respectively, that is $\theta_j \overset{\mathrm{iid}}{\sim} 0.5\;N (-2,2.5) + 0.5\; N(2,0.8)$  chosen so that the marginal variance matches $\sigma^2_\theta=2.386^2$. 
\end{description}
The combination of these two different heterogeneity features (i.e., latent trait multimodality and rater dynamic effects) results in four different generative scenarios: \textit{Static-Unimodal} (S-U); \textit{Static-Bimodal} (S-B); \textit{Dynamic-Unimodal} (D-U); \textit{Dynamic-Bimodal} (D-B).  

The remaining measurement parameters are fixed or generated analogously across scenarios. We fix the item residual variance $\varphi_i^2=1$, and generate the local difficulty shift for each item as $\beta_{ir}\overset{\mathrm{iid}}{\sim} N(0,1)$. The step difficulty parameters for each item are generated as $\gamma_{ik} \overset{\mathrm{iid}}{\sim} N(0,3^2)$, for $k=1,\dots,K-1$, with $\gamma_{iK}=-\sum_{k=1}^{K-1}\gamma_{ik}$. Raters' log-variances of local biases 
are generated i.i.d. across raters, that is $log(\omega_r^2) \overset{\mathrm{iid}}{\sim} N(-0.5, 0.5^2)$, and the threshold step parameters for each rater as $\tau_{rk} \overset{\mathrm{iid}}{\sim} N(-0.5,2^2)$, for $k=1,\dots,K-1$ with $\tau_{r1}=-\sum_{k=2}^{K}\tau_{rk}$. We generate the item log-discrimination parameter as $\lambda_i \overset{\mathrm{iid}}{\sim}N(0, 0.2^2)$, $i=2,\dots,I$, with $\lambda_1=-\sum_{i=2}^I\lambda_i$. Finally,  the rater local log-discrimination shifts $\xi_r\overset{\mathrm{iid}}{\sim}N(0, 0.3^2)$, for $r=2,\dots,R$, with $\xi_1=-\sum_{r=2}^R\xi_r$.

The generative values of all parameters satisfy the identifiability constraints detailed in Appendix~\ref{identifiability}.

We use posterior means as point estimates for all parameters. Parameter recovery performance is evaluated using the Root Mean Square Error (RMSE) and the Average Bias (Bias). The RMSE is defined as the root mean squared difference between the posterior mean and the true parameter value across simulated data sets, whereas the Average Bias is computed as the average difference between the posterior mean and the true value across data sets. For parameters associated with a group of units (subjects, items,  or raters), both metrics are first averaged over the relevant units within each replicate and then across the $10$ independent replications.

\subsection{Results}
The main results on parameter recovery are summarized in Table~\ref{Tab_sim_theta} and Figure~\ref{RMSE_theta}. Table~\ref{Tab_sim_theta} reports RMSE and Bias for $\theta$, $\xi$, $\lambda$, and $\rho$; the remaining parameters are reported in Table~1 of the Supplementary Materials. Figure~\ref{RMSE_theta} displays boxplots of the RMSE and average Bias for $\theta$ across the $10$ independent replications. Overall, the \gls{egmfr} and \gls{degmfr} models achieve substantially lower and more consistent RMSE compared to the \gls{gmfr} across all scenarios.
As expected, the \gls{egmfr} and \gls{degmfr} models each perform best in estimating the subject latent trait under the generative scenario that matches their specification: the \gls{egmfr} provides more accurate estimates of the subject latent trait under static scenarios, while the \gls{degmfr} improves estimation under dynamic scenarios, when anchoring biases and practice effects are present (Table~\ref{Tab_sim_theta}, S-U and S-B vs D-U and D-B). This pattern is also observed for the rater-specific parameters $\xi$ (Table~\ref{Tab_sim_theta}) and $LIC$, and the item residual variance $\varphi^2$ (Table 1 of the Supplementary Materials). The estimates of the category step parameters $\tau$ and $\gamma$ are comparable across scenarios under both the EG- and \gls{degmfr} models (Table 1 of the Supplementary Materials).\\
The average estimated density of the latent trait $\theta$ under the \gls{egmfr} and \gls{degmfr} models is identical under the Static scenarios (Figure \ref{Densities_sim} and Figure 1 of the Supplementary Materials); the \gls{degmfr} model provides, on average, a better estimate of this density under the Dynamic-Bimodal generative scenario (Figure \ref{Densities_sim}). Given the small sample size ($N=30$ individuals, $I=30$ items, $R=5$ raters), some difference between estimated and true density is expected, especially in low-density regions; the credible intervals in Figure \ref{Densities_sim} capture this uncertainty. Figure 2 in the Supplementary Materials shows the posterior similarity matrix for subject allocation, which correctly identifies a single cluster in the Static-Unimodal scenario and two distinct subgroups in the Dynamic-Bimodal scenario, suggesting that the \gls{degmfr} model recovers both the clustering structure and the marginal density of $\theta$. The P-spline estimates under the \gls{degmfr} model show good performance in the interior region, whereas boundary effects increase estimation uncertainty, particularly near the end of the rating process (Figure \ref{P-Spline_sim}).  

\begin{table}
\centering
        \begin{tabular}{l l c c c c c c}
           \toprule
         &      & \multicolumn{2}{c}{\gls{gmfr}}  & \multicolumn{2}{c}{\gls{egmfr}} & \multicolumn{2}{c}{\gls{degmfr}} \\    \midrule  
  Parameter & Generative Scenario    &   RMSE          & Bias     &   RMSE          & Bias       &   RMSE          &  Bias      \\ \midrule 
  
$\theta$   &S-U   & 1.229 & 0.036 & \textbf{0.709} & 0.023 & 0.744          &  \textbf{0.011} \\
           & S-B  & 1.583 & 0.056 & \textbf{0.845} & 0.044 & 0.863          &  \textbf{0.033} \\
           & D-U  &  1.371 & 0.034 & 0.885          & 0.024 & \textbf{0.751} & \textbf{-0.016} \\
           & D-B  & 1.768 & 0.057 & 1.033          & 0.036 & \textbf{0.840} &  \textbf{0.027} \\ \midrule

  $\xi$  &S-U       & 0.222 & -0.031 & \textbf{0.172} & -0.011          & 0.174           & \textbf{-0.009} \\ 
         & S-B      &  0.678 & -0.201 & \textbf{0.593} & -0.173          & 0.604           & \textbf{-0.171} \\ 
         & D-U      &   0.271 & -0.021 & 0.190          & -0.010          & \textbf{0.189}  & -0.010 \\ 
         & D-B      & 0.744 & -0.182 & 0.603          & \textbf{-0.165} & \textbf{0.595}  & -0.167 \\ \midrule 

 $\lambda$  &S-U  & 0.630 &  0.142          & \textbf{0.376} & \textbf{0.068} & 0.377          & \textbf{0.068} \\
            & S-B & 0.563 &  0.138          & 0.413          & \textbf{0.079} & \textbf{0.412} & \textbf{0.079} \\
            & D-U &   0.457 & \textbf{-0.037} & 0.381          & 0.071          & \textbf{0.376} & 0.069 \\
            & D-B & \textbf{0.393} & -0.129 & 0.411          & \textbf{0.079}          & 0.405          & \textbf{0.079} \\ \midrule

$\rho$      &S-U  & - & - & - &- & 0.139 & 0.029\\ 
            & S-B & - & - & - &- &  0.091 & 0.015\\ 
            & D-U & - & - & - &- &  0.155 & 0.019\\ 
            & D-B & - & - & - &- & 0.101 & 0.035\\ 
              \bottomrule
        \end{tabular}
   
\caption[]{ \small Root Mean Square Error (RMSE) and Average Bias (Bias) across $10$ independent replications are reported for the \acrfull{gmfr} model, 
the \acrfull{egmfr}, 
and the \acrfull{degmfr}
. The \gls{egmfr} and \gls{degmfr} are fitted with a \gls{dp} mixture prior on the latent trait. Results are shown under four generative scenarios: \textit{Static-Unimodal}(S-U); \textit{Static-Bimodal}(S-B); \textit{Dynamic-Unimodal} (D-U); \textit{Dynamic-Bimodal} (D-B). We indicate the smallest RMSE and absolute Bias values across models for each scenario in bold; bold is omitted when all models achieve the same value.}
 \label{Tab_sim_theta}
    \end{table}

\begin{figure}
    \centering
    \includegraphics[width=\textwidth]{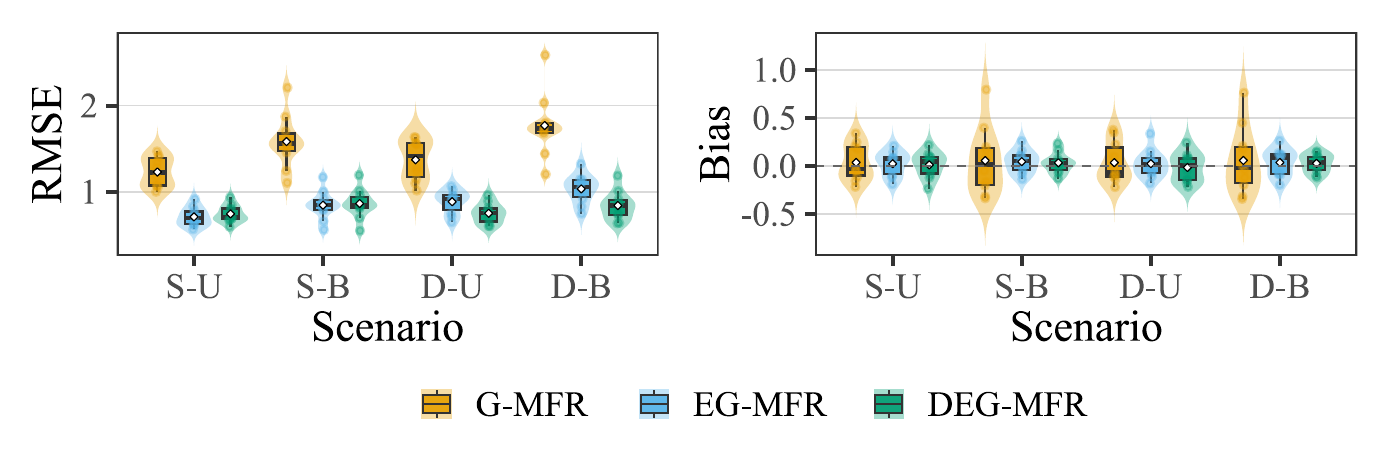}
    \caption{\small Boxplots of the Root Mean Square Error (RMSE) and the Average Bias of subject-specific latent trait $\theta_j$ computed across subjects, over 10 independent datasets. Results are shown under the \acrfull{gmfr}, the \acrfull{egmfr}, and the \acrfull{degmfr}
    with \gls{dp} Mixture prior, under four generative scenarios: \textit{Static-Unimodal}(S-U); \textit{Static-Bimodal}(S-B); \textit{Dynamic-Unimodal}(D-U); and \textit{Dynamic-Bimodal}(D-B).}
    \label{RMSE_theta}
\end{figure}

\begin{figure}
    \centering
    \includegraphics[scale=0.65]{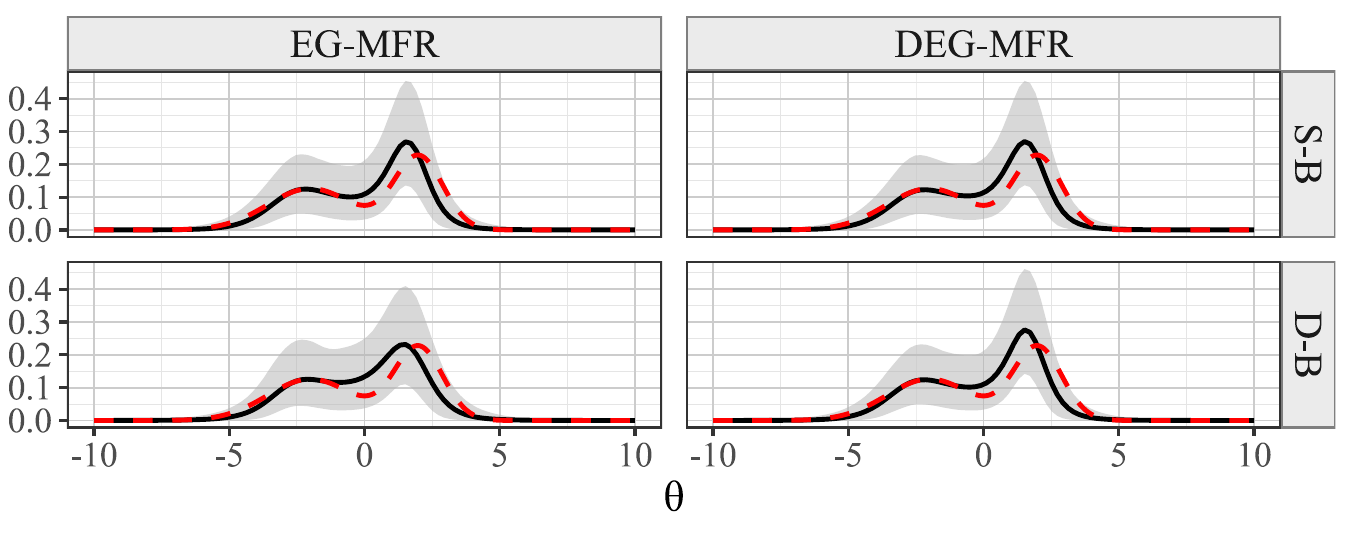}
    \caption{Average estimated density of the latent trait $\theta$ across 10 independent datasets under different scenarios. The top-left and top-right figures display the estimated density of the \acrfull{egmfr} 
    and the \acrfull{degmfr} with \gls{dp} Mixture prior. Results are shown under two generative scenarios: \textit{Static-Bimodal}(S-B) and \textit{Dynamic-Bimodal}(D-B).
    The dashed red lines indicate the true density; the black solid line and the shaded grey area indicate, respectively, the pointwise mean and $95\%$ quantile-based credible intervals.}
    \label{Densities_sim}
\end{figure}

\begin{figure}
    \centering
    \includegraphics[scale=0.75]{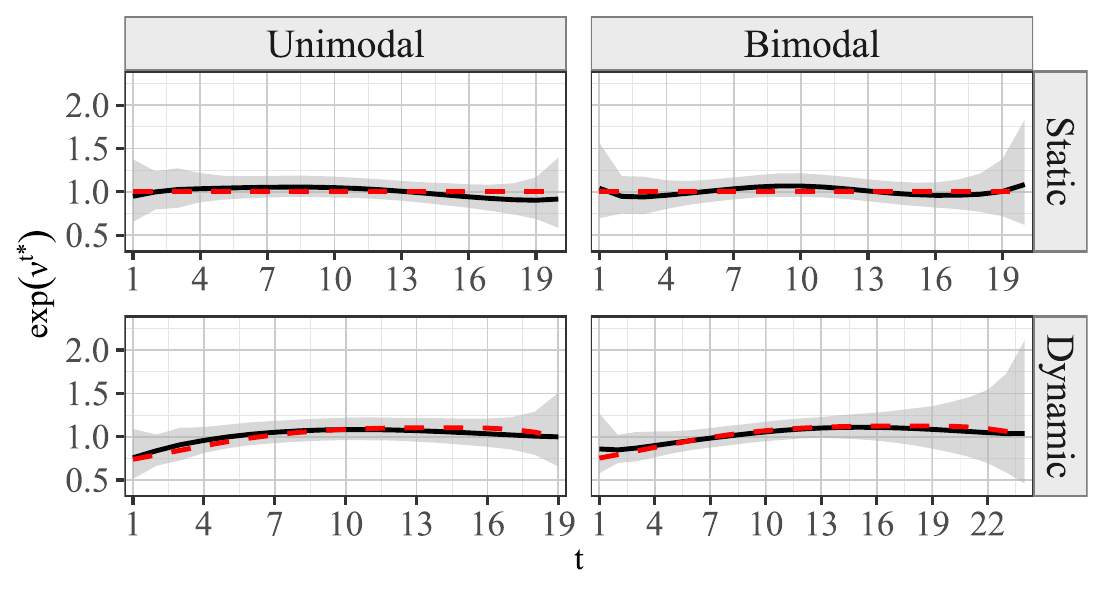}
    \caption{ Estimates of the posterior average rater practice effect $\exp(\nu^{t*})$ across the sequence of subjects ($t=1,\dots,T$), under the \gls{degmfr} Model for different generative scenarios. The dashed red lines indicate the true latent trend; the black solid line and the shaded grey area indicate, respectively, the pointwise posterior mean and $95\%$ quantile-based credible intervals.}
    \label{P-Spline_sim}
\end{figure}

\section{Health Care Professional Assessment}\label{application}
The Objective Structured Clinical Examinations (OSCEs) represent a widely adopted standardized approach for evaluating the clinical skills, knowledge, and professional attitudes of medical students \citepApp{Uto_24}. In this setting, a small panel of expert raters independently assesses a cohort of students on multiple structured items, making OSCEs a paradigmatic application for many-facet Rasch modeling. At the same time, the limited number of raters and the structured nature of the evaluation process make these data an informative benchmark for assessing the behavior of our framework under a small-$R$, large-$I$ regime. Prior analysis of similar data suggests that student ability distributions may be multimodal or skewed \citepApp{Uto_24}, motivating a nonparametric mixture prior over the latent trait.

The dataset, made available by \citeApp{Uto_24} at \url{https://datadryad.org/stash/dataset/doi:10.5061/dryad.tmpg4f56q}, 
consists of ratings collected from OSCEs at Tokyo Medical and Dental University in 2021. The evaluation targeted sixth-year dental students, whose performance during recorded medical interview videos was scored by a panel of five expert raters ($R = 5$) on $I = 30$ structured items using a 4-point ordinal scale ($K = 3$). Evaluation items in the OSCE scoring rubric are reported in Table~\ref{Tab_OSCE_items} of the Appendix~\ref{Additional_fig_tab}. The total sample comprises $J = 30$ subjects. Among the five raters, two assessed all 30 subjects, while the remaining three each evaluated a non-overlapping subset of 10 subjects. All raters worked independently, without communication or consultation with one another. This yields a sparse, unbalanced rater--subject design with a total of  $5 \times 30 = 150$ subject--rater pairs, of which 90 
are observed. As the evaluation order of subjects within each rater is not explicitly recorded, we treat the sequential index in the original dataset as a proxy for the order in which subjects were assessed. 
The results concerning dynamic effects should, however, be interpreted with this limitation in mind.

We compare six nested models using the Watanabe--Akaike information criterion (WAIC): the \gls{gmfr}, the \gls{gmfr} 
with a \gls{dp} mixture prior over the latent trait distribution (G-MFR-DP), the \gls{egmfr}, the \gls{degmfr}, and the DEG-MFR models with a \gls{pyp} mixture over the latent trait distribution with discount parameter equal, respectively, to $d=0.3$ and $d=0.6$. WAIC provides a fully Bayesian approach to estimating the out-of-sample expectation, relying on the computed log-pointwise posterior predictive density and a penalty term on the effective number of parameters to prevent overfitting \citepApp{Gelman2014}.

\begin{table}
\centering
        \begin{tabular}{l c c }
           \toprule
          Fitted Model & $WAIC$ & $\Delta WAIC$ \\
            \midrule  
         DEG-MFR PYP Model $(d=0.3)$   &    2262.126     &    -        \\
         DEG-MFR PYP  Model $(d=0.6)$  &    2268.360     &   -6.234          \\
         DEG-MFR DP  Model             &    2269.313     &   -7.187          \\
         EG-MFR  DP  Model             &    2293.549     &   -31.423      \\
         G-MFR   DP Model              &    3351.643     &   -1089.517                \\
         G-MFR    Model                &    3357.317     &   -1095.191    \\
       
 \bottomrule
        \end{tabular}
\caption[]{\small The Watanabe--Akaike information criterion (WAIC) is reported for each of the fitted models: Generalized Many-Facet Rasch model (G-MFR), the Generalized Many-Facet Rasch model with a \gls{dp} mixture prior over the latent trait distribution (G-MFRM DP), the Extended Generalized Many-Facet Rasch model (EG-MFR DP), the Dynamic Extended Generalized Many-Facet Rasch model (DEG-MFR DP), and the DEG-MFR models with a \gls{pyp} mixture over the latent trait distribution (DEG-MFR PYP) Model with discount parameter equal, respectively, to $d=0.3$ and $d=0.6$. We report the pairwise WAIC difference ($\Delta WAIC$) between the best-fitting model and the others.}
\label{Tab_WAIC_OSCE}
\end{table}

No convergence or mixing issues emerged from graphical inspections of the MCMC chains and associated diagnostics from the \texttt{coda} package \citepApp{CODA}; further details and examples of trace plots are given in Section 4 of the Supplementary Materials. Table \ref{Tab_WAIC_OSCE} shows the WAIC indices for each fitted model and indicates that the DEG-MFR PYP model with a discount parameter $d=0.3$ best fits the data and outperforms the others in predicting out-of-sample ratings. Based on the model comparison procedure, we focus on the results from the DEG-MFR PYP with $d=0.3$. \\
Table~\ref{Tab_posteriors_OSCE} in Appendix~\ref{Additional_fig_tab} reports the Expected A Posteriori (EAP) and 95\% credible intervals (95\% CI) for each subject latent trait $\theta_j$, the rater-specific discrimination parameter $\xi_r$, the variance of the local biases $r$, the rater's anchoring bias $\rho_r$, the rater-specific category step parameters $\{\tau_{rk}\}_{k=1}^{K}$, and item discrimination parameters $\lambda_i$. Section 5 of the Supplementary Materials reports the posterior summary of the other parameters (e.g., the rater local biases $\beta_{ir}$), and Figure \ref{Density_P-Spline_app} displays the posterior estimates of the subject latent trait density and the rater practice effect curve.
The 95\% credible intervals for all estimated anchoring effects $\rho_r$ include zero (see Figure ~\ref{Wright_map}), consistent with the practice effect curve in Figure~\ref{Density_P-Spline_app}. The estimated latent trait density is unimodal, but more concentrated around its mode than the standard normal, suggesting relatively limited heterogeneity in the latent clinical-performance construct measured by these items.

The posterior point estimate of the number of clusters in the sample suggested by the VI loss function is $\hat{M_N}=2$ (see Figure~\ref{PSM_real_data_set} in Appendix \ref{Additional_fig_tab}); it assigned subject 8 to a different cluster than the other subjects. This result is consistent with the results reported in \citeApp{Uto_24}, although our method provides a model-based detection of this systematic difference in subjects. 

Table~\ref{post_summary_LIC_OSCE} reports the posterior estimates of the pairwise LIC between rater 1 and the other raters. Latent consensus is moderate to high, and all 95\% credible intervals exclude zero. Consensus is largest between raters 1 and 5 ($LIC_{1,5}=0.910$ when both rate the subject first), the two most discriminating raters ($\xi_1=0.426$, $\xi_5=0.250$). In contrast, rater 4 shows the lowest consensus with rater 1 ($LIC_{1,4}=0.502$), consistent with having the lowest discrimination ($\xi_4=-0.339$) and the largest variance of local biases ($\omega_4^2=1.018$), nearly three times larger than any other rater. Because the anchoring effects are small and their credible intervals include zero, consensus changes little when raters share the preceding subject (e.g., from $0.910$ to $0.897$ for raters 1 and 5).

Figure~\ref{ICC_1} and Figure~\ref{ICC_1_bis} in Appendix \ref{Additional_fig_tab} show rater-specific global Item Characteristic Curves (ICC) across subjects. This is the ICC of a rater controlling for the subject-item interactions and the dynamic rater effects, which enables the comparison of raters on their average response styles across subjects.Figure~\ref{Wright_map} demonstrates that rater severity varies by rater for the same item. For item 17 (\textit{confirmation of necessary information for each disease}), Raters 1 and 5 clearly rate more strictly than the others, with larger local difficulty shifts for this item than the other raters. 
This pattern is also visible in Figure~\ref{ICC_1}: Raters 1 and 5 show much steeper, almost step-like category transitions for item 17, with little use of the second and fourth category, suggesting a more dichotomous response pattern 
for this item. This pattern is specific to item 17 and does not occur for Rater 5 on the other items in Figure~\ref{ICC_1_bis}. Figure~\ref{ICC_H} in Appendix \ref{Additional_fig_tab} shows the hierarchical changes in ICCs for item 30 (\textit{Overall rating}) after considering the local difficulty of subject 5 for this item and rater 2's effects.  


\begin{table}
\centering
\small
\begin{tabular}{l c c c c c c}
\toprule
 & Subject $j$ & $(t_r(j),\,t_{r'}(j))$ & Preceding subjects & $\rho_r\rho_{r'}$ & EAP & 95\% CI \\
\midrule
\multicolumn{7}{@{}l}{\textit{No subject precedes $j$ for either rater}} \\
  $LIC_{1,2}$ & 1 & (1, 1) & -- & 0.018 & 0.722 & (0.437, 0.934) \\
  $LIC_{1,5}$ & 1 & (1, 1) & -- & 0.123 & 0.910 & (0.775, 0.979) \\
\midrule
\multicolumn{7}{@{}l}{\textit{Both raters rated the same subject immediately before $j$}} \\
  $LIC_{1,2}$ & 2 & (2, 2) & (1, 1) & 0.018 & 0.645 & (0.313, 0.893) \\
  $LIC_{1,5}$ & 2 & (2, 2) & (1, 1) & 0.123 & 0.897 & (0.754, 0.975) \\
  $LIC_{1,3}$ & 12 & (12, 2) & (11, 11) & 0.119 & 0.663 & (0.378, 0.887) \\
  $LIC_{1,4}$ & 22 & (22, 2) & (21, 21) & 0.146 & 0.502 & (0.216, 0.778) \\
\bottomrule
\end{tabular}
\caption{\small Expected A Posteriori (EAP) and 95\% credible intervals of the pairwise latent inter-rater consensus $LIC_{rr'}$ (Eq.~(\ref{ICC})) between rater 1 and the other raters from OSCE data, under three rating configurations of a common subject $j$: no subject precedes $j$ for either rater ($\tilde\rho_r=\tilde\rho_{r'}=0$); both raters rated the same subject immediately before $j$ ($s_{rr'}(j)=1$); the raters rated different subjects immediately before $j$ ($s_{rr'}(j)=0$). Positions $(t_r(j),t_{r'}(j))$ and preceding subjects are reported for $r=1$ and $r'$, respectively; $\rho_r\rho_{r'}$ is the EAP of the product of the anchoring effects.}
\label{post_summary_LIC_OSCE}
\end{table}

\begin{figure}
    \centering
     \includegraphics[scale=0.70]{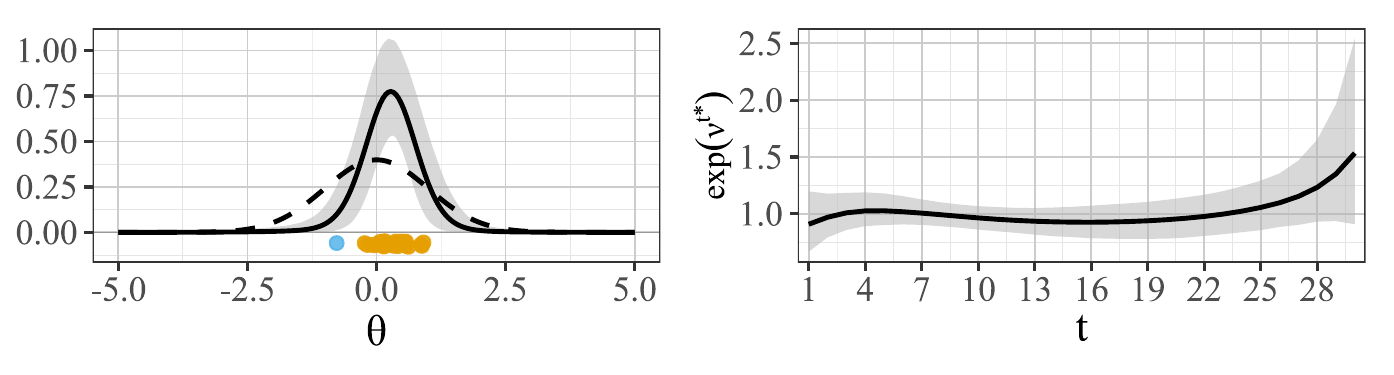}
    \caption{\small Left panel: graphical representation of the estimated latent trait density; the black solid line indicates the Expected A Posteriori, the shaded grey area indicates the 95\% quantile-based credible interval, and the black dashed line indicates the standard normal reference density. Right panel: graphical representation of the estimated rater practice effect $\text{exp}(\nu^{t*})$; the black solid line and the shaded grey area indicate, respectively, the  Expected A Posteriori and the 95\% quantile-based credible interval of the P-Spline. }
    \label{Density_P-Spline_app}
\end{figure}

\begin{figure}
    \centering
    \subfigure{\includegraphics[scale= 0.80]{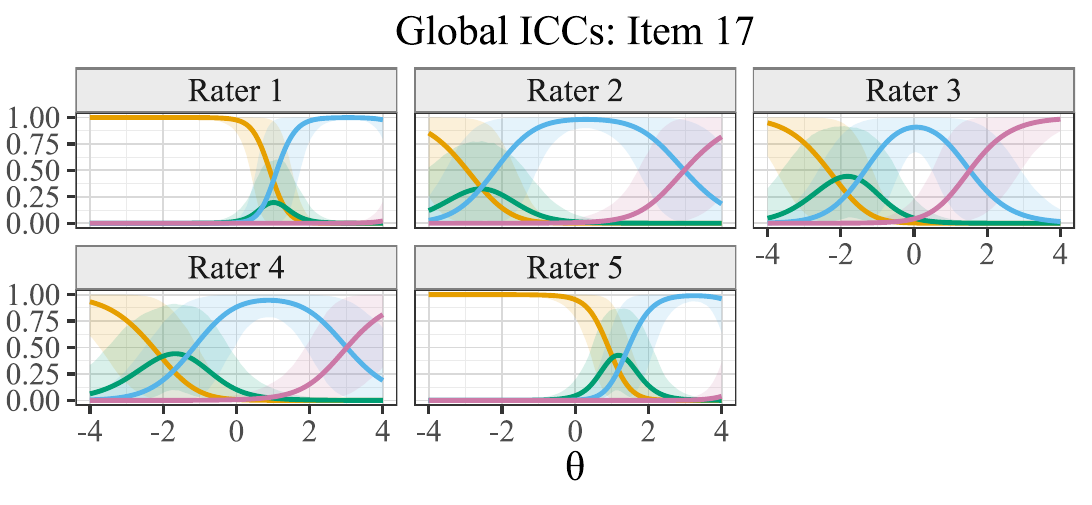}} 
    \caption{\small Global Item-Rater Characteristic Curves for item 17 (\textit{Confirmation of necessary information for each disease}) across subjects for the OSCE data, controlling for rater dynamic effects. The orange, green, aquamarine, and violet lines represent the four ordered response categories; the corresponding solid lines and shaded areas indicate, respectively, the Expected A Posteriori and the 95\% quantile-based credible interval for the category response curves.}
    \label{ICC_1}
\end{figure}

\begin{figure}
    \centering
    \includegraphics[scale=0.6]{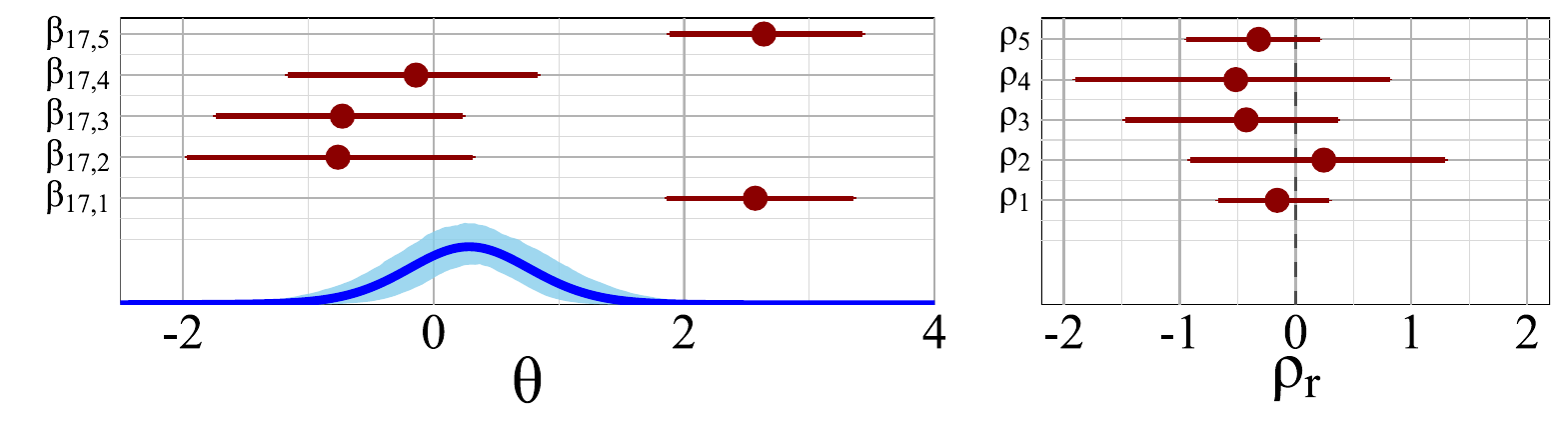}
    \caption{\small Left side: Wright map of raters' local difficulty shifts $\beta_{ir}$ for item 17 from OSCE data. Right panel: posterior anchoring effects $\rho_r$ for the same raters; the dashed vertical line marks $\rho_r=0$. For each parameter, the Expected A Posteriori value and the 95\% quantile-based credible interval are indicated by the red point and the bar, respectively. The Expected A Posteriori value and the 95\% quantile-based credible interval of the latent trait density are indicated, respectively, by the blue solid line and the blue shaded area.}
    \label{Wright_map}
\end{figure}

\section{Proofs}\label{A_P}
\subsection{Proof of Proposition 1}
\begin{proof} 
Let $\zeta_{jr}=\theta_j - \eta_{r}^{t_r(j)}$ and $\zeta_{jr'}=\theta_j - \eta_{r'}^{t_{r'}(j)}$, and assume mutual independence between the terms of the respective decompositions:
\[
Var(\zeta_{jr}|G, \omega^2_r,\rho_r) = \sigma^2_{\theta} + \omega^2_r + \tilde\rho^{\,2}_r \sigma^2_{\theta}, \quad \quad \quad Var(\zeta_{jr'}|G, \omega^2_{r'}, \rho_{r'}) = \sigma^2_{\theta} +  \omega^2_{r'} + \tilde\rho^{\,2}_{r'}\sigma^2_{\theta},
\]
after noting that $Var(\eta_r^{b_r(t)}|G, \omega^2_r,\rho_r)=E(Var(\eta_r^{b_r(t)}|G, \omega^2_r,\rho_r)) + Var(E(\eta_r^{b_r(t)}|G, \omega^2_r,\rho_r))=  \omega^2_r + \tilde\rho^{\,2}_r \sigma^2_{\theta}$. 

Moreover, since $\theta_j$ is independent of the preceding subjects' latent traits and the residuals of $\eta_{r}^{t_r(j)}$ and $\eta_{r'}^{t_{r'}(j)}$ are independent, the two local biases covary only through a shared preceding subject, i.e., $Cov(\eta_{r}^{t_r(j)},\eta_{r'}^{t_{r'}(j)})=\tilde\rho_r\tilde\rho_{r'}\sigma^2_\theta\, s_{rr'}(j)$:
\[
Cov(\zeta_{jr},\zeta_{jr'}|G, \boldsymbol{\rho, \omega}^2) = \sigma^2_\theta + Cov(\eta_{r}^{t_r(j)},\eta_{r'}^{t_{r'}(j)}) = \sigma^2_\theta\left(1 + \tilde\rho_r \, \tilde\rho_{r'}\, s_{rr'}(j)\right),
\]
Follows that: 
\[
Corr(\zeta_{jr},\zeta_{jr'}|G,\boldsymbol{\rho, \omega}^2) = \frac{Cov(\zeta_{jr},\zeta_{jr'}|G,\rho_r,\rho_{r'})}{\sqrt{Var(\zeta_{jr}|G,\rho_r)}\sqrt{Var(\zeta_{jr'}|G,\rho_{r'})}}= \frac{\sigma^2_\theta\left(1 + \tilde\rho_r \, \tilde\rho_{r'}\, s_{rr'}(j)\right)}{\sqrt{\sigma^2_\theta(1+\tilde\rho^{\,2}_r)+\omega_r^2}\sqrt{\sigma^2_\theta(1+\tilde\rho^{\,2}_{r'})+\omega_{r'}^2}}
\]

This completes the proof of Proposition \ref{prop_1}.
\\
Given the results on the posterior convergence of the density of $G$ \citepApp{Ghosal_van_der_Vaart_2017} and noting that, under our formulation, the posterior moments of $G$ are uniformly integrable, $\sigma_\theta^2$ is posterior consistent. We further note that also $\rho^2_r$ and $\omega_r^2$ are posterior consistent by the Bernstein--von Mises theorem. As a direct consequence, $LIC_{rr'}$ is also posterior consistent by the Continuous Mapping Theorem.

\end{proof}

\subsection{Model identifiability}\label{identifiability}

As for other \gls{irt}-type models \citepApp{Uto_24, Paganin2023, Uto_2020, Uto_2016}, we enforce identifiability by restricting the parameter space. For each rater and each item, we fix the first category step parameter to zero, that is $\tau_{r0}=0$, for $r\in \mathcal{R}$, and $\gamma_{i0}=0$, for $i=1,\dots,I$. Similarly to \citeApp{Uto_24}, we resolve the location indeterminacy of measurement parameters by applying a zero-sum constraint, respectively, to the rater- and item-specific category step parameters, that is $\sum_{k=1}^K \tau_{rk}=0$, for $r \in \mathcal{R}$, and $\sum_{k=1}^K \gamma_{ik}=0$, for $i=1,\dots,I$; and specifying a prior with fixed location parameter for all the rater local difficulty shifts $\beta_{ir}$, $i=1,\dots,I$ and $r \in \mathcal{R}$. Under the \gls{egmfr} model, all the item-specific residual terms $\epsilon_{ij}$, $i=1,\dots,I$, $j \in \mathcal{J}$, and all the raters' subject local biases $\eta_{rj}$, $r \in \mathcal{R}$, $j \in \mathcal{J}_r$, have zero mean by the model.

The \gls{degmfr} model requires additional considerations. For each rater, the mean of $\eta_{rj}$ is fixed to zero for $t=1$, and modelled through the dynamic latent regression specified in Eq.~(\ref{CL_path}) otherwise. For $t=2,\dots, J_r$, $\eta_{jr}$ is constrained by the effect of the latent trait of the subject $t-1$ for each rater. We resolve the scale indeterminacy of measurement parameters by constraint $\sum_{r=1}^R \xi_{r}=0$, $\sum_{i=1}^I \lambda_{i}=0$, and $\sum_{t=1}^T n_t\, \nu^{t*}=0$ (where $n_t$ is the number of subjects assigned to the $t$-th position in the evaluation process across raters), and fix the prior variance for all the rater local difficulty shifts $\beta_{ir}$, $i=1,\dots, I$  and anchoring bias $\rho_r$ for $r \in \mathcal{R}$.


The constraints on the measurement parameters detailed above fix the location and scale of the latent trait, allowing the nonparametric prior on $G$ to freely estimate the shape of the latent trait distribution. Posterior consistency of $G$  follows from standard results for \gls{dp}/\gls{pyp} mixtures with Gaussian kernels \citepApp{Ghosal1999}. Related but distinct notions of Bayesian identifiability in semiparametric \gls{irt} models are discussed in \citeApp{SanMartin2011}.

Under the \gls{degmfr} model $\prod_{j=1}^N (K+1)^{|\mathcal{R}_j|\times I}$ different response patterns might be observed, and $\prod_{j=1}^N K^{|\mathcal{R}_j|\times I}-c$ independent relations remain to identify $F$ with $c$ being the finite number of constraints over the measurement parameters; as $N\rightarrow\infty$ the model fully identify $F$, which is an infinite-dimensional parameter. In finite samples, a sufficiently large number of subjects $N$, items $I$, and raters per subject $|\mathcal{R}_j|$ increases the information available about $G$, making it possible to approximate the latent trait distribution reasonably well.

\section{Comparison of Naive and Model-based Latent Trait Estimates}\label{sec:naive}
In rater-mediated assessment, the most immediate approach to scoring a subject is the sample mean or the sum of the ratings across multiple items assigned to that subject across independent raters, namely 
\[
\hat\theta_{\text{Naive}, j} = \frac{1}{|\mathcal{R}_j|\, \cdot I} \sum_{r \in \mathcal{R}_j}
\sum_{i=1}^{I} Y_{jri}.
\]
This estimator implicitly treats raters as exchangeable and interval-scaled, ignoring rater severity, item difficulty, and any dynamic (e.g., time- or sequence-dependent) sources of measurement error. In the presence of pronounced rater effects, this scoring method might fail to recover the subject's latent trait and provide accurate estimates \citepApp{Mignemi_Chen_Moustaki_2025, Uto_24}. We provide some examples of the cost of this simplification and compare $\hat\theta_{\text{Naive}}$ against the model-based posterior mean $\hat\theta_{\text{DEG}}$ obtained under the \gls{degmfr} model
, which explicitly accounts for rater and item facets together with a time-varying multiplicative effect. In the example from the simulation study, both estimators were benchmarked against the generating latent trait
$\theta$; for the two real-data applications, in the absence of a ground-truth $\theta$, we assessed agreement between $\hat\theta_{\text{Naive}}$ and $\hat\theta_{\text{DEG}}$ directly. Because such scores are routinely used to inform high-stakes and selection decisions, where what matters is the ordinal position of a subject relative to others rather than the metric value of the score itself, we quantified agreement via the Spearman rank correlation coefficient $\rho$, which is invariant to monotonic transformations and therefore directly targets rank preservation rather than linear concordance. We note that only in this Section do we refer to $\rho$ as the rank correlation coefficient rather than the rater-specific anchoring effect; for ease of communication, we wanted to preserve the usual Greek letter $\rho$ to denote the Spearman correlation. \\
Figure \ref{ScatterPlot_1} compares the latent score recovery performance of the two methods in two datasets generated under the Static-Unimodal and Dynamic-Bimodal scenarios, respectively. Considering the dataset in which rater dynamic effects are not present and the subject latent trait distribution is unimodal (i.e., \textit{Static-Unimodal Scenario}), it emerges that, although the scale of the estimates provided by the Naive method is rather different from the true one, the rank of subject scores provided by the two methods is comparable, namely $\rho=0.98$.  On the contrary, considering the  \textit{Dynamic-Bimodal Scenario} which is a far more complex setting, our proposal yields a more accurate ranking of subject scores ($\rho=0.95$) than the naive method ($\rho=0.87$). \\
Figure~\ref{ScatterPlot_2} shows the Spearman rank correlation between estimates from the naive and the \gls{degmfr} estimators for the OSCE and the Odor data sets analyzed in Appendix~\ref{application} and Section~\ref{Odor} of the main text, respectively. 

The subject rankings recovered by the two estimators, although comparable ($\rho = 0.90$), are different across the two real data sets, as shown by Figure~\ref{rank_shift}. 

\begin{figure}
    \centering
    \includegraphics[scale=0.68]{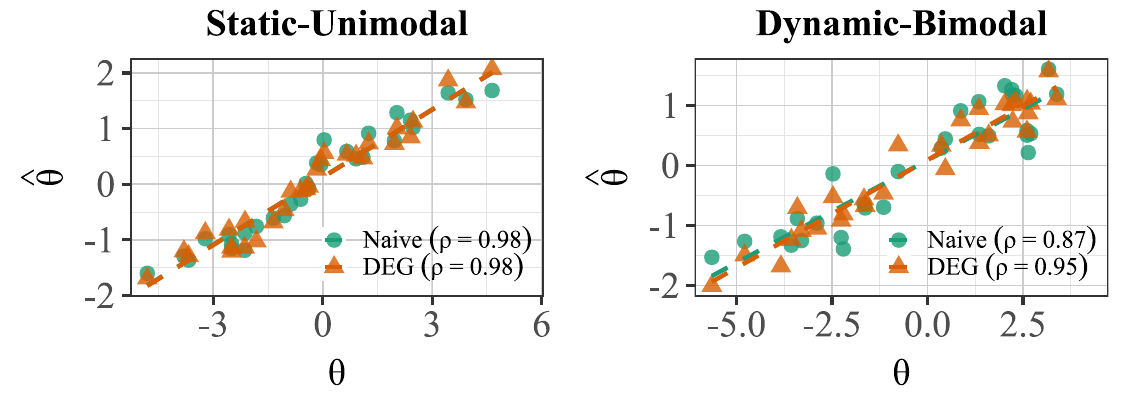}
    \caption{Comparison between Naive average score estimates and DEG model-based estimates of subject latent trait $\theta$ in two datasets from different scenarios; and respective Spearman rank correlation between the true values and the standardized estimates (with zero mean and unit variance) of the two methods for comparison reasons.}
    \label{ScatterPlot_1}
\end{figure}

\begin{figure}
    \centering
    \includegraphics[scale=0.7]{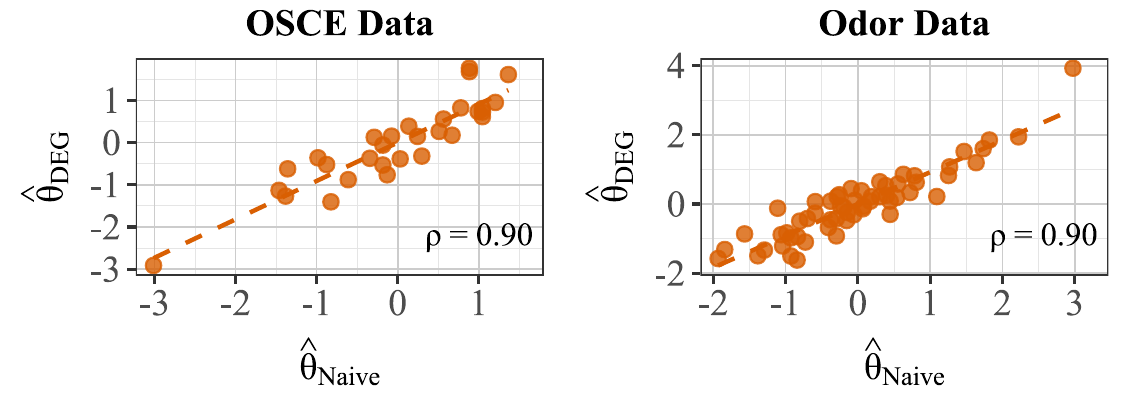}
    \caption{Spearman rank correlation between the Naive average score estimates and DEG model-based estimates of subject latent trait $\theta$ in the OSCE and Olfactory Assessment data sets.}
    \label{ScatterPlot_2}
\end{figure}

\begin{figure}
    \centering
    \includegraphics[scale=0.88]{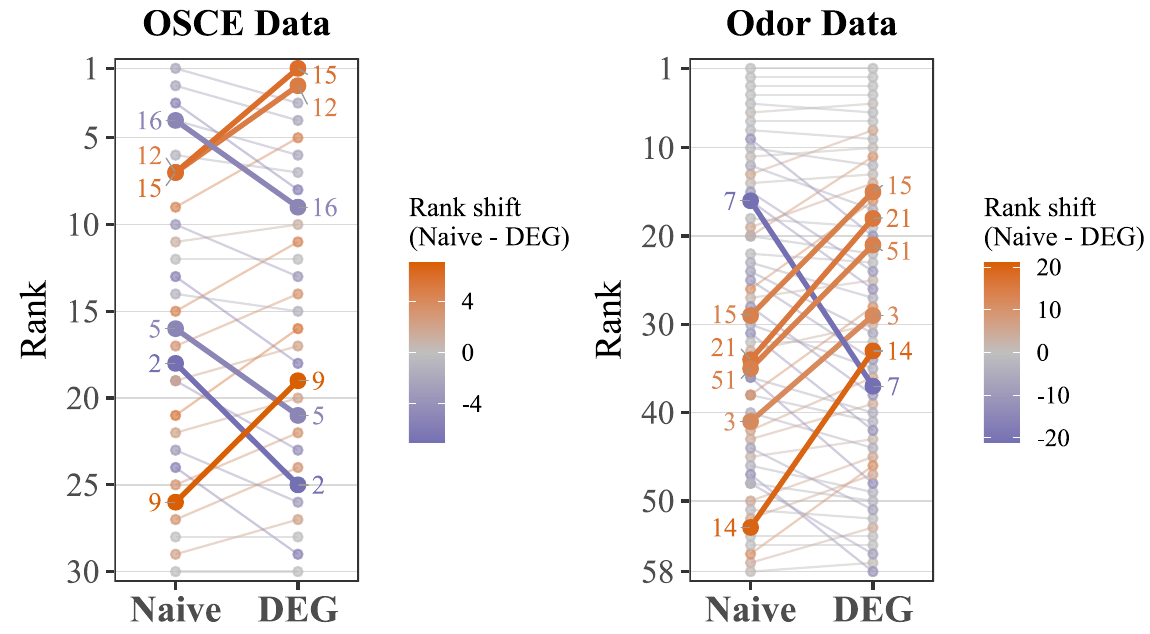}
    \caption{Different ranks implied, respectively, by the Naive average score estimates and DEG model-based estimates of subject latent trait $\theta$ in the OSCE and Olfactory Assessment data sets.}
    \label{rank_shift}
\end{figure}

\section{Additional Figures and Tables}\label{Additional_fig_tab}
We report in this Section additional Figures and Tables related to the \gls{degmfr} model formulation and the real data applications. 

\begin{figure}
    \centering
    \includegraphics[scale=0.45]{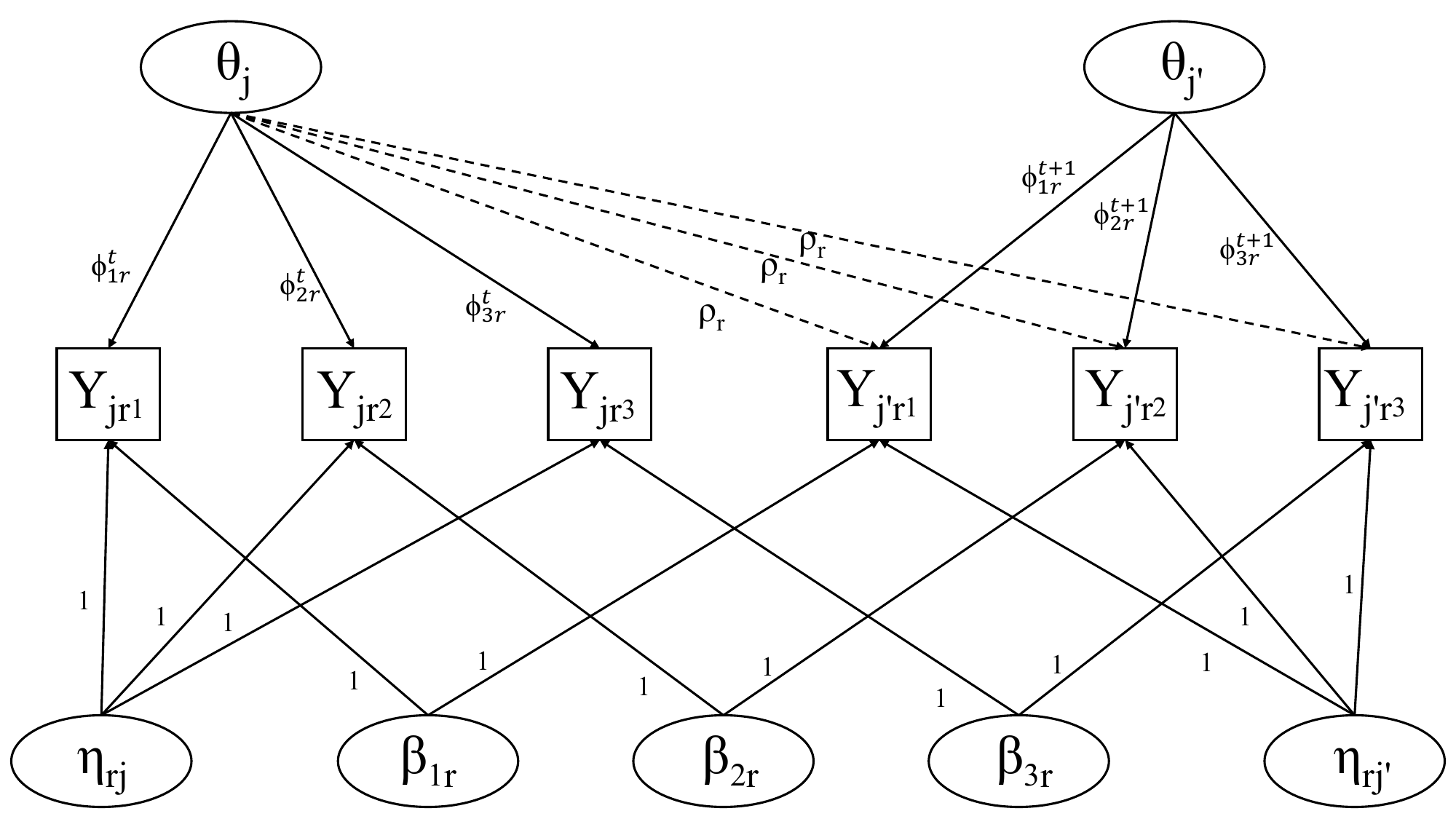}
    \caption{\small  Graphical representation of \gls{degmfr} model for two consecutive subjects $j,j'\in\mathcal{J}_r$ scored on three items by rater $r \in \mathcal{R}$. Here $log(\phi_{1r}^t)=\lambda_1+\xi_r+\nu^t$ and $log(\phi_{1r}^{t+1})=\lambda_1+\xi_r+\nu^{t+1}$; the dotted lines indicate the rater-specific cross-lagged effect $\rho_r$. Note that the solid and dotted lines do not imply linear effects. The other item-specific parameters and the category step parameters are omitted for graphical reasons.}

    \label{Example_path}
\end{figure}

\newcommand{\cmark}{\ding{51}}   
\newcommand{\xmark}{\ding{55}}   
\newcommand{\pmark}{$\circ$}     

\begin{table}[p]
\centering
\begin{adjustbox}{max width=\textwidth, max totalheight=0.95\textheight, keepaspectratio}
\begin{minipage}{\textwidth}
\singlespacing
\footnotesize
\begin{threeparttable}
\caption{Subject, item, and rater effects accounted for by the DEG-MFR model and by existing multi-rater IRT models.}
\label{TAB_comp}
\begin{tabularx}{\textwidth}{@{}>{\raggedright\arraybackslash}X
  >{\centering\arraybackslash}p{2.3cm}
  >{\centering\arraybackslash}p{2.3cm}
  >{\centering\arraybackslash}p{2.3cm}
  >{\centering\arraybackslash}p{2.3cm}@{}}
\toprule
Effect & DEG-MFR & HRTM\tnote{a} & MMR\tnote{b} & G-MFR\tnote{c} \\
 & (this paper) & Molenaar et al.\ (2021) & Verhelst \& Verstralen (2001) & Uto \& Ueno (2020) \\
\midrule
\multicolumn{5}{@{}l}{\textit{Subject}} \\
Latent trait
  & \cmark\ ($\theta_j$) & \cmark & \cmark & \cmark \\
Non-normal latent trait distribution
  & \cmark\ (DP/PYP mixture) & \xmark & \xmark & \xmark \\
Model-based subject clustering
  & \cmark & \xmark & \xmark & \xmark \\
\addlinespace
\multicolumn{5}{@{}l}{\textit{Subject $\times$ item}} \\
Subject--item interaction (quality of response)
  & \cmark\ ($\epsilon_{ij}$) & \cmark\ (continuous ideal score) & \cmark\ (continuous quality of response) & \xmark \\
Item-specific variance of subject--item effects
  & \cmark\ ($\varphi^2_i$) & \pmark\tnote{d} & \xmark\ (common variance) & \xmark \\
\addlinespace
\multicolumn{5}{@{}l}{\textit{Item}} \\
Item difficulty / category steps
  & \cmark\ ($\gamma_{ik}$) & \cmark\tnote{e} & \cmark & \pmark\ (difficulty only)\tnote{f} \\
Item discrimination
  & \cmark\ ($\lambda_i$) & \cmark & \pmark\tnote{g} & \cmark\ ($\alpha_i$) \\
\addlinespace
\multicolumn{5}{@{}l}{\textit{Rater}} \\
Rater severity
  & \cmark\tnote{h} & \cmark & \cmark & \cmark\ ($\beta_r$) \\
Rater $\times$ item severity
  & \cmark\ ($\beta_{ir}$) & \cmark & \xmark & \xmark \\
Rater $\times$ subject bias
  & \cmark\ ($\eta_{rj}$) & \xmark & \xmark & \xmark \\
Rater variability (variance of rater-specific errors)
  & \cmark\ ($\omega^2_r$) & \cmark & \xmark & \xmark \\
Rater discrimination / consistency (slope)
  & \cmark\ ($\xi_r$) & \pmark\tnote{i} & \xmark & \cmark\ ($\alpha_r$) \\
Joint item $\times$ rater discrimination
  & \cmark\ ($\log\phi_{ir}=\lambda_i+\xi_r$) & \pmark\tnote{i} & \xmark & \cmark\ ($\alpha_i\alpha_r$) \\
Rater-specific category thresholds (response style)
  & \cmark\ ($\tau_{rk}$) & \cmark\tnote{e} & \xmark & \cmark\ ($d_{rm}$) \\
\addlinespace
\multicolumn{5}{@{}l}{\textit{Dependence and rating dynamics}} \\
Dependence among multiple ratings of the same response\tnote{j}
  & \cmark & \cmark & \cmark & \xmark \\
Sequential anchoring on the previously rated subject
  & \cmark\ ($\rho_r$) & \xmark & \xmark & \xmark \\
Practice / fatigue effect on discrimination
  & \cmark\ ($\nu^t$, P-spline) & \xmark & \xmark & \xmark \\
\addlinespace
\multicolumn{5}{@{}l}{\textit{Inter-rater agreement}} \\
Model-based latent inter-rater consensus
  & \cmark\ ($\mathrm{LIC}_{rr'}$) & \xmark & \xmark & \xmark \\
\bottomrule
\end{tabularx}
\begin{tablenotes}[flushleft]
\footnotesize
\item \textit{Note.} \cmark\ = included; \pmark\ = partially included or available as an optional extension; \xmark\ = not included.
\item[a] Hierarchical Rater Thresholds Model, used here as the representative hierarchical rater model with continuous ideal scores. The original HRM \citepApp{Patz2002} uses discrete ideal scores and a signal-detection rater level with rater bias $\phi_r$ and variability $\psi_r$, without rater-specific thresholds.
\item[b] IRT Model for Multiple Raters. The original formulation is for dichotomous ratings; the polytomous version follows \citeApp{Mariano2007covariates}.
\item[c] Generalized Many-Facet Rasch model. The later G-MFR of \citeApp{Uto_24} additionally includes rater $\times$ item severity and item-specific category steps, and is nested in the EG-MFR when $\epsilon_{ij}=0$ and $\eta_{rj}=0$.
\item[d] Residual variance of the continuous ideal scores in the second-order factor formulation; whether it is item-specific depends on the identification constraints.
\item[e] Thresholds $\tau_{irc}$ are item-, rater-, and category-specific, so item step and rater threshold effects are not separated.
\item[f] Task difficulty $\beta_i$ only; category steps are rater-specific ($d_{rm}$), not item-specific.
\item[g] Not included in the original model; an item discrimination parameter may be added \citeApp{Mariano2007covariates}.
\item[h] No separate global severity parameter; the severity of rater $r$ on item $i$ is given by the local difficulty shift $\beta_{ir}$.
\item[i] Available as an extension with item-specific rater discrimination $\lambda_{ir}$ (or rater-specific $\lambda_r$).
\item[j] Ratings are conditionally independent given the subject--item quality of response rather than given $\theta$, addressing the information-accumulation problem of standard MFR models \citeApp{mariano2002critique}.
\end{tablenotes}
\end{threeparttable}
\end{minipage}
\end{adjustbox}
\end{table}

\begin{table}[ht]
\centering
\small
\begin{tabular}{@{}c p{12cm}@{}}
\hline
Item & Evaluation Item \\
\hline
1  & Greeting \\
2  & Self-introduction \\
3  & Patient identification \\
4  & Medical interview explanation and consent \\
5  & Confirmation of chief complaint \\
6  & Confirmation of affected site(s) \\
7  & Confirmation of when the complaint started \\
8  & Confirmation of current symptoms (spontaneous pain, swelling, redness, agitation, purulent discharge) \\
9  & Confirmation of evoked pain (cold pain, warm pain, occlusion pain) \\
10 & Confirmation of the nature and severity of the evoked pain \\
11 & Confirmation of history of hospital visits \\
12 & Confirmation of whether there is any medication for the chief complaint \\
13 & Confirmation of consultation behavior \\
14 & Confirmation of symptom history up to the present day \\
15 & Confirmation of history of anesthesia and tooth extraction and any abnormalities during those procedures \\
16 & Confirmation of whether any systemic disease is present \\
17 & Confirmation of necessary information for each disease (whether potential risks for dental treatment could be identified) \\
18 & Confirmation of allergies (food, drugs) \\
19 & Summarization and confirmation of the chief complaint \\
20 & Confirmation of whether there is anything the examinee has forgotten to say \\
21 & Cleanliness during examination \\
22 & Efforts to understand the patient's situation (whether the examinee has confirmed the patient's wishes, status of hospital visits, financial situation, etc.) \\
23 & Smooth flow of conversation \\
24 & Appropriateness of the main points of the conversation (whether the judgment on the main points of the conversation was appropriate according to the patient's situation) \\
25 & Eye contact, pauses, listening \\
26 & Appropriateness of open questions \\
27 & Respectful language, attention to terminology \\
28 & Interaction with the patient \\
29 & Listening to the patient's concerns \\
30 & Overall rating \\
\hline
\end{tabular}
\caption{Evaluation items in the OSCE scoring rubric, reproduced from \protect\citeApp{Uto_24}}
\label{Tab_OSCE_items}
\end{table}

\begin{table}
\centering
\begin{tabular}{l c c | l c c | l c c}
\hline
 & EAP & 95\% CI & & EAP & 95\% CI & & EAP & 95\% CI \\
\hline
  $\theta_{1}$ & -0.040 & (-0.528,\;0.479) & $\xi_{1}$ & 0.426 & (0.201,\;0.667) & $\lambda_{1}$ & 0.024 & (-0.803,\;0.926) \\
  $\theta_{2}$ & 0.001 & (-0.488,\;0.505) & $\xi_{2}$ & -0.228 & (-0.487,\;0.030) & $\lambda_{2}$ & 0.521 & (-0.385,\;1.568) \\
  $\theta_{3}$ & 0.330 & (-0.178,\;0.855) & $\xi_{3}$ & -0.108 & (-0.383,\;0.165) & $\lambda_{3}$ & 0.659 & (-0.123,\;1.469) \\
  $\theta_{4}$ & 0.571 & (0.022,\;1.161) & $\xi_{4}$ & -0.339 & (-0.586,\;-0.098) & $\lambda_{4}$ & -0.028 & (-0.724,\;0.814) \\
  $\theta_{5}$ & 0.137 & (-0.443,\;0.721) & $\xi_{5}$ & 0.250 & (0.047,\;0.470) & $\lambda_{5}$ & 0.748 & (0.023,\;1.528) \\
  $\theta_{6}$ & 0.089 & (-0.403,\;0.588) & $\omega^2_{1}$ & 0.030 & (0.004,\;0.084) & $\lambda_{6}$ & -0.125 & (-0.795,\;0.583) \\
  $\theta_{7}$ & 0.476 & (-0.128,\;1.059) & $\omega^2_{2}$ & 0.286 & (0.036,\;0.897) & $\lambda_{7}$ & 0.348 & (-0.455,\;1.294) \\
  $\theta_{8}$ & -0.771 & (-1.417,\;-0.221) & $\omega^2_{3}$ & 0.379 & (0.069,\;1.116) & $\lambda_{8}$ & 0.308 & (-0.289,\;0.933) \\
  $\theta_{9}$ & 0.146 & (-0.423,\;0.746) & $\omega^2_{4}$ & 1.018 & (0.223,\;2.945) & $\lambda_{9}$ & -0.559 & (-1.378,\;0.194) \\
  $\theta_{10}$ & 0.142 & (-0.434,\;0.719) & $\omega^2_{5}$ & 0.033 & (0.005,\;0.096) & $\lambda_{10}$ & -0.172 & (-0.732,\;0.393) \\
  $\theta_{11}$ & 0.328 & (-0.169,\;0.840) & $\rho_{1}$ & -0.193 & (-0.732,\;0.276) & $\lambda_{11}$ & -0.233 & (-0.739,\;0.336) \\
  $\theta_{12}$ & 0.880 & (0.312,\;1.488) & $\rho_{2}$ & 0.215 & (-0.946,\;1.331) & $\lambda_{12}$ & 0.180 & (-0.517,\;1.022) \\
  $\theta_{13}$ & 0.541 & (-0.027,\;1.159) & $\rho_{3}$ & -0.429 & (-1.452,\;0.418) & $\lambda_{13}$ & 0.294 & (-0.438,\;1.171) \\
  $\theta_{14}$ & 0.855 & (0.296,\;1.451) & $\rho_{4}$ & -0.560 & (-2.061,\;0.840) & $\lambda_{14}$ & -0.314 & (-0.852,\;0.196) \\
  $\theta_{15}$ & 0.910 & (0.317,\;1.569) & $\rho_{5}$ & -0.348 & (-0.967,\;0.185) & $\lambda_{15}$ & -0.190 & (-0.867,\;0.502) \\
  $\theta_{16}$ & 0.497 & (-0.126,\;1.172) & $\tau_{1,1}$ & -1.127 & (-1.593,\;-0.647) & $\lambda_{16}$ & -0.610 & (-1.205,\;-0.052) \\
  $\theta_{17}$ & 0.085 & (-0.454,\;0.615) & $\tau_{2,1}$ & -1.758 & (-2.367,\;-1.170) & $\lambda_{17}$ & 0.562 & (-0.142,\;1.488) \\
  $\theta_{18}$ & 0.161 & (-0.371,\;0.729) & $\tau_{3,1}$ & -1.232 & (-1.809,\;-0.663) & $\lambda_{18}$ & -0.296 & (-0.925,\;0.374) \\
  $\theta_{19}$ & -0.131 & (-0.652,\;0.380) & $\tau_{4,1}$ & -1.712 & (-2.386,\;-1.068) & $\lambda_{19}$ & -0.374 & (-1.009,\;0.295) \\
  $\theta_{20}$ & 0.052 & (-0.462,\;0.574) & $\tau_{5,1}$ & -1.322 & (-1.797,\;-0.839) & $\lambda_{20}$ & -0.201 & (-0.839,\;0.519) \\
  $\theta_{21}$ & 0.371 & (-0.172,\;0.921) & $\tau_{1,2}$ & -1.332 & (-1.768,\;-0.940) & $\lambda_{21}$ & -0.050 & (-0.783,\;0.931) \\
  $\theta_{22}$ & 0.565 & (-0.047,\;1.191) & $\tau_{2,2}$ & -1.319 & (-1.826,\;-0.821) & $\lambda_{22}$ & -0.852 & (-1.369,\;-0.343) \\
  $\theta_{23}$ & -0.229 & (-0.792,\;0.333) & $\tau_{3,2}$ & -0.389 & (-0.910,\;0.102) & $\lambda_{23}$ & 0.884 & (0.201,\;1.643) \\
  $\theta_{24}$ & 0.415 & (-0.161,\;0.983) & $\tau_{4,2}$ & -0.832 & (-1.347,\;-0.323) & $\lambda_{24}$ & 0.177 & (-0.609,\;0.909) \\
  $\theta_{25}$ & 0.339 & (-0.223,\;0.920) & $\tau_{5,2}$ & -0.927 & (-1.338,\;-0.533) & $\lambda_{25}$ & 0.037 & (-0.765,\;0.791) \\
  $\theta_{26}$ & 0.254 & (-0.326,\;0.807) & $\tau_{1,3}$ & 2.459 & (1.870,\;3.070) & $\lambda_{26}$ & 0.316 & (-0.483,\;1.347) \\
  $\theta_{27}$ & 0.618 & (0.033,\;1.249) & $\tau_{2,3}$ & 3.077 & (2.347,\;3.811) & $\lambda_{27}$ & -0.509 & (-1.101,\;0.153) \\
  $\theta_{28}$ & 0.321 & (-0.240,\;0.914) & $\tau_{3,3}$ & 1.621 & (0.985,\;2.282) & $\lambda_{28}$ & -0.065 & (-0.719,\;0.581) \\
  $\theta_{29}$ & 0.536 & (-0.056,\;1.143) & $\tau_{4,3}$ & 2.545 & (1.874,\;3.247) & $\lambda_{29}$ & -0.453 & (-1.072,\;0.198) \\
  $\theta_{30}$ & -0.180 & (-0.768,\;0.385) & $\tau_{5,3}$ & 2.249 & (1.673,\;2.859) & $\lambda_{30}$ & -0.029 & (-0.823,\;0.803) \\
\hline
\end{tabular}
\caption{\small Expected A Posteriori (EAP) and 95\% credible interval of selected model parameters (DEG-MFR PYP - $d=0.3$ model) from OSCE data. }
\label{Tab_posteriors_OSCE}
\end{table}

\begin{figure}
    \centering
    \includegraphics[scale= 0.68]{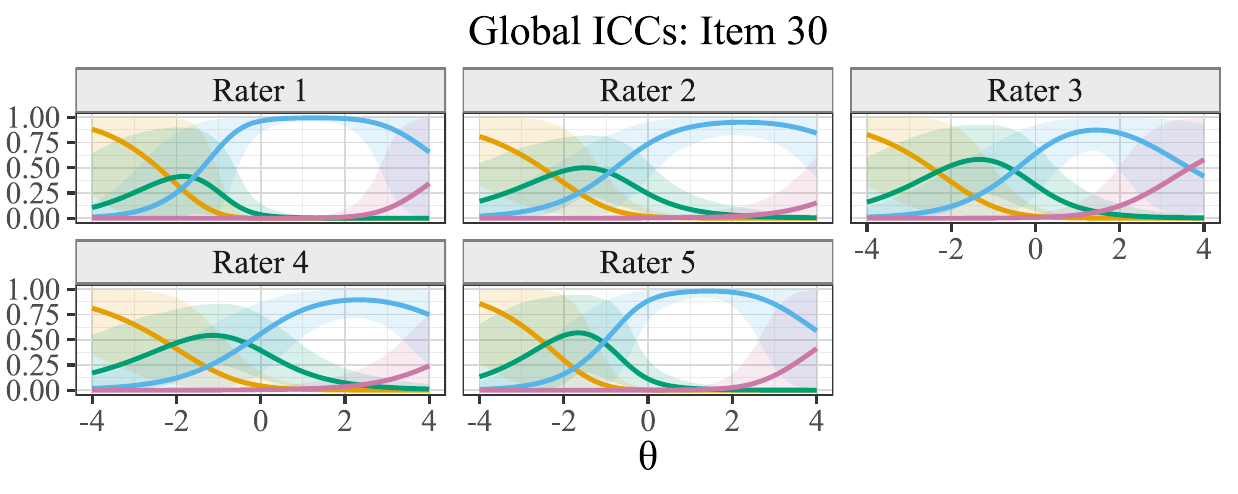}
    \caption{ \small Global Item-Rater Characteristic Curves for item 30 (\textit{Overall rating}) across subjects for the OSCE data, controlling for rater dynamic effects. The orange, green, aquamarine, and violet lines represent the four ordered response categories; the corresponding solid lines and shaded areas indicate, respectively, the Expected A Posteriori and the 95\% quantile-based credible interval for the category response curves.}
    \label{ICC_1_bis}
\end{figure}

\begin{figure}
    \centering
\includegraphics[scale= 0.73]{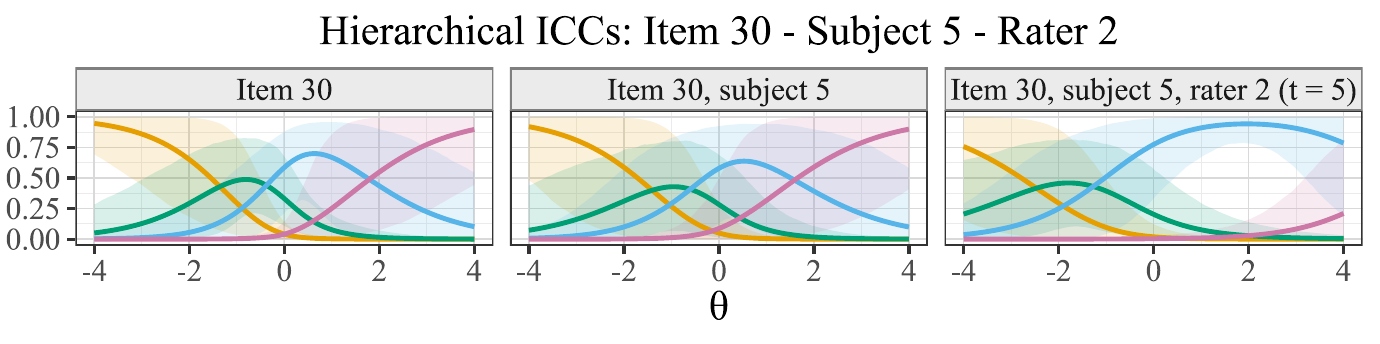}

   \caption{\small Hierarchical Item-Rater Characteristic Curves for item 30 (\textit{Overall rating}) from OSCE data. Left panel: item-level curves, depending only on the item discrimination $\lambda_i$ and the item category step parameters. Center: curves for subject 5, adding the subject-specific item effect $\epsilon_{ij}$. Right panel: curves for subject 5 as rated by rater 2 at position $t$ of the rating sequence, including all rater-specific and dynamic components, i.e., the local difficulty shift $\beta_{ir}$, the rater category step parameters, the local bias $\eta^t_{rj}$ (with the anchoring effect $\rho_r\theta^{t-1}$), and the dynamic joint discrimination $\phi^t_{ir}$. The orange, green, aquamarine, and violet lines represent the four ordered response categories; the corresponding solid lines and shaded areas indicate, respectively, the Expected A Posteriori and the 95\% quantile-based credible interval for the category response curves.}
    \label{ICC_H}
\end{figure}

\begin{table}[ht]
\centering
\begin{tabular}{l c c | l c c | l c c}
\hline
 & EAP & 95\% CI & & EAP & 95\% CI & & EAP & 95\% CI \\
\hline
  $\theta_{1}$ & -0.374 & (-0.816,\;0.012) & $\rho_{3}$ & 0.472 & (-1.085,\;1.909) & $\omega_{1}$ & 0.427 & (0.036,\;1.702) \\
  $\theta_{2}$ & -0.522 & (-0.982,\;-0.097) & $\rho_{6}$ & -0.013 & (-1.184,\;1.151) & $\omega_{2}$ & 0.532 & (0.033,\;2.193) \\
  $\theta_{3}$ & -0.186 & (-0.746,\;0.338) & $\rho_{64}$ & 0.702 & (-0.988,\;2.294) & $\omega_{3}$ & 0.609 & (0.107,\;1.930) \\
  $\theta_{4}$ & -0.435 & (-0.904,\;0.045) & $\rho_{74}$ & 0.488 & (-1.227,\;2.120) & $\omega_{4}$ & 0.710 & (0.127,\;2.096) \\
  $\theta_{5}$ & -0.647 & (-1.212,\;-0.116) & $\rho_{77}$ & -0.449 & (-1.883,\;0.984) & $\omega_{5}$ & 0.281 & (0.023,\;1.161) \\
  $\theta_{6}$ & -0.297 & (-0.815,\;0.192) & $\rho_{88}$ & -0.305 & (-2.004,\;1.464) & $\omega_{6}$ & 0.273 & (0.036,\;0.946) \\
  $\theta_{7}$ & -0.309 & (-0.783,\;0.235) & $\lambda_{1}$ & 0.977 & (0.844,\;1.130) & $\omega_{7}$ & 0.266 & (0.026,\;0.994) \\
  $\theta_{8}$ & 0.004 & (-0.390,\;0.378) & $\lambda_{2}$ & 0.824 & (0.705,\;0.962) & $\omega_{8}$ & 0.430 & (0.047,\;1.469) \\
  $\theta_{9}$ & -0.733 & (-1.253,\;-0.244) & $\lambda_{3}$ & 0.846 & (0.727,\;0.980) & $\omega_{9}$ & 0.264 & (0.023,\;0.999) \\
  $\theta_{10}$ & -0.082 & (-0.501,\;0.324) & $\lambda_{4}$ & 1.491 & (1.236,\;1.839) & $\omega_{10}$ & 0.748 & (0.093,\;2.477) \\
  $\theta_{11}$ & -0.361 & (-0.842,\;0.094) & $\xi_{1}$ & 0.645 & (0.222,\;1.397) & $\omega_{11}$ & 1.006 & (0.109,\;4.092) \\
  $\theta_{12}$ & 0.065 & (-0.414,\;0.563) & $\xi_{2}$ & 0.354 & (0.161,\;0.653) & $\omega_{12}$ & 0.584 & (0.078,\;1.912) \\
  $\theta_{13}$ & 0.436 & (-0.061,\;0.946) & $\xi_{3}$ & 1.336 & (0.542,\;2.712) & $\omega_{13}$ & 0.275 & (0.019,\;1.170) \\
  $\theta_{14}$ & -0.250 & (-0.762,\;0.236) & $\xi_{4}$ & 1.410 & (0.609,\;2.757) & $\omega_{14}$ & 0.431 & (0.033,\;1.690) \\
  $\theta_{15}$ & -0.061 & (-0.358,\;0.202) & $\xi_{5}$ & 0.948 & (0.361,\;1.931) & $\omega_{15}$ & 0.319 & (0.035,\;1.174) \\
  $\theta_{16}$ & -0.486 & (-0.968,\;-0.035) & $\xi_{6}$ & 1.822 & (0.612,\;3.974) & $\omega_{16}$ & 0.844 & (0.088,\;3.034) \\
  $\theta_{17}$ & 0.295 & (-0.183,\;0.785) & $\xi_{7}$ & 0.726 & (0.350,\;1.341) & $\omega_{17}$ & 1.031 & (0.156,\;3.584) \\
  $\theta_{18}$ & 0.076 & (-0.353,\;0.480) & $\xi_{8}$ & 0.893 & (0.373,\;1.783) & $\omega_{18}$ & 0.215 & (0.017,\;0.940) \\
  $\theta_{19}$ & -0.139 & (-0.550,\;0.285) & $\tau_{1,1}$ & -0.207 & (-1.200,\;0.747) & $\omega_{19}$ & 1.698 & (0.276,\;5.733) \\
  $\theta_{20}$ & -0.230 & (-0.713,\;0.276) & $\tau_{2,1}$ & 0.663 & (-0.484,\;1.942) & $\omega_{20}$ & 1.643 & (0.321,\;5.061) \\
  $\theta_{21}$ & -0.121 & (-0.601,\;0.373) & $\tau_{3,1}$ & -0.517 & (-1.222,\;0.164) & $\omega_{21}$ & 0.371 & (0.025,\;1.474) \\
  $\theta_{22}$ & -0.099 & (-0.365,\;0.243) & $\tau_{1,2}$ & -0.244 & (-1.305,\;0.727) & $\omega_{22}$ & 0.556 & (0.058,\;1.919) \\
  $\theta_{23}$ & -0.349 & (-0.812,\;0.201) & $\tau_{2,2}$ & 0.174 & (-1.101,\;1.508) & $\omega_{23}$ & 0.357 & (0.034,\;1.315) \\
  $\theta_{24}$ & 0.062 & (-0.390,\;0.521) & $\tau_{3,2}$ & -0.498 & (-1.236,\;0.178) & $\omega_{24}$ & 3.129 & (0.506,\;10.988) \\
  $\theta_{25}$ & -0.747 & (-1.222,\;-0.247) & $\tau_{1,3}$ & 0.114 & (-0.926,\;1.014) & $\omega_{25}$ & 2.021 & (0.232,\;7.611) \\
  $\theta_{26}$ & 0.322 & (-0.183,\;0.834) & $\tau_{2,3}$ & 0.070 & (-1.329,\;1.490) & $\omega_{26}$ & 2.759 & (0.456,\;8.997) \\
  $\theta_{27}$ & -0.497 & (-1.127,\;0.049) & $\tau_{3,3}$ & -0.414 & (-1.169,\;0.183) & $\omega_{27}$ & 0.972 & (0.064,\;4.015) \\
  $\theta_{28}$ & -0.613 & (-1.192,\;-0.081) & $\tau_{1,4}$ & 0.337 & (-0.699,\;1.617) & $\omega_{28}$ & 1.327 & (0.276,\;4.054) \\
  $\theta_{29}$ & 0.147 & (-0.285,\;0.640) & $\tau_{2,4}$ & -0.907 & (-2.666,\;0.581) & $\omega_{29}$ & 0.958 & (0.186,\;2.965) \\
  $\theta_{30}$ & -0.706 & (-1.341,\;-0.112) & $\tau_{3,4}$ & 1.429 & (0.391,\;2.707) & $\omega_{30}$ & 1.274 & (0.154,\;4.635) \\
\hline
\end{tabular}
\caption{\small Expected A Posteriori (EAP) and 95\% credible interval of selected model parameters (DEG-MFR DP model) from Odors Assessment data.}
\label{Tab_posteriors_Odor}
\end{table}

\begin{table}
\centering
\small
\begin{tabular}{l c c c c c c}
\toprule
 & Odor $j$ & $(t_r(j),\,t_{r'}(j))$ & Preceding odors & $\rho_r\rho_{r'}$ & EAP & 95\% CI \\
\midrule
\multicolumn{7}{@{}l}{\textit{Both raters rate odor $j$ first}} \\
  $LIC_{1,3}$ & 15 & (1, 1) & -- & -0.095 & 0.370 & (0.137, 0.673) \\
  $LIC_{1,4}$ & 15 & (1, 1) & -- & 0.099 & 0.352 & (0.126, 0.657) \\
  $LIC_{1,2}$ & 15 & (1, 1) & -- & 0.044 & 0.410 & (0.135, 0.760) \\
\midrule
\multicolumn{7}{@{}l}{\textit{Both raters rated the same odor immediately before $j$}} \\
  $LIC_{1,3}$ & 22 & (2, 2) & (15, 15) & -0.095 & 0.280 & (-0.119, 0.656) \\
  $LIC_{1,4}$ & 22 & (2, 2) & (15, 15) & 0.099 & 0.326 & (0.001, 0.696) \\
  $LIC_{1,2}$ & 22 & (2, 2) & (15, 15) & 0.044 & 0.334 & (-0.105, 0.756) \\
\midrule
\multicolumn{7}{@{}l}{\textit{The raters rated different odors immediately before $j$}} \\
  $LIC_{1,3}$ & 13 & (7, 7) & (56, 48) & -0.095 & 0.304 & (0.121, 0.556) \\
  $LIC_{1,4}$ & 13 & (7, 7) & (56, 48) & 0.099 & 0.301 & (0.114, 0.570) \\
  $LIC_{1,5}$ & 22 & (2, 7) & (15, 56) & -0.035 & 0.386 & (0.144, 0.691) \\
\bottomrule
\end{tabular}
\caption{\small Expected A Posteriori (EAP) and 95\% credible intervals of the pairwise latent inter-rater consensus $LIC_{rr'}$ (Eq.~(\ref{ICC})) between rater 1 and a sample of other raters from Odor Assessment data, under three rating configurations of a common odor $j$: both raters rate $j$ first ($\tilde\rho_r=\tilde\rho_{r'}=0$); both raters rated the same odor immediately before $j$ ($s_{rr'}(j)=1$); the raters rated different odors immediately before $j$ ($s_{rr'}(j)=0$). Positions $(t_r(j),t_{r'}(j))$ and preceding odors are reported for $r=1$ and $r'$, respectively; $\rho_r\rho_{r'}$ is the EAP of the product of the anchoring effects.}
\label{post_summary_LIC_odors}
\end{table}

\begin{figure}
    \centering

  \includegraphics[scale= 0.75]{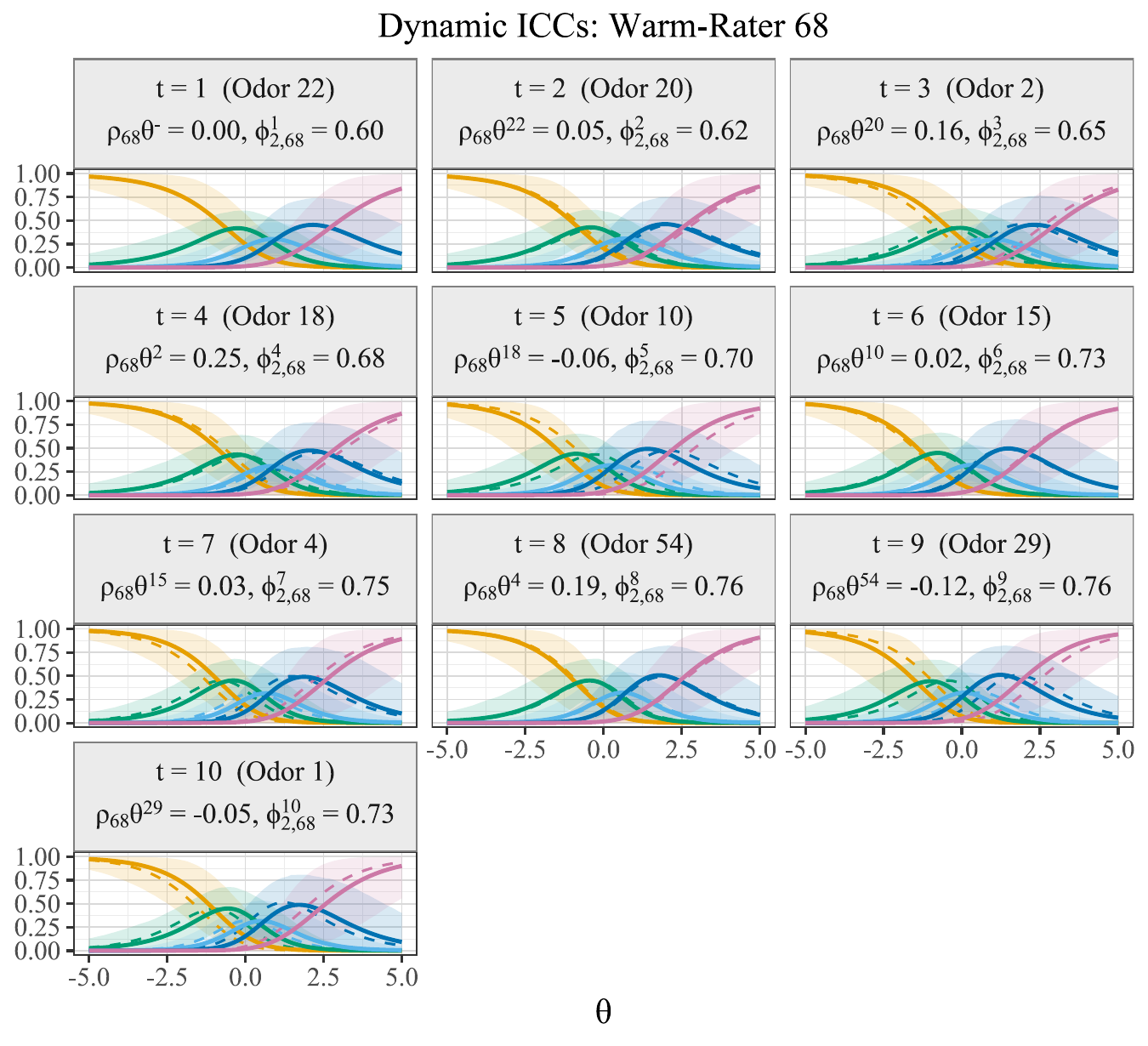}
 \caption{\small Dynamic Item-Rater Characteristic Curve for \textit{warm} and rater 68 across the sequence of odors assessed, from Olfactory Assessment data. The orange, green, aquamarine, blue, and violet lines represent the five ordered response categories; the corresponding solid lines and shaded areas indicate the Expected A Posteriori and the 95\% quantile-based credible interval for the category response curves, respectively. Dashed lines reproduce the Expected A Posteriori curves of the odor rated at the previous position $t-1$, so that each panel shows how the location and shape of the curves change from one odor to the next. Each panel title reports the posterior mean of the anchoring effect $\rho_r\theta^{t-1}$ and of the dynamic joint discrimination $\phi^t_{ir}$. Curves are computed controlling for the odor-specific item residual $\epsilon_{ij}$, so that differences across panels reflect only rater-related effects. }
    \label{ICC_Odor_Dynamic_2}
\end{figure}

\begin{figure}
    \centering
    \includegraphics[scale=0.85]{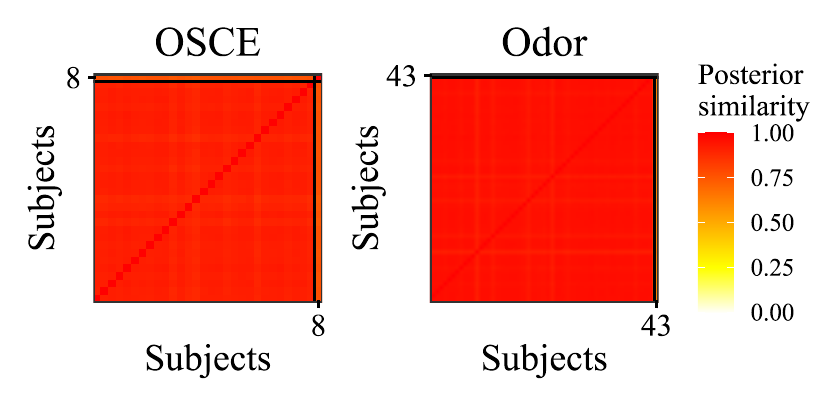}
    \caption{Posterior similarity matrix based on the selected models from OSCE (left panel) and Odor data (right panel).  }
    \label{PSM_real_data_set}
\end{figure}

\section{Further Extensions}\label{sec:extensions}
In this section, we shall present two straightforward extensions of the \gls{degmfr} model. The first extension allows the joint modeling of rater-specific characteristics and drawing population-level inferences. Through the second extension, contextual information might be incorporated within the model to inform the response probability. 
\subsection{Raters Structural Model}\label{A-A}
When the sample of raters is large enough, we might specify a joint distribution over the raters' characteristics and estimate the population-level parameters.
We might jointly model the dependency between their global parameters and let them be correlated \citepApp{Mignemi_Chen_Moustaki_2025}. We shall consider the \gls{degmfr} model in this extension, as the \gls{egmfr} is clearly a nested version. \\
The measurement model specification presented in Section \ref{measurement:dynamic} introduces $K+3$ global parameters for each rater $r \in \mathcal{R}$, namely, rater consistency $\omega_r^2$, rater discrimination $\xi_r$, the anchoring bias $\rho_r$, and the $K$ category step parameters $\tau_{rk}$, for $k=1,\dots,K$, of which $K-1$ are free under the zero-sum constraint. We jointly model the rater's global parameters and assume that $(log(\omega_r^2),\xi_r, \rho_r,\tau_{r1},\dots,\tau_{rK-1})$, $r \in \mathcal{R}$, are i.i.d. from a multivariate normal distribution with mean vector $\boldsymbol{\mu}$ and covariance matrix $\boldsymbol{\Sigma}$: 
\begin{equation}\label{raters_features}
    (log(\omega_r^2),\xi_r, \rho_r,\tau_{r},\dots,\tau_{r, K-1})\;|\;\boldsymbol{\mu}, \boldsymbol{\Sigma} \; \overset{\mathrm{iid}}{\sim} \; N(\boldsymbol{\mu}, \boldsymbol{\Sigma}), \quad r=1,\dots, R.
\end{equation}
Here $\boldsymbol{\mu}=(\mu_1,\dots,\mu_{K+2})$ represents the expected values of the global parameters of the raters in their population; the positive definite matrix $\boldsymbol{\Sigma}=(\sigma_{mn})_{{(K+2)} \times (K+2)}$ captures the variability of these global parameters and their relationships in the population of raters. For instance, the parameter $\mu_1$ indicates the average log-variance of the local biases ($\log\omega_r^2$) across raters, $\mu_2$ is their average discrimination ($\xi_r$), and $\mu_3$ is the average anchoring tendency ($\rho_r$). The parameter $\sigma_{11}^2$ is the variance of $\log\omega_r^2$across raters, and $\sigma_{22}^2$ represents the variance of $\xi_r$. The parameter $\sigma_{12}$ captures the covariance between $\log\omega_r^2$ and $\xi_r$; $\sigma_{13}$ represents the covariance between $\log\omega_r^2$ and $\rho_r$.

\paragraph{Posterior Computation and Results.}
We present in this section an application of the \gls{degmfr} Model with a joint modeling of the rater-specific parameters as detailed in Eq.~(\ref{raters_features}).  For illustrative purposes, we analyze the Olfactory data presented in Section \ref{Odor} only considering \textit{pleasantness}, \textit{familiarity}, and \textit{edibility} for computational tractability. Hence, we note that this is not exactly the same dataset analyzed in Section \ref{Odor}.

The rater structural model detailed in Appendix~\ref{A-A} yields the following posterior of the model parameters:

\begin{equation}\label{joint_posterior_MVR}
\begin{aligned}
p(\boldsymbol{S} \mid \mathbf{y}) \propto & \prod_{r=1}^{R} \prod_{i=1}^I  p\Big(y^1_{ir} \;\Big|\; \theta^{b_r(1)}, \, \epsilon_{i}^{b_r(1)}, \, \beta_{ir}, \, \boldsymbol{\psi_r}, \, \boldsymbol{\gamma}_{i}, \, \eta_r^1,  \, \lambda_{i}, \, \boldsymbol{a}\Big) \\
& \times\prod_{t=2}^{J_r} p\Big(y^t_{ir} \;\Big|\; \theta^{b_r(t)}, \, \theta^{b_r(t-1)}, \, \epsilon_{i}^{b_r(t)}, \, \beta_{ir}, \, \boldsymbol{\psi_r}, \, \boldsymbol{\gamma}_{i}, \, \eta_r^t,  \, \lambda_{i}, \, \boldsymbol{a}\Big) \\
& \times \prod_{j=1}^{N} \left[ p(\theta_j \mid \mu_{z_j}, \sigma^2_{z_j}) \cdot p(z_j \mid \mathbf{\pi}) \right] \times p(\mathbf{\pi} \mid \boldsymbol{v}) \times \left[ \prod_{c=1}^{C-1} p(v_c \mid \alpha, d) \right] \times p(\alpha) \\
& \times \prod_{c=1}^{C} \left[ p(\mu_c) \cdot p(\sigma^2_c) \right] \times p(\boldsymbol{a} \mid b) \cdot p(b) \\
& \times \prod_{i=1}^{I} \left[ p(\lambda_i) \cdot p(\varphi_i^2) \cdot \prod_{j=1}^{N} p(\epsilon_{ij} \mid \varphi_i^2) \right] \\  
& \times \prod_{r=1}^{R} p\left(\eta_{r}^{b_r(1)} \mid {r} \right) \prod_{t=2}^{J_r} p\left(\eta_{r}^{b_r(t)} \mid {r},\rho_r, \theta^{b_r(t-1)}\right) \\
& \times \prod_{i=1}^{I}\prod_{r=1}^{R} p(\beta_{i,r}) \times \prod_{i=1}^{I} p(\boldsymbol{\gamma}_i) \\ 
& \times \prod_{r=1}^{R} p\left(\boldsymbol{\psi}_r \mid \boldsymbol{\mu}_\psi, \boldsymbol{\Sigma}\right) \times p(\boldsymbol{\mu}) \times p\left(\boldsymbol{\Sigma}^{-1}\right), 
\end{aligned}
\end{equation}
Where $\boldsymbol{\psi}_r =(log(\omega_r^2),\xi_r, \rho_r,\tau_{r1},\dots,\tau_{rK-1})$ is the $q$-dimensional vector of rater-specific parameter, with $q=K+2$.\\
We place the following hyperpriors over the mean vector $\boldsymbol{\mu}$ and the precision matrix $\boldsymbol{\Sigma}^{-1}$ of the multivariate normal:
\[
\boldsymbol{\mu} \sim N\!\left(\mathbf{0},\, 100\,\mathbf{I}_q\right), 
\qquad 
\boldsymbol{\Sigma}^{-1} \sim \mathcal{W}\!\left(\mathbf{I}_q,\, q+2\right).
\]
The prior on $\boldsymbol{\mu}$ is highly diffuse, with variance 100 
on each component. For the Wishart, setting the second parameter to $q+2$ is the minimum degrees of freedom, ensuring a proper distribution, yielding a weakly informative prior on $\boldsymbol{\Sigma}$ with prior mean $\mathbb{E}[\boldsymbol{\Sigma}^{-1}] = (q+2)\mathbf{I}_q$.\\
The rater-specific vector $\boldsymbol{\psi}_r$ is updated jointly using 
an adaptive block random-walk Metropolis--Hastings sampler, with full adaptation of both the scale and the proposal covariance matrix. The hyperparameters $\boldsymbol{\mu}_\psi$ and $\boldsymbol{\Sigma}^{-1}$ admit conjugate updates---multivariate normal and Wishart, respectively---which are exploited for exact sampling.

 We report in this paragraph the most interesting results on odor posterior density and rater practice effect P-Spline in Figure \ref{Density_P-Spline_MVN_Odor}, the posterior similarity matrix in Figure \ref{PSM_MVN} and the posterior of the rater features covariance matrix in Table \ref{Table_Cor_Mat_MVN_Odor} with the respective implied correlation matrix given in Figure \ref{Cor_Mat_MVN_Odor}. 
 
 The posterior mean correlation matrix of the rater-specific parameter vector $(\log(\omega_r^2), \xi_r, \rho_r, \tau_{r1}, \tau_{r2}, \tau_{r3})^\top$, based on the covariance matrix $\boldsymbol{\Sigma}$, is displayed in 
Figure~\ref{Cor_Mat_MVN_Odor}, and the covariance values are reported numerically in 
Table~\ref{Table_Cor_Mat_MVN_Odor}. Figure \ref{Cor_Mat_MVN_Odor} displays a negative moderate posterior mean correlation between the anchoring effect $\rho_r$ and the rater-specific category thresholds $\tau_{r1}$ and $\tau_{r2}$, and a small positive posterior mean correlation with the third category threshold $\tau_{r3}$. 
A rather strong positive posterior mean correlation is shown between $\tau_{r1}$ and $\tau_{r2}$, whose credible interval excludes zero,  while the weaker negative correlations of both with $\tau_{r3}$ are less certain. 

Figure \ref{Density_P-Spline_MVN_Odor} shows a positively skewed latent trait posterior density, with six odors clustered into different singletons.

\begin{table}[ht]
\centering
\begin{tabular}{l c c}
\hline
 & EAP & 95\% CI \\
\hline
  $\sigma^2(\rho_r)$ & 0.161 & (0.073,\;0.317) \\
  $\sigma(\rho_r,\,\log\omega_r^2)$ & -0.011 & (-0.112,\;0.073) \\
  $\sigma(\rho_r,\,\xi_r)$ & 0.007 & (-0.054,\;0.077) \\
  $\sigma(\rho_r,\,\tau_{r1})$ & -0.062 & (-0.189,\;0.053) \\
  $\sigma(\rho_r,\,\tau_{r2})$ & -0.039 & (-0.116,\;0.027) \\
  $\sigma(\rho_r,\,\tau_{r3})$ & 0.016 & (-0.054,\;0.090) \\
  $\sigma^2(\log\omega_r^2)$ & 0.161 & (0.067,\;0.331) \\
  $\sigma(\log\omega_r^2,\,\xi_r)$ & -0.031 & (-0.120,\;0.031) \\
  $\sigma(\log\omega_r^2,\,\tau_{r1})$ & 0.014 & (-0.151,\;0.151) \\
  $\sigma(\log\omega_r^2,\,\tau_{r2})$ & -0.005 & (-0.093,\;0.071) \\
  $\sigma(\log\omega_r^2,\,\tau_{r3})$ & -0.014 & (-0.096,\;0.064) \\
  $\sigma^2(\xi_r)$ & 0.115 & (0.057,\;0.212) \\
  $\sigma(\xi_r,\,\tau_{r1})$ & 0.023 & (-0.068,\;0.118) \\
  $\sigma(\xi_r,\,\tau_{r2})$ & 0.011 & (-0.040,\;0.066) \\
  $\sigma(\xi_r,\,\tau_{r3})$ & 0.007 & (-0.047,\;0.065) \\
  $\sigma^2(\tau_{r1})$ & 0.418 & (0.249,\;0.657) \\
  $\sigma(\tau_{r1},\,\tau_{r2})$ & 0.145 & (0.054,\;0.266) \\
  $\sigma(\tau_{r1},\,\tau_{r3})$ & -0.075 & (-0.193,\;0.018) \\
  $\sigma^2(\tau_{r2})$ & 0.146 & (0.079,\;0.246) \\
  $\sigma(\tau_{r2},\,\tau_{r3})$ & -0.036 & (-0.114,\;0.020) \\
  $\sigma^2(\tau_{r3})$ & 0.134 & (0.069,\;0.229) \\
\hline
\end{tabular}
\caption{\small Expected A Posteriori (EAP) and 95\% credible intervals of the variances and covariances of the rater-specific parameters $\psi_r$ (DEG-MFR DP model with multivariate normal rater prior) from Odors Assessment data.}
\label{Table_Cor_Mat_MVN_Odor}
\end{table}

\begin{figure}
    \centering
    \subfigure{\includegraphics[scale= 0.77]{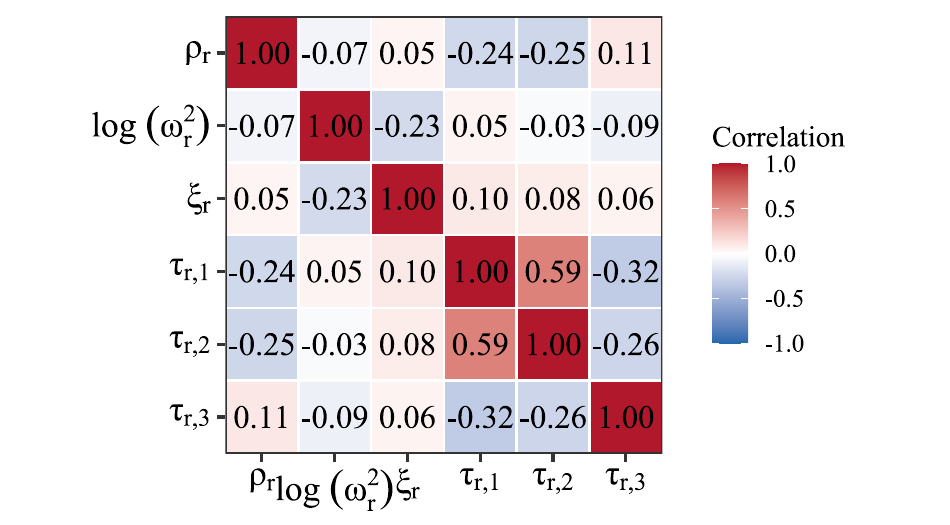}} 
    \caption{\small Posterior mean correlation matrix of the rater-specific parameter vector $(\log(\omega_r^2), \xi_r, \rho_r, \tau_{r1}, \tau_{r2}, \tau_{r3})$, derived from the estimated covariance matrix $\boldsymbol{\Sigma}$.}
    \label{Cor_Mat_MVN_Odor}
\end{figure}

\begin{figure}
    \centering
    \includegraphics[scale=0.75]{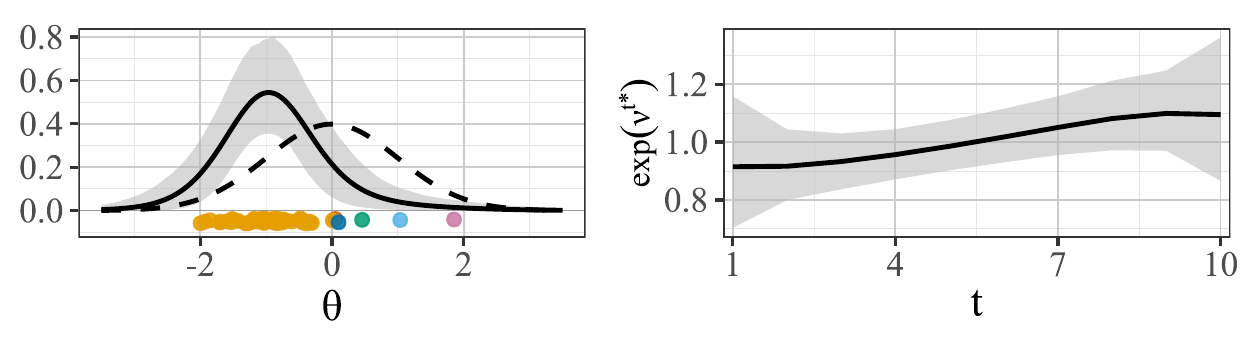}
    \caption{\small Left panel: graphical representation of the estimated latent trait density; the black solid line, the shaded grey and the black dashed line indicate, respectively, the  Expected A Posteriori and the 95\% quantile-based credible interval of the subject latent trait density, and the standard normal density. Right panel: graphical representation of the estimated rater practice effect $\exp{\nu^{t*}}$; the blue solid line and the shaded grey indicate, respectively, the  Expected A Posteriori and the 95\% quantile-based credible interval of the P-Spline. }
    \label{Density_P-Spline_MVN_Odor}
\end{figure}

\begin{figure}
    \centering
    \includegraphics[scale=0.77]{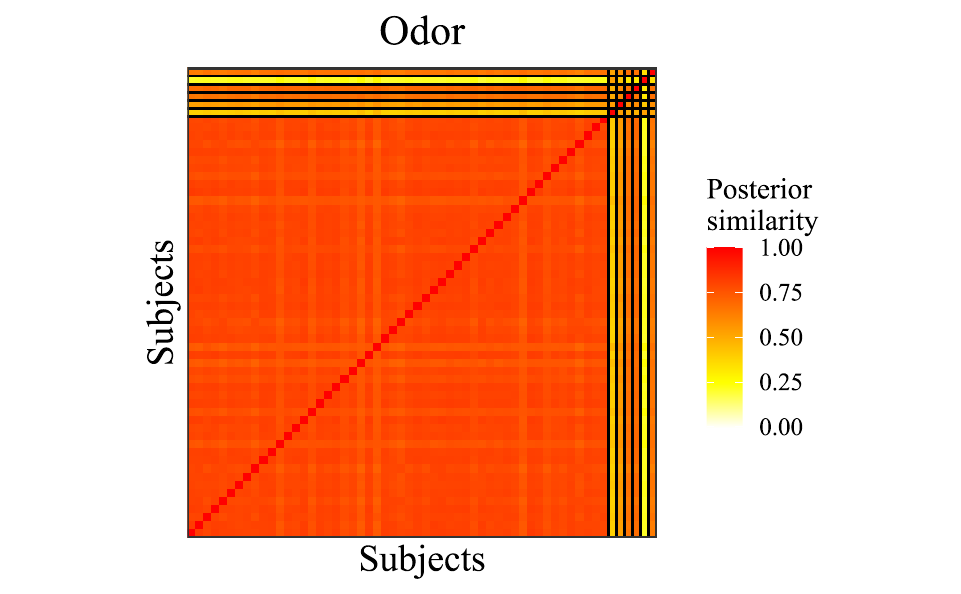}
    \caption{Posterior similarity matrix based on the MVN DEG model from the second Odor data set analyzed in which only \textit{pleasantness}, \textit{familiar}, and \textit{edible} are considered.  }
    \label{PSM_MVN}
\end{figure}

\subsection{Explanatory Response Extension}\label{sec:explanatory}
In this Section, we present an extension of the \gls{degmfr} model that accounts for the effect of item-, rater-, or person-properties or covariates on the probability of the responses. Since this instance of the model relates to the \textit{explanatory approach} introduced by \citeApp{deboeck2004explanatory}, we refer to this extension as the \textit{Explanatory \gls{degmfr} model}. This extension offers a natural modeling solution for performance assessments conducted on onscreen or internet-based platforms, where auxiliary information such as rating times is routinely recorded \citeApp{Eckes_2024}.

\paragraph{Item-Rater Explanatory Model.}
When extra information is available for the set of items, we can estimate its effect on the items' measurement error. We can generalize Eq.~(\ref{item-shift}) to account for the effects of this set of predictors on item-specific residual terms: 
\begin{equation}\label{item-shift_exp}
   \epsilon_{ij}|\varphi_i^2 \overset{\mathrm{iid}}{\sim} N\left(
   \sum_{s=1}^S w_sZ_{is}, \varphi_i^2 \right), \quad j \in \mathcal{J}, \;\; i=1,\dots, I.
\end{equation}
where $w_s$ is the effect of the $s$-th predictor $Z_{is}$ of item $i$. \\
We might similarly generalize Eq.~\eqref{Measurement:difficulty2} to model the effects of a set of rater-related predictors on rater's subject-local biases:
\begin{equation}\label{Measurement:difficulty2_exp}
    \eta_{rj}|\omega^2_r \overset{\mathrm{iid}}{\sim} N\left( \sum_{q=1}^Q u_qM_{rq},\omega^2_r \right), \quad j \in \mathcal{J}_r, \;\; r \in \mathcal{R}.
\end{equation}
where $u_q$ is the effect of the $q$-th predictor $M_{rq}$ of rater $r$. 
\paragraph{Person Explanatory Model.}
We might include subject-level covariates in the model to model the respective effects on the response probability. For each subject $j \in \mathcal{J}$, we let $\theta_j^{'}$ be the subject-related explanatory portion of the model:
\begin{equation}
    \theta_j^{'}= \theta_j + \sum_{p=1}^{P} \kappa_pX_{jp}, 
\end{equation}
where $\kappa_p$ is the effect of the $p$-th predictor $X_{jp}$ of the subject $r$. We assume the subject latent trait $\theta_j$ is distributed as detailed in Eq.~(\ref{Structural_1}). We note that, under this explanatory specification, $\theta_j^{'}$ replaces $\theta_j$ in the measurement model defined in Eq.~(\ref{eq:marg_D}). \\
Computational strategies for posterior sampling and identifiability constraints for the explanatory models presented in this Section are left for future investigation.

\bibliographystyleApp{apacite}
\makeatletter
\begingroup
\def\@extra@b@citeb{}%
\let\NAT@orig@lbibitem\@lbibitem
\def\@lbibitem[#1]#2{\NAT@orig@lbibitem[#1]{#2@App}}%
\makeatother
\bibliographyApp{reference}
\endgroup

\end{document}